\documentclass[sigconf]{acmart}
\usepackage{framed}
\usepackage{varwidth}
\usepackage{parskip}
\usepackage{amsmath,subcaption,booktabs,xcolor,amsthm,comment,graphics,comment,multirow,colortbl}
\usepackage[font=small,labelfont=bf]{caption}
\usepackage[listings,skins,breakable]{tcolorbox}
\usepackage{algorithmic}
\usepackage{csquotes}
\usepackage{array}
\usepackage[export]{adjustbox}
\usepackage{adjustbox}
\usepackage{nicefrac}
\usepackage{float}
\usepackage{enumitem}
\usepackage{caption,graphicx,newfloat}
\newcommand{\name}{\textsf{EIFFeL}}

\newcommand{\DP}{\mathcal{D}_P}
\newcommand{\Ser}{\mathcal{S}}
\newcommand{\Cl}{\mathcal{C}}
\newcommand{\CA}{\mathcal{C^*}}
\newcommand{\Ver}{\mathcal{V}}
\newcommand{\Pro}{\mathcal{P}}
\newcommand{\Ul}{\mathcal{U}}
\newcommand{\Bl}{\mathcal{B}}
\newcommand{\Vl}{\textsf{Valid}}
\newcommand{\Fl}{\mathbb{F}}
\newcommand{\Ml}{\textsf{M}}

\newtheorem{theorem}{Theorem}
\newtheorem{definition}{Definition}
\newcommand{\larrow}{\leftarrow}

\newcommand{\randarrow}{\xleftarrow{\$}}
 \newcommand{\SAIV}{\textsf{SAVI}}
 \newcommand\numberthis{\addtocounter{equation}{1}\tag{\theequation}}
\newtheorem{corollary}{Corollary}[theorem]
\newtheorem{lemma}[theorem]{Lemma}
 \newcommand{\squishlist}{
	\begin{list}{$\bullet$}
		{
			\setlength{\itemsep}{0pt}
			\setlength{\parsep}{0.5pt}
			\setlength{\topsep}{0.5pt}
			\setlength{\partopsep}{0pt}
			\setlength{\leftmargin}{1em}
			\setlength{\labelwidth}{1em}
			\setlength{\labelsep}{0.5em} } }
	
\newcommand{\squishend}{
\end{list}  }
  \newcommand{\squishlistdash}{
	\begin{list}{$-$}
		{
			\setlength{\itemsep}{0pt}
			\setlength{\parsep}{0.5pt}
			\setlength{\topsep}{0.5pt}
			\setlength{\partopsep}{0pt}
			\setlength{\leftmargin}{1em}
			\setlength{\labelwidth}{1em}
			\setlength{\labelsep}{0.5em} } }
	
\newcommand{\squishenddash}{
\end{list}  }
 
 \usepackage{graphicx}

\newcommand*{\img}[1]{%
    \raisebox{-.3\baselineskip}{%
        \includegraphics[
        height=\baselineskip,
        width=\baselineskip,
        keepaspectratio,
        ]{#1}%
    }%
}
 \newcommand*{\imag}[1]{%
    \raisebox{-.04\baselineskip}{%
        \includegraphics[
        height=0.8\baselineskip,
        width=\baselineskip,
        keepaspectratio,
        ]{#1}%
    }%
}

\usepackage{url}
\usepackage{titlesec}
\definecolor{aliceblue}{rgb}{0.94, 0.97, 1.0}
\copyrightyear{2022}
\acmYear{2022}
\setcopyright{acmcopyright}\acmConference[CCS '22]{Proceedings of the 2022 ACM SIGSAC Conference on Computer and Communications Security}{November 7--11, 2022}{Los Angeles, CA, USA}
\acmBooktitle{Proceedings of the 2022 ACM SIGSAC Conference on Computer and Communications Security (CCS '22), November 7--11, 2022, Los Angeles, CA, USA}
\acmPrice{15.00}
\acmDOI{10.1145/3548606.3560611}
\acmISBN{978-1-4503-9450-5/22/11}
\begin{document}
\title{\textsf{E}\hspace{-0.057cm}\scalebox{0.85}{\imag{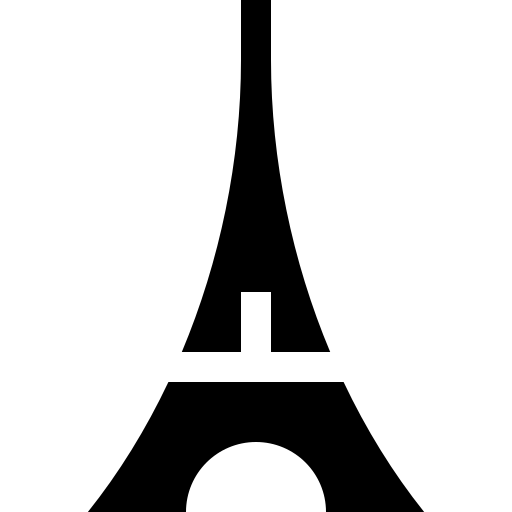}}\hspace{-0.057cm}\textsf{FFeL}: Ensuring Integrity For Federated Learning}


\settopmatter{authorsperrow=4}
\author{Amrita Roy Chowdhury}\authornote{Work done during internship at Meta AI}
\affiliation{UW Madison \country{}}
\author{Chuan Guo}
\affiliation{Meta AI \country{}}
\author{Somesh Jha}
\authornote{Employed part-time at Meta during this work}
\affiliation{UW Madison \country{}}
\author{Laurens van der Maaten}
\affiliation{Meta AI \country{}}
\renewcommand{\shortauthors}{Roy Chowdhury et al.}



\begin{abstract}
Federated learning (FL) enables clients to collaborate with a server to train a machine learning model. To ensure privacy, the server performs secure aggregation of updates from the clients. 
Unfortunately, this prevents verification of the well-formedness (integrity) of the updates  as the updates are masked. Consequently, malformed updates  designed to poison the model can be injected without detection. In this paper, we formalize the problem of ensuring \textit{both} update privacy and integrity in FL and present a new system, \name, that enables secure aggregation of \textit{verified} updates. \name~is a general framework that can  enforce \textit{arbitrary} integrity checks and remove  malformed updates from the aggregate, without violating privacy. Our empirical evaluation demonstrates the practicality of \name. For instance, with $100$ clients and $10\%$ poisoning, \name~can train an \textsf{MNIST} classification model to the same accuracy as that of a non-poisoned federated learner in just $2.4$s per iteration.
\end{abstract}


\maketitle
\section{Introduction}\label{sec:intro}
  \begin{figure}\centering
  \begin{minipage}[b]{1.0\linewidth}
    \centering
   
 \raisebox{\dimexpr -0.5\height}{\includegraphics[width=0.8\textwidth]{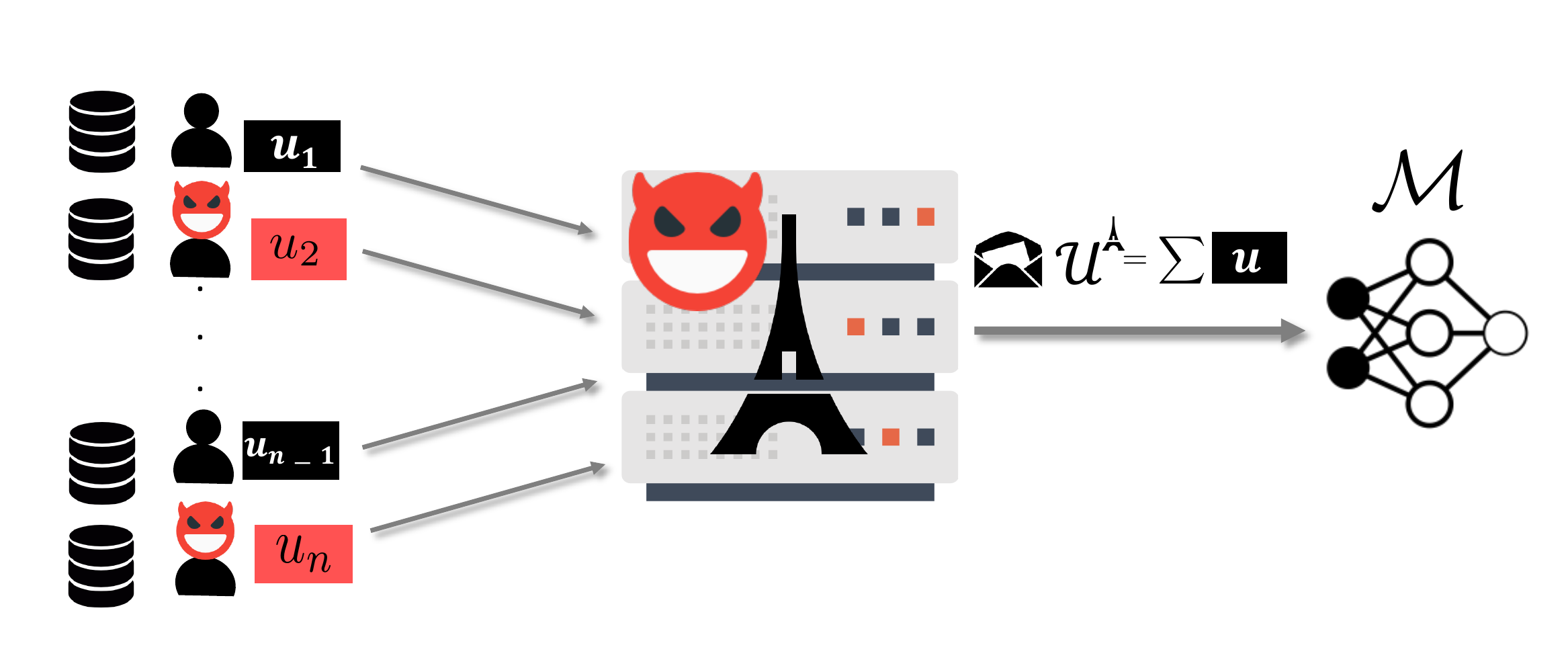}}
    \\
   \resizebox{0.8\columnwidth}{!}{\begin{tabular}{lc}
      \toprule
      \bf Security Goal & \bf Cryptographic Primitive\\\midrule
      Input Privacy & Shamir's Threshold  Secret Sharing Scheme~\cite{Shamir79}\\
      ~&~\\
         \multirow{2}{*}{Input Integrity} & Secret-Shared Non-Interactive Proof~\cite{corrigan2017prio} \\&Verifiable Secret Shares~\cite{Feldman87} \\
      \bottomrule
    \end{tabular}}
  \end{minipage}\vspace{-0.1cm}\caption{\name~performs secure aggregation of \emph{verified} inputs in FL. The table lists its security goals and the cryptographic primitives we adopt to achieve them.}\label{fig:name:overview}\vspace{-0.4cm}
  \end{figure}
Federated learning (FL; \cite{mcmahan2017communication}) is a learning paradigm for decentralized data in which multiple clients collaborate with a server to train a machine-learning (ML) model. 
Each client computes an update on its \textit{local} training data and shares it with the server; the server aggregates the local updates into a \textit{global} model update.
This allows the clients to contribute to model training without sharing their private data. However, the local updates can still reveal information about a client's private data \cite{Melis2019,Bhowmick2018ProtectionAR,Zhu2019DeepLF,Yin2021SeeTG,nasr2019}. FL addresses this by using \emph{secure aggregation}: clients mask the updates they share, and the server can recover \emph{only} the final aggregate in the clear.

A major challenge in FL is that it is vulnerable to Byzantine attacks. 
In particular, malicious clients can inject poisoned updates into the learner with the goal of reducing the global model accuracy \cite{biggio2021poisoning,mei2015teaching,Fang2020LocalMP,bhagoji2019analyzing,kairouz2019advances} or implanting backdoors in the model that can be exploited later \cite{chen2017targeted,bagdasaryan2018backdoor,Xie2020DBADB}.
Even a single malformed update can significantly alter the trained model~\cite{blanchard2017machine}. 
Thus, ensuring the well-formedness of the updates, \emph{i.e.}, upholding their \textit{integrity}, is essential for ensuring robustness in FL. This problem is especially challenging in the context of secure aggregation as the individual updates are masked from the server, which prevents audits on them. 

These challenges in FL  lead to the research question: \emph{How can a federated learner efficiently verify the integrity of clients' updates without violating their privacy?}

We formalize this problem by proposing \textit{secure aggregation of verified inputs} (\SAIV) protocols that: $(1)$ securely verify the integrity of each local  update, $(2)$ aggregate \textit{only} well-formed updates, and $(3)$ release only the final aggregate in the clear. 
 A \SAIV~protocol allows for checking the well-formedness of updates \textit{without observing them}, thereby ensuring \textit{both} the privacy and integrity of updates.  

We demonstrate the feasibility of \SAIV~ by proposing \name: a system that instantiates a \SAIV~protocol that can perform \textit{any integrity check that can be expressed as an arithmetic circuit with public parameters}.  This provides \name~the flexibility to implement a plethora of modern ML approaches that ensure robustness to Byzantine attacks by checking the integrity of per-client updates before aggregating them \cite{sun2019backdoor,steinhardt2017certified,xie2020zeno,xie2019zeno,Li20,Damaskinos2018AsynchronousBM,bagdasaryan2018backdoor,Shejwalkar2021ManipulatingTB}. \name~is a general framework that empowers a federated learner to deploy (multiple) \textit{arbitrary} integrity checks of their choosing on the ``masked'' updates.

\name~uses secret-shared non-interactive proofs (SNIP;~\cite{corrigan2017prio}) which are a type of zero-knowledge proofs that are optimized for the client-server setting.  
SNIP, however, requires multiple honest verifiers to check the proof. \name~extends SNIP to a \textit{malicious} threat model by carefully \textit{co-designing its architectural and  cryptographic components}. 
Moreover, we develop a suite of optimizations that improve \name's performance by at least $2.3\times$. 
Our empirical evaluation of
\name~demonstrates its practicality for real-world usage. For instance, with $100$ clients and a poisoning rate of $10\%$, \name~can train an \textsf{MNIST} classification model to the same accuracy as that of a non-poisoned federated learner in just $2.4s$ per iteration.

\section{Problem Overview}\label{sec:overview}
\begin{figure*}[ht]
    \begin{subfigure}{0.23\linewidth}
        \hspace{-0.1cm} \includegraphics[width=0.9\linewidth,left]{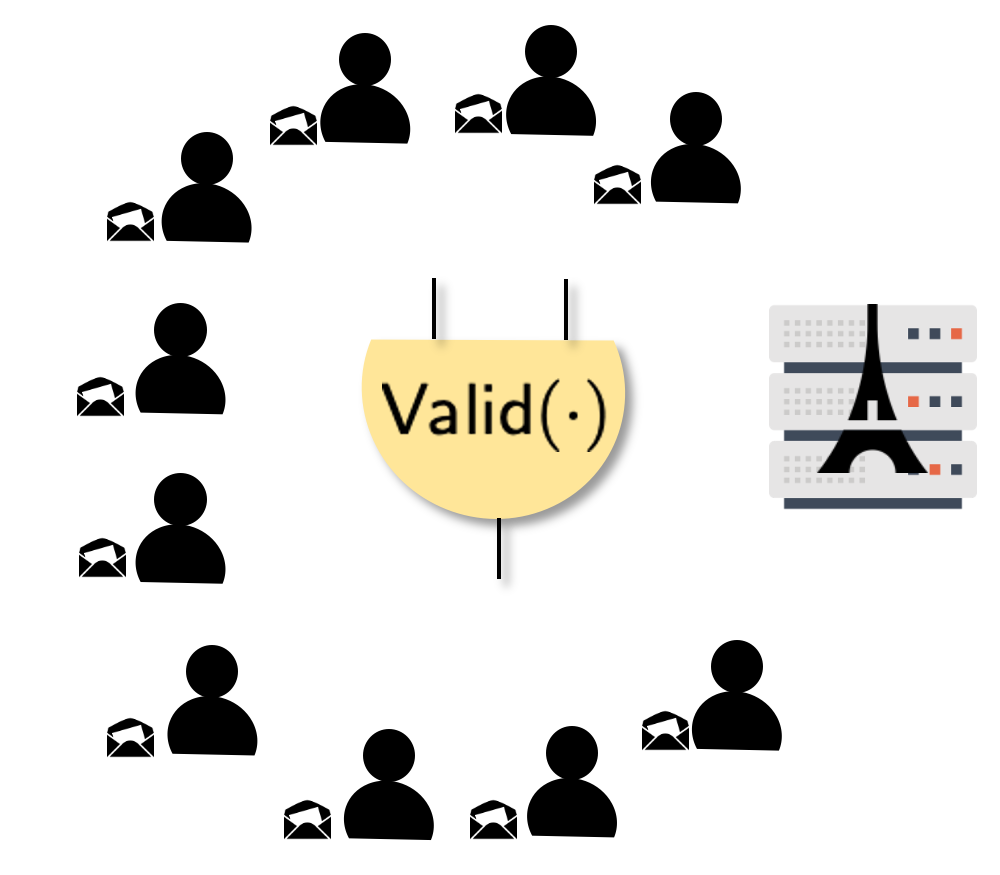}
       \caption{ \name~consists of multiple clients $\Cl$ and a server $\Ser$ with a public validation predicate $\Vl(\cdot)$ that defines the integrity check. A client $\Cl_i$ needs to provide a proof $\pi_i$ for $\Vl(u_i)=1$ (Round 1).}
        \label{fig:Eiffel1}
    \end{subfigure}
    ~~~~~
    \hspace{0.1cm}\begin{subfigure}{0.23\linewidth}
  \includegraphics[width=0.9\linewidth]{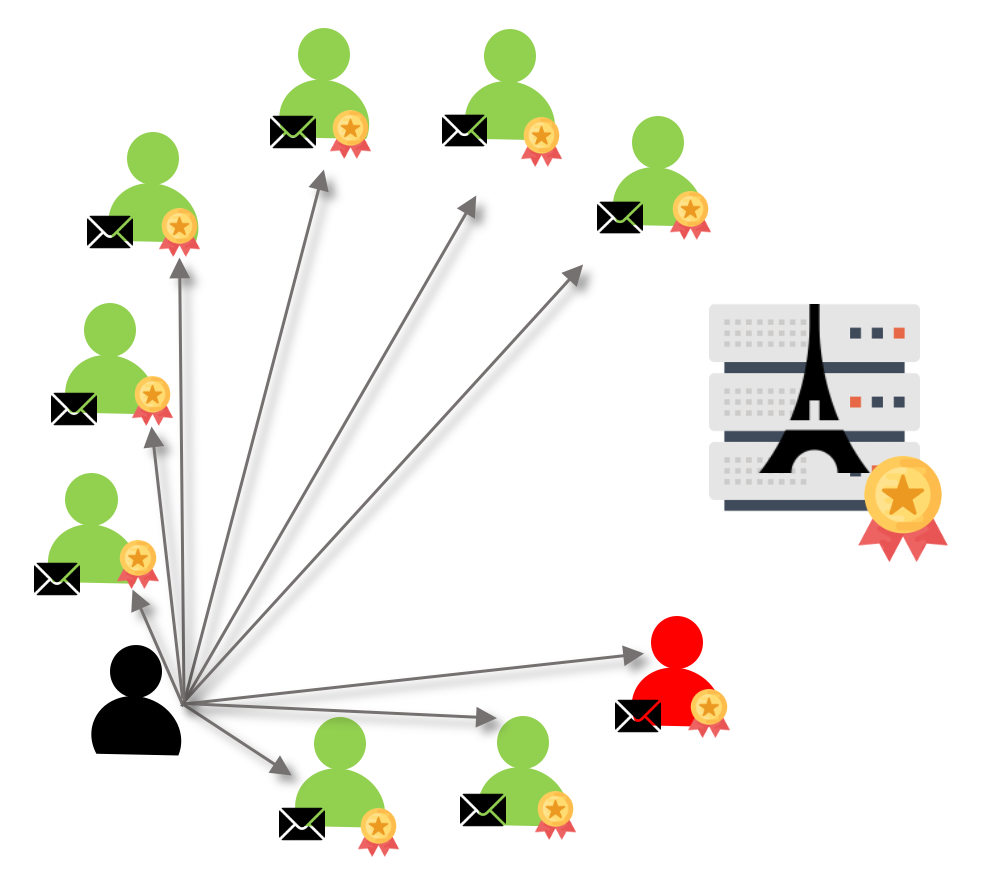}   
 \caption{For checking the proof $\pi_i$, all other clients $\Cl_{\setminus i}$ act as the verifiers under the supervision of $\Ser$. $C_i$ splits its update $u_i$ and proof $\pi_i$ using Shamir's  scheme with threshold $m+1$   and shares it with $\Cl_{\setminus i}$ (Round 2).}
        \label{fig:Eiffel2}\end{subfigure}
        ~~~~
         \hspace{0.1cm} \begin{subfigure}
         {0.23\linewidth}
    \includegraphics[width=0.9\linewidth]{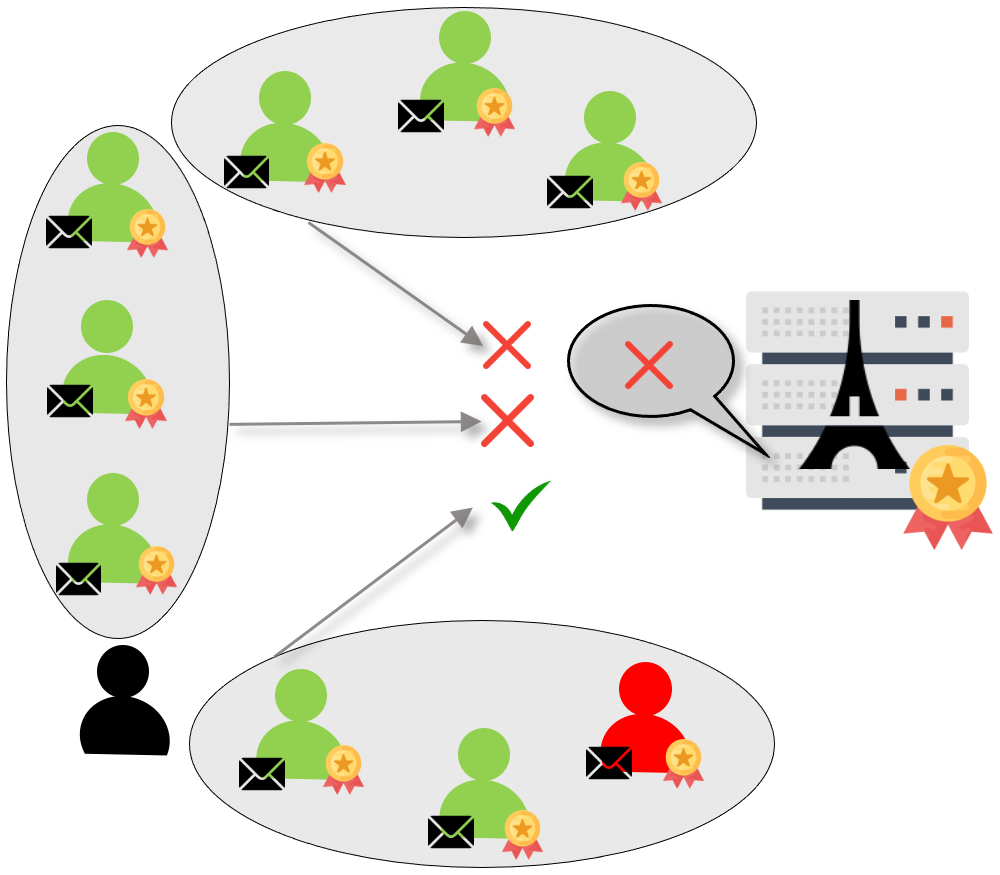}   
 \caption{Conceptually, any set of $m+1$ clients in $\Cl_{\setminus i}$ can emulate the SNIP verification protocol. The server uses this redundancy to \textit{robustly} verify the proof (Round 3).\\~}
        \label{fig:Eiffel3}\end{subfigure}
        ~~~~ \hspace{0.1cm} 
         \begin{subfigure}{0.23\linewidth}
    \raggedright\hspace{0.05cm} \includegraphics[width=0.9\linewidth,right]{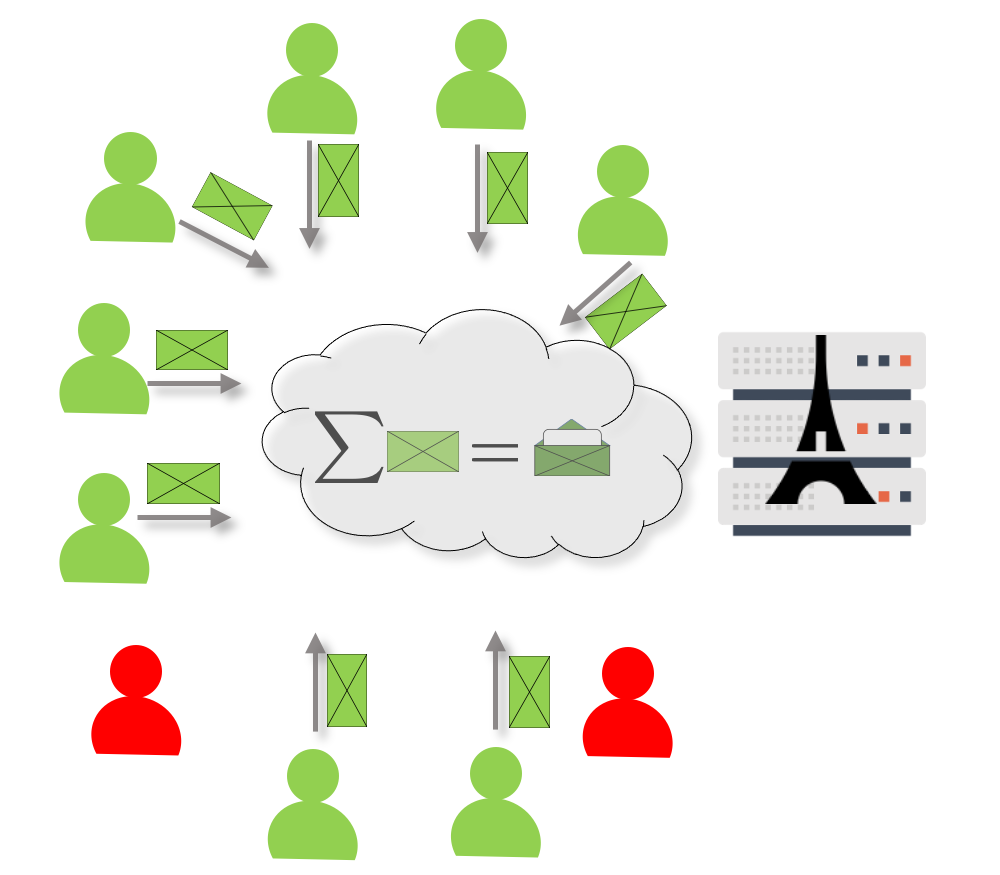}   
 \caption{ The clients only aggregate the shares of well-formed updates and the resulting aggregate is revealed to the server (Round 4).\\~\\}
        \label{fig:Eiffel4}\end{subfigure}
  \vspace{-0.3cm} \caption{High-level overview of \name. See Sec.~\ref{sec:solution_overview} for key ideas, and Sec.~\ref{sec:system:workflow} for a detailed description of the system.}
   \label{fig:name}
   \vspace{-0.3cm}
\end{figure*}
In this section, we introduce the problem setting, followed by its threat analysis and an overview of our solution. 
\vspace{-0.1cm}
\subsection{Problem Setting}\label{sec:overview:setting}
In FL, multiple parties with distributed data jointly train a \emph{global model}, $\mathcal{M}$,  without explicitly disclosing their data to each other. FL has two types of actors:
\vspace{-0.25cm}
\squishlist
\item \textbf{Clients.} There are $n$ clients where each client, $\Cl_i, i \in [n]$, owns a private dataset, $D_i$. The raw data is never shared, instead, every client computes a local update for $\mathcal{M}$, such as the average gradient, over the private dataset $D_i$.
\item \textbf{Server.} There is a single \textit{untrusted} server, $\Ser$, who coordinates the updates from different clients to train $\mathcal{M}$. 
 \squishend
 \vspace{-0.25cm}
A single training iteration in FL consists of the following steps:

\vspace{-0.2cm}
\squishlist
\item \textbf{Broadcast.} The server broadcasts the current parameters of the model $\mathcal{M}$ to all the clients. 
\item \textbf{Local computation.} Each client $\Cl_i$ locally computes an update, $u_i$, on its dataset $D_i$. 
\item \textbf{Aggregation.} The server $\Ser$ collects the client updates and aggregates them, $\Ul = \sum_{i\in [n]} u_i$.
\item \textbf{Global model update.} The server $\Ser$ updates the global model $\mathcal{M}$ based on the aggregated update $\Ul$. 
\squishend
\vspace{-0.2cm}
In settings where there is a large number of clients, it is typical to subsample a small subset of clients to participate in a given iteration. We assume $n$ to denote the number of clients that participate in each iteration and $\Cl$ denotes the subset of these $n$ clients, which the server announces at the beginning of the iteration.
\vspace{-0.2cm}

\subsection{Security Goals}\label{sec:overview:goals}
\squishlist\item \textbf{Input Privacy (Client's Goal).} The first goal is to ensure privacy for all \textit{honest} clients. That is, no party should be able learn anything about the raw input (update) $u_i$ of an honest client \scalebox{0.9}{$C_i$}, other than what can be learned from the final aggregate \scalebox{0.9}{$\Ul$}.  
\item \textbf{Input Integrity (Server's Goal).} The server $\Ser$ is motivated to ensure that the individual updates from each client are well-formed. Specifically, the server has a \textit{public} validation predicate, \scalebox{0.9}{$\textsf{Valid}(\cdot)$}, that defines a syntax for the inputs (updates). An input (update) $u$ is considered valid and, hence, passes the integrity check iff \scalebox{0.9}{$\textsf{Valid}(u)=1$}. For instance, any per-client update check, such as  Zeno++~\cite{xie2020zeno}, can be a good candidate for \scalebox{0.9}{$\Vl(\cdot)$} (we evaluate four state-of-the-art validation predicates in Sec. \ref{sec:eval:models}). 
\squishend
\vspace{-0.2cm}
We assume that the honest clients, denoted by \scalebox{0.9}{$\Cl_H$}: $(1)$ follow the protocol correctly, \textit{and} $(2)$ have well-formed inputs. We require the second condition because, in case the input of an honest client does not pass the integrity check (which can be verified locally since \scalebox{0.9}{$\Vl(\cdot)$} is public), the client has no incentive to participate in the training iteration.

\vspace{-0.1cm}
\subsection{Threat Model}
We consider a \textit{malicious adversary} threat model:
\squishlist\vspace{-0.2cm}
\item \textbf{Malicious Server.} We consider a malicious server that can deviate from the protocol arbitrarily with the aim of recovering the raw updates $u_i$ for $i \in [n]$ (see Remark 1 later for more details). 
\item \textbf{Malicious Clients.}   We also consider a set of $m$ malicious clients, $\Cl_M$. Malicious clients can arbitrarily deviate from the protocol with the goals of: (1) sending malformed inputs to the server and thus, compromising the final aggregate; (2) failing the integrity check of an honest client that submits well-formed updates;  (3) violating the privacy of an honest client, potentially in collusion with the server. 
\squishend

\subsection{Solution Overview}
\label{sec:solution_overview}
Prior work has mostly focused on ensuring input privacy via secure aggregation, \emph{i.e.}, securely computing the aggregate \scalebox{0.9}{$\Ul=\sum_{\Cl_i\in\Cl}u_i$}.
Motivated by the above problem setting and threat analysis, we introduce a new type of FL protocol, called \textit{secure aggregation with verified inputs} (\SAIV), that ensures \textit{both} input privacy and integrity. The goal of a \SAIV~protocol is to securely aggregate \textit{only} well-informed inputs.  \\In order to demonstrate the feasibility of \SAIV, we propose \name: a system that instantiates a \SAIV~protocol for any \scalebox{0.9}{$\Vl(\cdot)$} that can be expressed as an arithmetic circuit with public parameters (Fig. \ref{fig:name:overview}). \name~ensures input privacy by using Shamir's threshold secret sharing scheme~\cite{Shamir79} (Sec. \ref{sec:system:block}). 
Input integrity is guaranteed via SNIP and verifiable secret shares (VSS) which validates the correctness of the secret shares (Sec. \ref{sec:system:block}).  
The key ideas are:
\vspace{-0.2cm}
\squishlist \item SNIP requires multiple honest verifiers. \name~enables this in a single-server setting by having the clients act as the verifiers for each other under the supervision of the server (in Fig. \ref{fig:Eiffel2}, verifiers are marked by \img{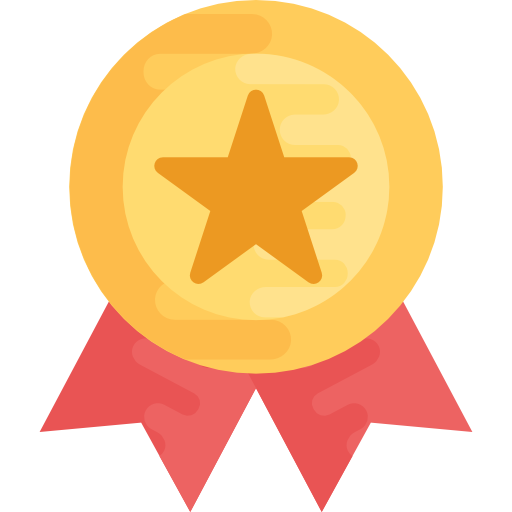}).
\item \name~extends SNIP to the malicious threat model to account for the malicious clients (verifiers). Our key observation is that using a threshold secret sharing scheme creates multiple subsets of clients that can emulate the SNIP verification protocol. The server uses this redundancy to robustly verify the proofs and aggregate updates with verified proofs \textit{only} (Fig. \ref{fig:Eiffel3} and \ref{fig:Eiffel4}).  \squishend

\section{Secure Aggregation with Verified Inputs} 
\label{sec:overview:savi}

Below, we provide the formal definition of a \emph{secure aggregation with verified inputs} (\SAIV) protocol.

    \begin{definition}Given a public validation predicate \scalebox{0.9}{$\Vl(\cdot)$} and security parameter \scalebox{0.9}{$\kappa$}, a protocol \scalebox{0.9}{$\Pi(u_1,\cdots,u_n)$} is a secure aggregation with verified inputs (\SAIV) protocol if: \vspace{-0.1cm}\squishlist \item \textbf{Integrity.} The output of the protocol, $\textsf{out}$, returns the aggregate of a subset of clients, \scalebox{0.9}{$\Cl_\Vl$}, such that all clients in \scalebox{0.9}{$\Cl_\Vl$} have  well-formed inputs. \vspace{-0.2cm} \begin{gather*} \mathrm{Pr}\big[\textsf{out}=\Ul_{\Vl}\big] \geq 1 -\mathrm{negl}(\kappa) \mbox{ where }\Ul_\Vl=\sum_{\Cl_i\in \Cl_\Vl}u_i \\ \mbox{for all } \Cl_i \in \Cl_\Vl \mbox{ we have } \Vl(u_i)=1 \\ \Cl_H\subseteq\Cl_\Vl\subseteq\Cl.   \numberthis\label{eq:SAVI:integrity}\vspace{-0.5cm}\end{gather*} \vspace{-0.5cm}

    \item \textbf{Privacy.} For a set of malicious clients \scalebox{0.9}{$\Cl_M$} and a malicious server \scalebox{0.9}{$\Ser$}, there exists a probabilistic polynomial-time (P.P.T.) simulator \scalebox{0.9}{$\mathrm{Sim}(\cdot)$} such that: \begin{gather*}\vspace{-0.3cm}\mathrm{Real}_{\Pi}\big(
    \{u_{\Cl_H}\}, \Omega_{\Cl_M\cup\Ser}\big)\equiv_C\mathrm{Sim}\big(\Ul_H,\Cl_H,\Omega_{\Cl_M\cup \Ser}\big)\\\mbox{where }\Ul_H=\sum_{\Cl_i\in \Cl_H}u_i.\numberthis\label{eq:SAVI:privacy}\vspace{-0.5cm}\end{gather*}  \scalebox{0.9}{$\{u_{\Cl_H}\}$} denotes the input of all the honest clients, \scalebox{0.9}{$\mathrm{Real}_\Pi$} denotes a random variable representing the joint view of all the parties in $\Pi$'s execution,  \scalebox{0.9}{$\Omega_{\Cl_M\cup\Ser}$} indicates a polynomial-time algorithm implementing
the “next-message” function of the parties in \scalebox{0.9}{$\Cl_M\cup\Ser$} (see App. \ref{app:security}), 
and \scalebox{0.9}{$\equiv_C$} denotes computational indistinguishability. \squishend \label{def:SAVI}\end{definition}
From Def. \ref{def:SAVI}, the output of a \SAIV~protocol is of the form:\vspace{-0.1cm} \begin{gather}\Ul_{valid}=\underbrace{\Ul_H}_{\substack{\text{well-formed updates of}\\\textit{all }\text{honest clients }\Cl_H}}+\underbrace{\sum_{\Cl_i\in \Cl_\Vl\setminus\Cl_H}u_i.}_{\substack{\text{well-formed updates of}\\\text{some malicious clients}}}\label{eq:SAVI:U}\end{gather} 
The clients in \scalebox{0.9}{$\Cl_\Vl\setminus \Cl_H$} are clients who have submitted well-formed inputs but can behave maliciously otherwise (\emph{e.g.}, by violating input privacy/integrity of honest clients).

The privacy constraint of the \SAIV~protocol means that a simulator \textsf{Sim} can generate the views of all parties with just access to the list of the honest clients \scalebox{0.9}{$\Cl_H$} and their aggregate \scalebox{0.9}{$\Ul_H$}. Note that \textsf{Sim} takes \scalebox{0.9}{$\Ul_H$} as an  input instead of the protocol output \scalebox{0.9}{$\Ul_\Vl$}. This is because the clients in \scalebox{0.9}{$\Cl_\Vl\setminus\Cl_H$}, by virtue of being malicious, can behave arbitrarily and announce their updates to reveal \scalebox{0.9}{$\Ul_H=\Ul_\Vl-\sum_{\Cl_i\in \Cl_\Vl\setminus\Cl_H} u_i$}. 
Thus, \SAIV~ensures that nothing can be learned about the input $u_i$ of an honest client \scalebox{0.9}{$\Cl_i \in \Cl_H$} except:
    \squishlist
        \item that $u_i$ is well-formed, \emph{i.e.}, \scalebox{0.9}{$\Vl(u_i)=1$},
        \item anything that can be learned from the aggregate $\Ul_H$.
    \squishend
    
 \begin{tcolorbox}[sharp corners, breakable]\vspace{-0.2cm}\textbf{Remark 1.} Note that we consider a malicious server only for input privacy and the reason is as follows. For input integrity, a malicious server can do the following:
\squishlist \item \textbf{Case 1.}  Mark the input of an honest client as invalid and not include it in the final aggregate.

\item \textbf{Case 2.} Mark the (invalid) input of a malicious client as valid.
\squishend
\name~prevents Case 1 from happening since, from Def. \ref{def:SAVI}, \name~is guaranteed to output the aggregate of \textit{all} honest clients (Lemma \ref{lemma:3}). If we consider a malicious server even for data integrity, the only thing that can happen is Case 2. However, the server’s primary goal is to ensure that each input is well-formed. Hence, Case 2, i.e., marking the (invalid) input of a malicious client as valid, is at odds with the server’s goal.  Therefore, it is unnecessary to protect against this behavior in our setting and we consider the server to be honest for the purposes of input integrity. 

\end{tcolorbox}
    \begin{tcolorbox}[sharp corners]\vspace{-0.2cm} \textbf{Remark 2.} The integrity constraint of \SAIV~requires the protocol to detect \textit{and} remove \textit{all} malformed inputs before computing the final aggregate. Note that there is a fundamental difference between the design choice of just detection of a malformed input versus detection \textit{and} removal. 
In the former, the server can only abort the current round even when just a \textit{single} malformed input is detected. 
This allows an adversary to stage a denial-of-service attack that renders the server incapable of training the model.  
When the protocol can both detect and remove malformed inputs, such denial-of-service attacks are prevented as the server can train the model using just the valid updates.
\vspace{-0.2cm}\end{tcolorbox}
    
\vspace{-0.2cm}\section{\textsf{EIFF}\texorpdfstring{\MakeLowercase{e}}{e}\textsf{L}~System Description}\label{sec:system}
This section introduces \name: the system we propose to perform secure aggregation of verified inputs.  
\vspace{-0.3cm}\subsection{Cryptographic Building Blocks}\label{sec:system:block}
\textbf{Arithmetic Circuit.} An arithmetic circuit, \scalebox{0.9}{$\mathcal{C}: \Fl^k\mapsto \Fl$}, represents a computation over a finite field \scalebox{0.9}{$\Fl$}.
Conceptually, it is similar to a Boolean circuit but it uses finite field addition, multiplication and multiplication-by-constant instead of \textsf{OR}, \textsf{AND}, and \textsf{NOT} gates. 

\noindent\textbf{Shamir's $t$-out-of-$n$ Secret Sharing Scheme \cite{Shamir79}} allows distributing a secret $s$ among $n$ parties such that: (1) the complete secret can be reconstructed from any combination of $t$ shares; (2) any set of 
$t - 1$ or fewer shares reveals no information about $s$ where $t$ is the \textit{threshold} of the secret sharing scheme. The scheme is parameterized over a finite field $\Fl$ and consists of  two algorithms:
\squishlist\item \scalebox{0.9}{$\{(i,s_i)\}_{i\in P}\randarrow \textsf{SS.share}(s,P,t)$}. Given a secret $s  \in \Fl$, a set of $n$ unique field elements \scalebox{0.9}{$P \in\Fl^n$} and a threshold $t$ with $t \leq n$, this algorithm constructs $n$ shares. The algorithm chooses a random polynomial \scalebox{0.9}{$p\in \Fl[X]$} such that \scalebox{0.9}{$p(0)=s$} and generates the shares as \scalebox{0.9}{$(i,p(i)), i \in P$}.

\item \scalebox{0.9}{$s\leftarrow \textsf{SS.recon}(\{(i,s_i)_{i\in Q}\})$}. Given the  shares corresponding to a subset \scalebox{0.9}{$Q\subseteq P, |Q|\geq t$}, the reconstruction algorithm recovers the secret $s$.  
\squishend
\vspace{-0.1cm}
Shamir's secret sharing scheme is \emph{linear}, which means a party can \textit{locally} perform: $(1)$ addition of two shares, $(2)$ addition of a constant, and $(3)$ multiplication by a constant. 
\\Shamir's secret sharing scheme is closely related to Reed-Solomon error correcting codes \cite{lin2004error}, which is a group of polynomial-based error correcting codes. The share generation is similar to (non-systemic) message encoding in these codes which can successfully recover a message even in the presence of errors and erasures (message dropouts). Consequently, we can leverage Reed-Solomon decoding for robust reconstruction of Shamir's secret shares:
\squishlist \vspace{-0.2cm}\item \scalebox{0.9}{$s\leftarrow\textsf{SS.robustRecon}(\{(i,s_i)\}_{i\in Q})$}. Shamir's secret sharing scheme results in a \scalebox{0.9}{$[n,t,n-t+1]$} Reed-Solomon code that can tolerate up to \scalebox{0.9}{$q$} errors and \scalebox{0.9}{$e$} erasures (message dropouts) such that \scalebox{0.9}{$2q+e< n-t+1$}. Given any subset of $n\!-\!e$ shares \scalebox{0.9}{$Q \subseteq P, |Q|\geq n-e$} with up to \scalebox{0.9}{$q$} errors, any standard Reed Solomon decoding algorithm~\cite{Blahut1983} can robustly reconstruct \scalebox{0.9}{$s$}. \name~uses Gao's decoding algorithm~\cite{Gao2003}.
\squishend
\textit{Verifiable secret sharing scheme} is a related concept where the scheme has an additional property of \textit{verifiability}. Given a share of the secret, a party must be able to check whether it is indeed a valid share. If a share is valid, then there exists a unique secret which will be the output of the reconstruction algorithm when run on any $t$ distinct valid shares. Formally:\vspace{-0.1cm}
\squishlist\item \scalebox{0.9}{$1/0\leftarrow\textsf{SS.verify}((i,v),\Psi))$}.  The verification algorithm inputs a share and a check string $\Psi_s$ such that \begin{gather*}\vspace{-0.2cm}
\small\forall \thinspace V \subset \Fl\times\Fl \mbox{ where } |V|=t, \exists s  \in \Fl \mbox{ s.t. }\\ \hspace{-0.4cm}\small (\forall (i,v) \in V,\textsf{SS.verify}((i,v),\Psi_s)=1)  \implies \textsf{SS.recon}(V)=s \vspace{-0.3cm}\end{gather*} The share construction algorithm is augmented to output the check string as \scalebox{0.9}{$(\{(i,s_i)_{i \in P}\},\Psi_s)\larrow \textsf{SS.share}(s,P,t)$}. \squishend \vspace{-0.2cm}
For \name, we use the non-interactive verification scheme by Feldman \cite{Feldman87} (details in App. \ref{app:background}).

\textbf{Key Agreement Protocol.} 
    A key agreement protocol  consists of a tuple of the following three algorithms:
\squishlist \vspace{-0.1cm}
\item \scalebox{0.9}{$(pp)\randarrow \textsf{KA.param}(1^\kappa)$}. The parameter generation algorithm samples a set of public parameters $pp$ with security parameter $\kappa$.
\item\scalebox{0.9}{$(pk,sk)\randarrow\textsf{KA.gen}(pp)$}. The key generation algorithm samples a public/secret key pair from the public parameters. 
\item \scalebox{0.9}{$sk_{ij}\leftarrow \textsf{KA.agree}(pk_i,sk_j)$}. The key agreement protocol receives a public key $pk_i$ and a secret key $sk_j$ as input and generates  the shared key $sk_{ij}$. 
\squishend

\noindent\textbf{Authenticated Encryption} provides confidentiality and integrity guarantees for messages exchanged between
two parties. It consists of a tuple of three algorithms as follows: \squishlist\vspace{-0.2cm}\item \scalebox{0.9}{$k\randarrow \textsf{AE.gen}(1^\kappa)$}.  The key generation algorithm that
outputs a private key   $k$ where $\kappa$ is the security parameter.
\item \scalebox{0.9}{$\overline{x}\randarrow \textsf{AE.enc}(k,x)$}. The encryption algorithm takes as input a key $k$ and a message $x$, and outputs a ciphertext $\overline{x}$.
\item \scalebox{0.9}{$x\leftarrow\textsf{AE.dec}(k,\overline{x})$}. The decryption algorithm takes as
input a ciphertext and a key and outputs either the original plaintext, or a special error symbol $\bot$ on failure.\vspace{-0.2cm} \squishend

\noindent\textbf{Secret-shared Non-interactive Proofs.}
\begin{figure*}
    \begin{subfigure}{0.28\linewidth}
        \centering
         \includegraphics[width=.7\linewidth, height=2.3cm]{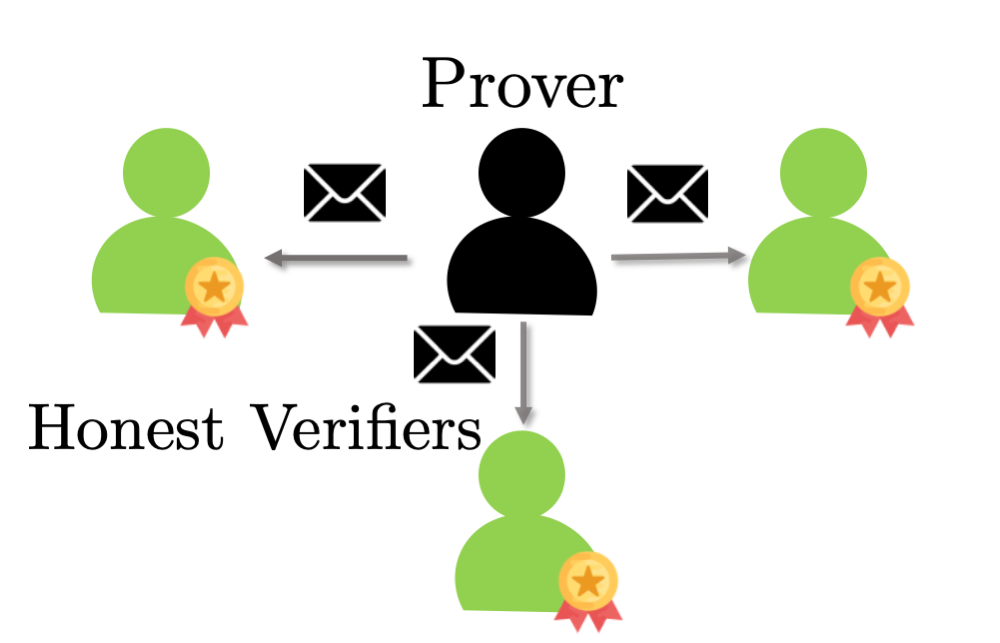}
      \vspace{-0.2cm}\caption{  Prover sends secret shares of its input  and the SNIP proof to multiple verifiers.}
        \label{fig:Snip1}
    \end{subfigure}
    \hspace{1cm}\begin{subfigure}{0.28\linewidth}\centering
  \includegraphics[width=0.7\linewidth,height=2.3cm]{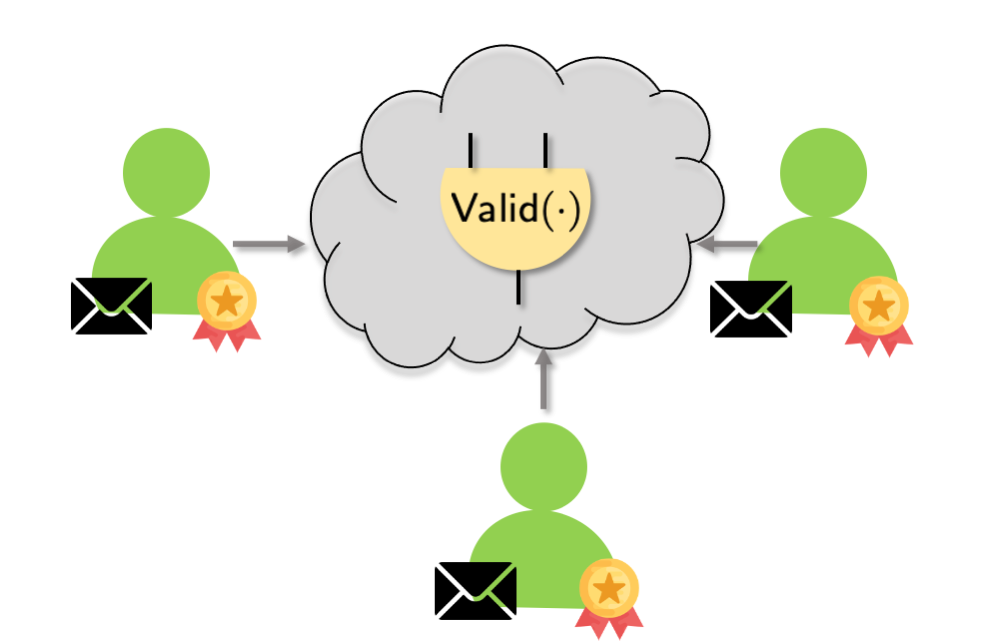}   
  \vspace{-0.2cm}\caption{The verifiers gossip among themselves and check the proof.}
        \label{fig:Snip2}\end{subfigure}
        \hspace{1cm}
         \begin{subfigure}{0.28\linewidth}\centering
 \includegraphics[width=0.7\linewidth,height=2.3cm]{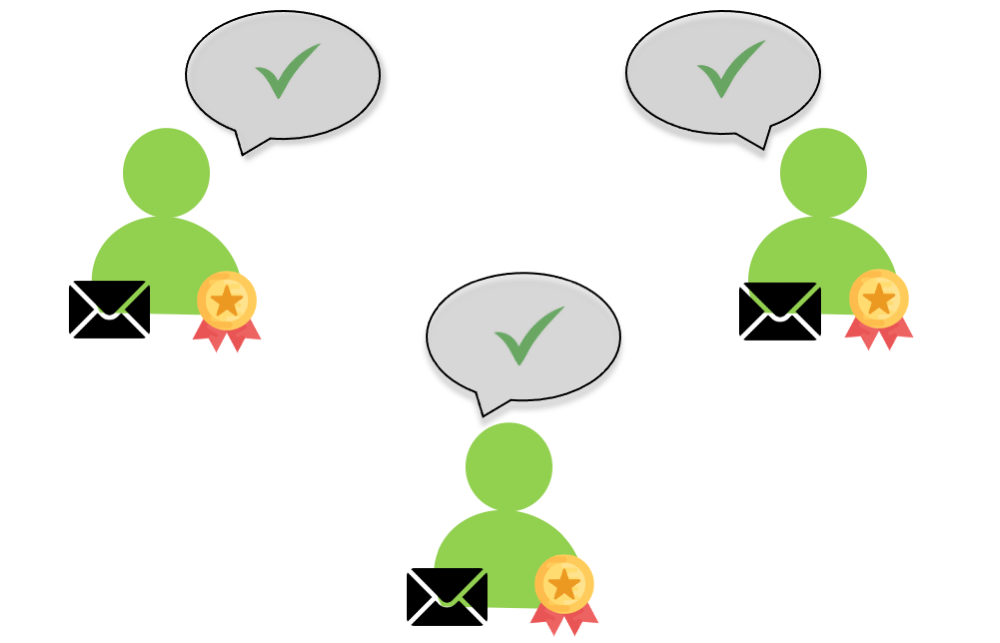}   
 \vspace{-0.2cm} \caption{The check passes successfully if all verifiers are honest.}
        \label{fig:Snip3}\end{subfigure}
  \vspace{-0.37cm} \caption{High-level overview of a secret-shared non-interactive proof (SNIP;~\cite{corrigan2017prio}).} 
   \label{accuracy}\vspace{-0.5cm}\label{fig:SNIP}
\end{figure*}
The secret-shared non-interactive proof (SNIP) \cite{corrigan2017prio} is an information-theoretic zero-knowl- edge proof for distributed data (Fig. \ref{fig:SNIP}). SNIP is designed for a multi-verifier setting where the private data is distributed or secret-shared among the verifiers. Specifically, SNIP relies on an additive secret sharing scheme over a field $\Fl$ as described below.  A secret $s \in \Fl$ is split into $k$ random shares $([s]_1,\cdots,[s]_k)$ such that $\sum_{i=1}^k[s]_i=s$. A subset of up to $k-1$ shares reveals \textit{no} information about the secret $s$. The additive secret-sharing scheme is linear as well.

\textit{SNIP Setting.} SNIP considers \scalebox{0.9}{$k\geq 2$} verifiers \scalebox{0.9}{$\{\Ver_i\}, i \in [k]$} and a prover \scalebox{0.9}{$\Pro$} with a private vector \scalebox{0.9}{$x\in\Fl^d$}. All parties also hold a \textit{public} arithmetic circuit representing a
validation predicate \scalebox{0.9}{$\Vl: \Fl^d\mapsto \Fl$
}. Let \scalebox{0.9}{$\Ml$} be the number of multiplication gates in \scalebox{0.9}{$\Vl(\cdot)$}.  \scalebox{0.9}{$\Fl$} is chosen such that \scalebox{0.9}{$2\Ml \ll |\Fl|$}. The prover \scalebox{0.9}{$\Pro$} splits \scalebox{0.9}{$x$} into \scalebox{0.9}{$k$} shares \scalebox{0.9}{$\{[x_1],\cdots, [x_k]\}$}. Next, they generate $k$ proof strings \scalebox{0.9}{$[\pi]_i, i \in [k]$} based on \scalebox{0.9}{$\Vl(\cdot)$} and shares \scalebox{0.9}{$( [x_i], [\pi]_i)$} with every verifier \scalebox{0.9}{$\Ver_i$} (Fig. \ref{fig:Snip1}).
\\The prover's goal is to convince the verifiers that, indeed, \scalebox{0.9}{$\Vl(x) =
1$}. The prover does so via proof strings \scalebox{0.9}{$[\pi]_i, i\in [k]$}, that do not reveal anything else about $x$. After
receiving the proof, the verifiers gossip with each other to conclude either that \scalebox{0.9}{$\Vl(x) = 1$}
(the verifiers “\scalebox{0.9}{$\textsf{Accept } x$}") or not (“\scalebox{0.9}{$\textsf{Reject } x$}”, Figs. \ref{fig:Snip2} and \ref{fig:Snip3}). Formally, SNIP satisfies the following security properties:\squishlist\item \textit{Completeness.} If all parties are honest and \scalebox{0.9}{$\Vl(x)=1$}, then the verifiers will
accept $x$. \begin{gather*}\vspace{-0.4cm}\small\forall x \in \Fl \mbox{ s.t. } \Vl(x)=1: \mathrm{ Pr}_{\pi}[\textsf{Accept } x]=1. \label{eq:SNIP:complete}\vspace{-0.5cm}\end{gather*}
\item \textit{Soundness}. If all verifiers are honest, and if \scalebox{0.9}{$\Vl(x)=0$},
then for all malicious provers, the verifiers will reject $x$ with
overwhelming probability. 
\begin{gather*}\vspace{-0.4cm}\small\forall x \in \Fl \mbox{ s.t. } \Vl(x)=0:  \thinspace\thinspace
\mathrm{ Pr}_{\pi
}\big[\textsf{Reject } x\big] \geq 1- \nicefrac{(2M-2)}{|\Fl|}.\label{eq:SNIP:sound}\vspace{-0.4cm}\end{gather*}\vspace{-0.4cm}
\item \textit{Zero knowledge.} If the prover and at least one verifier are
honest, then the verifiers learn nothing about $x$, except
that \scalebox{0.9}{$\Vl(x) = 1$}. Formally, when \scalebox{0.9}{$\Vl(x)=1$}, there exists a simulator \scalebox{0.9}{$\mathrm{Sim}(\cdot)$} that can simulate the view of
the protocol execution for every proper subset of verifiers:\vspace{-0.3cm}
\begin{gather*}\vspace{-0.3cm}\small\forall x \mbox{ s.t. } \Vl(x)=1\mbox{ and } \forall \bar{\Ver}\subset \bigcup_{i=1}^k\Ver_i \mbox{ we have}\\ \vspace{-0.5cm}\small  \mathrm{ Sim}_\pi(\Vl(\cdot),\{([x]_i,[\pi]_i)\}_{i\in \bar{\Ver}})\equiv \mathrm{View}_{\pi,\bar{\Ver}}(\Vl(\cdot),x).
\vspace{-0.2cm}
\end{gather*}
\squishend 
Thus, SNIP allows the verifiers  to collaboratively check -- without ever accessing the prover’s private data in the clear -- that the prover's
submission is, indeed, well-formed. 
SNIP works in two stages as follows: \\
$(1)$\textit{ Generation of Proof.} For generating the proof, the prover $\Pro$ first evaluates the 
circuit \scalebox{0.9}{$\Vl(\cdot)$} on its input $x$ to obtain the value of every wire in the arithmetic circuit corresponding to the computation of \scalebox{0.9}{$\Vl(x)$}. 
Using these wire values, $\Pro$  constructs 
three polynomials \scalebox{0.9}{$f$}, \scalebox{0.9}{$g$}, and \scalebox{0.9}{$h$} of the lowest possible degrees such that \scalebox{0.9}{$h=f\cdot g$} and \scalebox{0.9}{$f(j),g(j)$} and \scalebox{0.9}{$h(j), j \in \big[\Ml\big]$} encode the values of the two input wires and one output wire of the $j$-th multiplication gate, respectively. 
$\Pro$ also samples a single set of Beaver's multiplication triples \cite{Beaver91}: \scalebox{0.9}{$(a,b,c)\in \Fl^3$} such that \scalebox{0.9}{$a\cdot b = c \in \Fl$}. Finally, it generates the shares of the proof, \scalebox{0.9}{$[\pi]_i =\big( [h]_i, ([a]_i,[b]_i,[c]_i)\big)$}, which consists of:\vspace{-0.2cm}
 \squishlist\item shares of the coefficients of the polynomial $h$, denoted by $[h]_i$, \item shares of the  Beaver's triples, \scalebox{0.9}{$([a]_i,[b]_i,[c]_i)\in \Fl^3$}. \squishend The prover then sends the respective shares of the input and the proof \scalebox{0.9}{ $( [x]_i, [\pi]_{i})$} to each of the verifiers $\Ver_i$.
 
 \noindent$(2)$\textit{ Verification of Proof}. To verify that \scalebox{0.9}{$\Vl(x)=1$} and hence, accept the input $x$, the verifiers need to check two things:  \squishlist\vspace{-0.2cm} \item check that the value of final output wire of the computation, \scalebox{0.9}{$\Vl(x)$}, denoted by $w^{out}$ is indeed $1$, and \item check the consistency of $\Pro$'s computation of \scalebox{0.9}{$\Vl(x)$}. \vspace{-0.2cm}\squishend 
 
To this end, each verifier \scalebox{0.9}{$\Ver_i$} \textit{locally} constructs the shares of every wire in \scalebox{0.9}{$\Vl(x)$} 
 via affine operations
on the shares of the private input \scalebox{0.9}{$[x]_i$} and \scalebox{0.9}{$[h]_i$}. Next, \scalebox{0.9}{$\Ver_i$} broadcasts a summary \scalebox{0.9}{$[\sigma]_i=([w^{out}]_i,[\lambda]_i)$}, where \scalebox{0.9}{$[w^{out}]_i$} is \scalebox{0.9}{$\Ver_i$}'s share of the output wire of the circuit and \scalebox{0.9}{$[\lambda]_i$} is a share of a random digest that the verifier computes from the shares of the other wire values and the proof  string $[\pi]_i$. Using these  summaries, the verifiers check the proof as follows:
\squishlist\vspace{-0.2cm}\item For checking the output wire, the verifiers can reconstruct its exact value from all the broadcasted shares \scalebox{0.9}{$w^{out}=\sum_{i=1}^k [w^{out}]_i$} and check whether \scalebox{0.9}{$w^{out}=1$}. This would imply that  \scalebox{0.9}{$\Vl(x)=1$}. \item The circuit consistency check is more involved and is performed using the random digest \scalebox{0.9}{$\lambda$}. 
First, \scalebox{0.9}{$\Ver_i$} \textit{locally} computes the shares of the polynomials $f$ and $g$ (denoted as \scalebox{0.9}{$[f]_i$} and \scalebox{0.9}{$[g]_i$}). 
To verify the  consistency of the circuit evaluation, the verifiers need to check that the shares  \scalebox{0.9}{$[h]_i$} sent by the prover $\Pro$ are of the correct polynomial, \emph{i.e.}, confirm that \scalebox{0.9}{$f\cdot g =h$}. For this, SNIP uses the  Schwartz-Zippel polynomial identity test \cite{Schwartz80, Zippel79}. Specifically, verifiers reconstruct \scalebox{0.9}{$\lambda = \sum_{1=1}^k[\lambda]_i$} from the broadcasted shares and test whether \scalebox{0.9}{$\lambda=r(f(r) \cdot g(r)-h(r))=0$} on a randomly selected $r\in \Fl$.  The computation of the share of the random digest \scalebox{0.9}{$[\lambda]_i$}  
uses the shares of Beaver's triples \scalebox{0.9}{$([a]_i,[b]_i,[c]_i)$}.  
\squishend \vspace{-0.2cm}
A more detailed description of the SNIP protocol is in App. \ref{app:background}.
\vspace{-0.2cm}\subsection{System Building Blocks}\vspace{-0.05cm}

\noindent\textbf{Public Validation Predicate.} \name~requires a public validation predicate \scalebox{0.9}{$\Vl(\cdot)$}, expressed by an arithmetic circuit, that captures the notion of update well-formedness. In principle, any per-client update robustness test \cite{sun2019backdoor,steinhardt2017certified,xie2020zeno,Li20,Damaskinos2018AsynchronousBM,bagdasaryan2018backdoor,Shejwalkar2021ManipulatingTB} from the ML literature can be a suitable candidate. The parameters of the test (for instance, threshold \scalebox{0.9}{$\rho$} for a norm bound check \scalebox{0.9}{$\Vl(u) = \mathbb{I}[ \lVert u\rVert_2 < \rho$}) can be computed from a clean, public dataset $\DP$ that is available to the server $\Ser$. This assumption of a clean, public dataset is common in both ML \cite{xie2020zeno,Cao2021FLTrustBF,kairouz2019advances} as well as privacy literature \cite{liu2021leveraging,Bassily2020PrivateQR,Beimel2020ThePO}. The dataset can be small and obtained by manual labeling \cite{Google17}.
\\
\noindent\textbf{Public Bulletin Board.} \name~assumes the availability of a public bulletin board $\Bl$ that is accessible to all the parties, similar to prior work \cite{Roth2019Honeycrisp, bonawitz2017practical, kairouz2019advances}. In practice, the bulletin $\Bl$ can be implemented as an append-only log hosted at a public web address where every message and its sender is visible.  Every party in \name~has read/write access to it. We use the bulletin $\Bl$ as a tool for broadcasting~\cite{Bracha1985,Crosby2009Efficient}.
\vspace{-0.6cm}\subsection{\name~Design Goals}\label{sec:goals} \vspace{-0cm}
In terms of the design, \name~should:
\squishlist \vspace{-0.2cm}
\item provide  \textit{flexibility in the choice of integrity checks.} 
\item be \textit{compatible with the existing FL infrastructure in deployment.}
\item be \textit{efficient} in performance. \squishend\vspace{-0.2cm}
\vspace{-0.1cm}
\subsection{\name~Workflow}\label{sec:system:workflow}
The  goal of \name~is to instantiate a secure aggregation with verified inputs (\SAIV) protocol in FL. For a given public validation predicate $\Vl(\cdot)$, \name~checks the integrity of every client update using SNIP and outputs the aggregate of \textit{only} well-formed updates, \emph{i.e.}, \scalebox{0.9}{$\Vl(u)=1$}. 
To implement SNIP for our setting, \name~introduces two main ideas:
\begin{tcolorbox}[sharp corners, breakable,left=8pt,right=8pt,bottom=0.8pt]\vspace{-0.2cm}
\textbf{Main Ideas.}
\squishlist
\item In \name, the clients act as the verifiers for each other. Specifically, for every client \scalebox{0.9}{$\Cl_i, i \in \Cl$}, all of the other $n\!-\!1$ clients, \scalebox{0.9}{$\Cl_{\setminus i}$}, and the server \scalebox{0.9}{$\Ser$} jointly acts as the verifiers. This is different from Prio~\cite{corrigan2017prio} (the original SNIP deployment setting) that rely on multiple \textit{honest servers} to perform verification. 
\item In \name, verification can be performed even with malicious verifiers. This is essential in our setting since we have $m$ malicious clients (\emph{i.e.}, verifiers).\footnote{
The server works honestly towards verifying the integrity (its primary goal) but could behave maliciously  (possibly colluding with other malicious clients) to violate the privacy of the honest clients (see Sec. \ref{sec:overview:goals}).} For this, \name~uses Shamir's $t$-out-of-$n$ threshold scheme for the entire protocol. This allows any cohort of $t$ verifiers to reconstruct a secret and, hence, instantiate a SNIP protocol. If $t\!<\!n$, we have multiple such instantiations and can use the redundancy to perform the integrity check even with some of the verifiers being malicious. 
\squishend
 \end{tcolorbox} 

The full protocol is presented in Fig. \ref{fig:name:protocol}. The protocol involves a setup phase followed by four rounds.
 
\noindent\textbf{Setup Phase.} In the setup phase, all parties are initialized with the system-wide parameters, namely the security parameter $\kappa$, the number of clients $n$ out of which \textit{only} \scalebox{0.9}{$m < \lfloor\frac{n-1}{3} \rfloor$} can be malicious, public parameters for the key agreement protocol \scalebox{0.9}{$pp\randarrow \textsf{KA.param}(\kappa)$}, and a field \scalebox{0.9}{$\Fl$} where \scalebox{0.9}{$|\Fl|\geq 2^{\kappa}$}. 
\name~works in a synchronous protocol between the server $\Ser$ and the $n$ clients in four rounds. To prevent the server from simulating an arbitrary number
of clients, the clients register themselves with a specific user ID on the public bulletin board $\Bl$ and are authenticated with the help of standard public key infrastructure (PKI). 
The bulletin board $\Bl$ allows parties to register IDs only for themselves, preventing impersonation. More concretely, the PKI enables the clients to register identities (public keys), and sign messages using their identity (associated secret keys),
such that others can verify this signature, but cannot
impersonate them \cite{Katz}. We omit this detail for the ease of exposition. For notational simplicity, we assume that each client $\Cl_i$ is assigned a unique logical ID in the form of an integer $i$ in $[n]$. Each client holds as input a $d$-dimensional vector $u_i\in \Fl^d$ representing its local update. All clients have a private, authenticated communication channel with the server $\Ser$. 
Additionally, every party (clients and server) has read and write access to the public bulletin $\Bl$ via authenticated channels.
For every client $\Cl_i$, the server $\Ser$ maintains a list, \scalebox{0.9}{$\textsf{Flag}[i]$}, of all clients that have flagged $\Cl_i$ as malicious. All \scalebox{0.9}{$\textsf{Flag}[i]$} lists are initialized to be empty lists.

\noindent\textbf{Round 1 (Announcing Public Information).} In the first round, all the parties announce their public information relevant to the protocol on the public bulletin $\Bl$. Specifically, each client $\Cl_i$ generates its key pair  \scalebox{0.9}{$({pk}_i,{sk}_i)\randarrow\textsf{KA.gen}(pp)$} and advertises the public key $pk_i$ on the public bulletin $\Bl$. The server $\Ser$ publishes the validation predicate \scalebox{0.9}{$\Vl(\cdot)$} on $\Bl$.

\noindent\textbf{Round 2 (Generate and Distribute Proofs).} Every client generates shares of its private update $u_i$ and the proof $\pi_i$, and distributes these shares to the other clients \scalebox{0.9}{$\Cl_{\setminus i}$}. 
First, client \scalebox{0.9}{$\Cl_i$} generates a common pairwise encryption key \scalebox{0.9}{$sk_{ij}$} for every other client \scalebox{0.9}{$\Cl_j\in \Cl_{\setminus i}$} using the key agreement protocol, \scalebox{0.9}{$sk_{ij}\larrow \textsf{KA.agree}(sk_i,pk_j)$}. Next, the client generates the secret shares of its private update \scalebox{0.9}{$\{ (1,u_{i1}), \cdots, (n,u_{in}),\Psi_{u_i}\}\randarrow \textsf{SS.share}(u,[n],m+1)$}. The sharing of $u_i$ is performed dimension-wise; we abuse notations and denote the $j$-th such share by \scalebox{0.9}{$(j,u_{ij}), j \in [n]$}. Note that the client $\Cl_i$ generates a share $(i,u_{ii})$ for \textit{itself} as well which will be used later in the protocol. 
Next, the client $\Cl_i$ generates the proof  for the computation \scalebox{0.9}{$\Vl(u_i)\!=\!1$}. Specifically, it computes the polynomials $f_i, g_i$, and $h_i=f_i\cdot g_i$ and samples a set of Beaver's multiplication triples \scalebox{0.9}{$(a_i,b_i,c_i)\in \Fl^3, a_i\cdot b_i=c_i \in \Fl$}. Since the other clients will verify the proof,  client $\Cl_i$ then splits the proof to generate shares \scalebox{0.9}{$\pi_{ij}=\big( (j,h_{ij}),$} \scalebox{0.9}{ $(j,a_{ij}),(j,b_{ij}),(j,c_{ij}) \big )$} for every other client \scalebox{0.9}{$\Cl_j \in \Cl_{\setminus i}$}. The shares themselves are generated via \scalebox{0.9}{$\{(1,h_{i1}),\cdots,(i-1,h_{i(i-1)}),(i+1,h_{i(i+1)}),$}  \scalebox{0.9}{$\cdots, (n,h_{in}),\Psi_{h_i}\}$}\scalebox{0.8}{$\randarrow$}\scalebox{0.9}{$\textsf{SS.share}(h_i, [n] \setminus i, m+1)$}, and so on. Finally, the client encrypts the proof strings (shares of the update $u_i$ and the proof $\pi_i$) using the corresponding pairwise secret key, \scalebox{0.9}{$\overline{(j,u_{ij})||(j,\pi_{ij})}$}
\scalebox{0.8}{$\randarrow$}\scalebox{0.9}{$\textsf{AE.enc}\big(sk_{ij},(j,u_{ij})||(j,\pi_{ij})\big)$}, and publishes the encrypted proof strings on the public bulletin $\Bl$. The client also publishes the check strings \scalebox{0.9}{$\Psi_{u_i}$} and \scalebox{0.9}{$\Psi_{\pi_i}=(\Psi_{h_i},\Psi_{a_i},\Psi_{b_i},\Psi_{c_i})$} for verifying the validity of the shares of $u_i$ and $\pi_i$, respectively.

\noindent\textbf{Round 3 (Verify Proof)}. In this round, every client $\Cl_i$ partakes in the verification of the proofs $\pi_j$ of all other clients $\Cl_j\in \Cl_{\setminus i}$, under the supervision of the server $\Ser$. The goal of the server is to identify the malicious clients, $\Cl_M$. To this end, the server maintains a (partial) list, $\CA$ (initialized as an empty list), of clients it has so far identified as malicious. 
The proof-verification round consists of three phases as follows:

\noindent$(i)$\textit{ Verifying the validity of the secret shares}.
First, every client $\Cl_i$ downloads and decrypts their shares from the bulletin \scalebox{0.9}{$\Bl$}, \scalebox{0.85}{$(i,u_{ji})||(i,\pi_{ji})$} \scalebox{0.9}{$\larrow\textsf{AE.dec}\big(sk_{ij},\overline{(i,u_{ji})||(i,\pi_{ji})}\big), \forall \Cl_j \in \Cl_{\setminus i}$}. Additionally, $\Cl_i$ downloads the check strings \scalebox{0.9}{$(\Psi_{{u}_i},\Psi_{{\pi}_i})$} and verifies the validity of the shares. If the shares from any client $\Cl_j$: \squishlist  \vspace{-0.2cm}
\item fail to be decrypted, \emph{i.e.}, \scalebox{0.9}{$\textsf{AE.dec}(\cdot)$} outputs $\bot$, OR 
\item fail to pass the verification, \emph{i.e.}, \scalebox{0.9}{$\textsf{SS.verify}(\cdot)$} returns $0$,  \vspace{-0.2cm}
\squishend 
 \scalebox{0.9}{$\Cl_i$} flags \scalebox{0.9}{$\Cl_j$} on the bulletin \scalebox{0.9}{$\Bl$}. 
Every time a client \scalebox{0.9}{$\Cl_i$} flags another client \scalebox{0.9}{$\Cl_j$}, the server updates the corresponding list \scalebox{0.8}{$\textsf{Flag}[j]\hspace{-0.1cm}\leftarrow\textsf{Flag}[j]\cup \Cl_i$}. If \scalebox{0.9}{$|\textsf{Flag}[j]|\geq m+1$}, the server $\Ser$ marks $\Cl_j$ as malicious: \scalebox{0.9}{$\CA\leftarrow\CA\cup \Cl_j$}. The server can do so because the pigeon hole principle implies that $\Cl_j$ must have sent an invalid share to at least one honest client; hence, the correctness of the value recovered from that client's shares cannot be guaranteed. In case \scalebox{0.9}{$1\leq |\textsf{Flag}[j]|\leq m$,} the server supervises the following actions.  Suppose client $\Cl_i$ has flagged client $\Cl_j$. Client $\Cl_j$ then reveals the shares for $\Cl_i$, \scalebox{0.9}{$\big((i,u_{ji}),(i,\pi_{ji})\big)$} in the clear (on bulletin $\Bl$) for the server $\Ser$ (or anyone else) to verify using \scalebox{0.9}{$\textsf{SS.Verify}(\cdot)$}. If that verification passes, $\Cl_i$ is instructed by the server to use the released shares for its computations. Otherwise, $\Cl_j$ is marked as malicious by the server $\Ser$. Note that this does not lead to privacy violation for an honest client since at most $m$ shares corresponding to the $m$ malicious clients would be revealed (see Sec. \ref{sec:security_analysis}). If a client $\Cl_i$ flags $\geq m+1$ other clients, $\Ser$ marks  $\Cl_i$ as malicious.  Thus, at this point every client on the list $\CA$ has either \vspace{-0.2cm}
\squishlist\item provided invalid shares to at least one honest client, OR 
\item flagged an honest client.   \vspace{-0.2cm} \squishend
In other words, every client who is \textit{not} in \scalebox{0.9}{$\CA$}, \scalebox{0.9}{$\Cl_i\in \Cl\setminus \CA$}, is guaranteed to have submitted at least \scalebox{0.9}{$n\!-\!m\!-\!1$} valid shares for the honest clients in \scalebox{0.9}{$\Cl_H\setminus\Cl_i$} (see Sec. \ref{sec:security_analysis} for details). Additionally, the server cannot be tricked into marking an honest client as malicious, \emph{i.e.}, \name~ensures \scalebox{0.9}{$\CA\cap\Cl_H=\varnothing$} (see Sec. \ref{sec:security_analysis}). The server $\Ser$ publishes $\CA$ on the bulletin $\Bl$.

\noindent$(ii)$ \textit{Computation of proof summaries by clients.}  For this phase, the server \scalebox{0.9}{$\Ser$} advertises a random value \scalebox{0.9}{$r \in \Fl$} on the  bulletin $\Bl$.
Next, a client \scalebox{0.9}{$\Cl_i$} proceeds to distill the proof strings of all clients \textit{not} in \scalebox{0.9}{$\CA$} to generate summaries for the server \scalebox{0.9}{$\Ser$}. Specifically, client \scalebox{0.9}{$\Cl_i$ }prepares a proof summary \scalebox{0.9}{$\sigma_{ji}=\big((i,w^{out}_{ji}),(i,\lambda_{ji})\big)$ }for \scalebox{0.9}{$\forall\Cl_j \in \Cl\setminus(\CA\cup \Cl_i)$} as per the description in the previous section, and publishes it on  \scalebox{0.9}{$\Bl$}. 

\noindent $(iii)$ \textit{Verification of proof summaries by the server.} Next, the server moves to the last step of verifying the proof summaries \scalebox{0.9}{$\sigma_i=(w^{out}_i,\lambda_i )$} for all clients not in $\CA$. Recall from the discussion in Sec. \ref{sec:system:block} that this involves recovering the values \scalebox{0.9}{$w^{out}_i$} and \scalebox{0.9}{$\lambda_i$} from the shares of $\sigma_i$ and checking whether  \scalebox{0.9}{$w^{out}_{i}=1$} and \scalebox{0.9}{$\lambda_i=0$}. However, we cannot simply use the naive share reconstruction algorithm from Sec. \ref{sec:system:block} since some of the shares might be incorrect (submitted by the malicious clients). To address this issue, \name~performs a 
\textit{robust reconstruction} of the shares as follows.
A naive strategy would be sampling multiple subsets of $m+1$ shares (each subset can emulate a SNIP setting), reconstructing the secret for each subset, and taking the  majority vote. 
However, we can do much better by exploiting the connections between Shamir's secret shares and Reed-Solomon error correcting codes (Sec. \ref{sec:system:block}). Specifically, the Shamir's secret sharing scheme used by \name~is a \scalebox{0.9}{$[n\!-\!1,m+1,n\!-\!m]$} Reed-Solomon code that can correct up to $q$ errors and $e$ erasures (message dropouts) where \scalebox{0.9}{$2q+e<n\!-\!m\!-\!1$}. The server $\Ser$ can, therefore, use \scalebox{0.9}{$\textsf{SS.robustRecon}(\cdot)$} to reconstruct the secret when \scalebox{0.9}{$m<\lfloor\frac{n-1}{3}\rfloor$}.  

After the robust reconstruction of the proof summaries, the server $\Ser$ verifies them and updates the list $\CA$ with \textit{all} malicious clients with malformed updates.  Specifically:    \vspace{-0.1cm}\begin{gather*}\small \forall \Cl_i \in \Cl\setminus \CA\\\vspace{-0.5cm}
\hspace{-1.8cm} \small
\Big(\textsf{SS.robustRecon}(\{(j,w^{out}_{ij})\}_{\Cl_j\in \Cl\setminus\{\CA\cup\Cl_i\}})\neq 1 \thinspace\thinspace \vee\\  \vspace{-0.5cm}\small \textsf{SS.robustRecon}(\{(j,\lambda_{ij})\}_{\Cl_j\in \Cl\setminus\{\CA\cup\Cl_i\}}) \neq 0\Big) \\\hspace{5cm}\small\implies \CA\leftarrow\CA \cup \Cl_i. \vspace{-0.5cm}\end{gather*}
Additionally, if  a client $\Cl_i$ withholds some of the shares of the proof summaries for other clients, $\Cl_i$ is marked as malicious as well by the server. Thus, in addition to the malicious clients listed above, the list $\CA$ now has all clients that have either:
\squishlist  \vspace{-0.2cm}\item failed the proof verification, \emph{i.e.}, provided malformed updates, OR \item withheld shares of proof summaries of other clients (malicious message dropout).
\squishend   \vspace{-0.2cm}
To conclude the round, the server publishes the updated list $\CA$ on the public bulletin $\Bl$.
\vspace{-0.1cm}

\noindent\textbf{Round 4 (Compute Aggregate).} This is the final round of \name~where the  aggregate of the well-formed updates is computed. If a client $\Cl_i$ is on $\CA$ wrongfully, it can dispute its malicious status by showing the other clients the transcript of the robust reconstruction from all the shares of $\sigma_i$ (publicly available on bulletin $\Bl$).
If any client $\Cl_i \in \Cl$ successfully raises a dispute, all clients abort the protocol because they conclude that the server $\Ser$ has acted maliciously by trying to withhold a verified well-formed update from the aggregation. If no client raises a successful dispute, every client \scalebox{0.9}{$\Cl_i \in \Cl\setminus\CA$} generates its share of the aggregate, \scalebox{0.9}{$(i,\Ul_{i})$} with \scalebox{0.9}{$\Ul_{i}=\sum_{\Cl_j\in \Cl\setminus \CA}u_{ji}$}, and sends that share to the server $\Ser$. Note that, herein, $\Cl_i$ uses its own share of the update, \scalebox{0.9}{$(i,u_{ii})$}, as well.
\\The server recovers the aggregate \scalebox{0.9}{$\Ul=\sum_{\Cl_i\in \Cl\setminus \CA} \Ul_j$} using robust reconstruction: \scalebox{0.9}{$\Ul\leftarrow \textsf{SS.robustRecon}(\{(i,\Ul_{i})\}_{\Cl_i\in \Cl\setminus \CA})$}.

\textbf{Discussion.} \name~ meets the design goals of Sec. \ref{sec:goals} as follows.\vspace{-0.3cm}\\\\
\textit{Flexibility of Integrity Checks.} SNIP supports arbitrary arithmetic circuits for \scalebox{0.9}{$\Vl(\cdot)$}. The server \scalebox{0.9}{$\Ser$} can choose a different  \scalebox{0.9}{$\Vl(\cdot)$} for every iteration (the protocol described above corresponds to a single iteration of model training in FL). Additionally, \scalebox{0.9}{$\Ser$} can hold 
multiple \scalebox{0.9}{$\Vl_1(\cdot),\cdots,\Vl_k(\cdot)$} and want
to check whether the client’s update passes them all. For this, we have \scalebox{0.9}{$\Vl_i(\cdot)$} return zero
(instead of one) on success. If \scalebox{0.9}{$w^{out}_i$} is the value on the output wire of the circuit \scalebox{0.9}{$\Vl_i(\cdot)$}, the server chooses random values \scalebox{0.9}{$(l_1, \cdots ,l_k) \in \Fl^k$} and recovers the sum \scalebox{0.9}{$\sum_{i=1}^kl_i\cdot w^{out}_i$} in Round 3. If any \scalebox{0.9}{$w^{out}_i=0$}, then the sum will be non-zero with high probability and $\Ser$ will reject.
\vspace{-0.2cm}\\\\\textit{Compatibility with FL's Infrastructure.} Current deployments of FL involves a \textit{single} server who wants to train the global model. Hence, as explained above, we design SNIP to be compatible with a single server in \name. Solutions involving two or more non-colluding servers are unrealistic for FL. For instance, currently the server can be owned by Meta who wants to train privately on the data of its user base. For a two-server model here, the second server has to be owned by an independent party. Moreover, both the servers have to do an equal amount of computation for model training (verification, aggregation etc) since SNIP uses secret shares. This would make sense only if \textit{both} the servers are interested in training the model. For instance, if Meta and Google collaborate to train a model on their joint user base which is an unrealistic scenario. 
\vspace{-0.2cm}\\\\
\textit{Efficiency.} \name's usage of SNIP as the underlying ZKP is made from the efficiency point of view. SNIP is a light-weight ZKP system that is \textit{specialized for the server-client settings} resulting in good performance. For instance, its performance is about three-orders of magnitude better than that of zkSNARKs~\cite{corrigan2017prio}. Instead of using ZKPs, one alternative is to use standard secure multi-party computation (MPC) for the entire aggregation to directly compute \scalebox{0.9}{$\Ul_{valid}=\sum_{\Cl_i}\Vl(u_i)\cdot u_i$}. However, doing the entire aggregation under MPC would result in a massive circuit with \scalebox{0.9}{$O(nd)$} multiplication gates where \scalebox{0.9}{$d$} is the data dimension. Multiplications are costly for MPC and each gate requires a round of communication in general making the above computation prohibitively costly. Extending the  computation to the malicious threat model would be even costlier. This is where SNIP proves to be advantageous: SNIP enables the verifiers to check all the multiplication gates very efficiently (in a non-interactive fashion) with just one polynomial identity test (Sec. \ref{sec:system:block}). \vspace{-0.3cm}
\begin{tcolorbox}[sharp corners, breakable]\vspace{-0.2cm}\textbf{Remark 3.} In a nutshell, \name's technical novelty is in providing an \textit{efficient} extension of SNIP to a $(1)$ fully malicious threat model in a $(2)$ single server setting.

\name's problem setting is different from that of the original SNIP proposal in the following ways:
\squishlist
\item The original SNIP deployment setting (Prio) uses $\geq 2$  non-colluding servers as the verifiers; 
\name~requires a single server.
\item Originally, SNIP considers honest verifiers; \name~supports a fully malicious threat model (provers and verifiers).
\squishend
The above changes cannot be supported in SNIP as is and are necessary to capture the constraints of a realistic FL setting.

The core essence of the technical differences between the original deployment of SNIP and \name~is that the roles of the parties are changed: the former has a clear distinction between the prover (clients) and verifiers ($k\geq 2$ honest servers), whereas the clients and the single server jointly act as the (malicious) verifiers in \name.

Consequently, SNIP's interaction pattern is different in \name~which required the following technical changes.

\squishlist
\item Originally, SNIP uses additive secret shares whereas \name~uses Shamir’s threshold secret sharing.
\item In the original construction of SNIP, the reconstruction (a fundamental operation in the protocol) of shares is very simple: just add (+) the individual shares. In \name, we propose a robust reconstruction technique, $\textsf{SS.robustRecon}(\cdot)$, based on Reed-Solomon codes. This is key in ensuring robust verification even with malicious verifiers.
\item \name, in addition, requires verifiable secret shares: shares are augmented with a check string $\Psi$ which is integrated into the protocol.
\item Messages (shares, proofs) are distributed in the clear in the original SNIP deployment; messages are encrypted in \name.
\item We propose novel optimizations for \name~(Sec. \ref{sec:opt:new} targets the new operations for verifying secret shares and robust reconstruction while Sec. \ref{sec:opt:general} provides general improvements for SNIP).
\squishend

\vspace{-0.2cm} \end{tcolorbox}
\begin{figure*}\small\fbox
{%
    \begin{varwidth}{0.99\textwidth}\begin{itemize}\item\textbf{Setup Phase.} \begin{itemize}\item All parties are given the security parameter $\kappa$, the number of clients $n$ out of which at most \scalebox{0.9}{$m<\lfloor \frac{n-1}{3}\rfloor$} are malicious,  honestly generated
\scalebox{0.9}{$pp \randarrow \textsf{KA.gen}(\kappa)$} and a field $\Fl$ to be used
for secret sharing. Server initializes lists \scalebox{0.9}{$\textsf{Flag}[i]=\varnothing, i \in [n]$} and \scalebox{0.9}{$\CA=\varnothing$}. 
 \end{itemize}
 \item \textbf{Round 1 (Announcing Public Information).} \\
\textit{Client}: Each client $\Cl_i$
\begin{itemize}\item
Generates its key pair and announces the public key. $(pk_i,sk_i)\randarrow \textsf{KA.gen}(pp)$, $\Cl_i \xrightarrow{pk_i} \Bl$.\end{itemize}

\textit{Server}:
\begin{itemize}\item Publishes the validation predicate $\Vl(\cdot)$. $\Ser \xrightarrow{\Vl(\cdot)}\Bl$  \end{itemize} \item \textbf{Round 2 (Generate and Distribute Proof).}\\\textit{Client}: Each client $C_i$\begin{itemize}\item Computes $n-1$ pairwise keys. $\forall \Cl_j \in \Cl_{\setminus i}, sk_{ij}\larrow \textsf{KA.agree}(pk_j,sk_i)$ \item Generates proof $\pi_{i}=\big( h_i,(a_i,b_i,c_i) \big), h_i \in \Fl[X],$ $(a_i,b_i,c_i)\in \Fl^3,$ $a_i\cdot b_i=c_i$ for the statement $\Vl(u_i)=1$. \item Generates shares of the input $u_i\in \Fl^d$. $\{(1,u_{i1}),\cdots ,(n,u_{in}),\Psi_{u_{i}}\}\randarrow\textsf{SS.share}(u_i,[n],m+1)$ \item Generates shares of the proof $\pi_i$.\vspace{-0.3cm} \begin{gather*}\small\{(1,h_{i1}),\cdots,(n,h_{in}),\Psi_{h_{i}}\}\randarrow \textsf{SS.share}(h_i,[n]\setminus i,m+1), \{(1,a_{i1}),\cdots,(n,a_{in}),\Psi_{a_{i}}\}\randarrow \textsf{SS.share}(a_i,[n]\setminus i,m+1 )\small\\\{(1,b_{i1}),\cdots,(n,b_{in}),\Psi_{b_{i}}\}\randarrow \textsf{SS.share}(b_i,[n]\setminus i,m+1 ), \{(1,c_{i1}),\cdots,(n,c_{in}),\Psi_{c_{i}}\}\randarrow \textsf{SS.share}(c_i,[n]\setminus i,m+1 )\end{gather*}\vspace{-0.35cm}\item Encrypts proof strings for all other clients. $ \forall \Cl_j\in\Cl_{\setminus i}, \overline{(j,u_{ij})||(j,\pi_{ij})}\randarrow \textsf{AE.enc}\big(sk_{ij},(j,u_{ij})||(j,\pi_{ij})\big), \pi_{ij}=h_{ij}||a_{ij}||b_{ij}||c_{ij}$. \item Publishes check strings and the encrypted proof strings on the bulletin. 
$ \forall \Cl_j\in\Cl_{\setminus i}, \Cl_i\xrightarrow{\overline{(j,u_{ij})||(j,\pi_{ij})}}\Bl; \Cl_i\xrightarrow{\Psi_{{u}_i},\Psi_{\pi_i}}\Bl$
\end{itemize}
\item \textbf{Round 3 (Verify Proof)}.
\\(i) \textit{ Verifying validity of secret shares}:
\\\textit{Client}:
Each client $\Cl_i$
\begin{itemize}\item Downloads and decrypts proof strings for all other clients from the public bulletin.  Flags a client in case their decryption fails. \vspace{-0.2cm}\begin{gather*}\vspace{-0.2cm}\small\forall \Cl_j \in \Cl_{\setminus i}, \Cl_i\xleftarrow{\overline{(i,u_{ji})||(i,\pi_{ji})},\Psi_{{u}_j},\Psi_{{\pi}_j}}\Bl, (i,u_{ji})||(i,\pi_{ji}) \leftarrow \textsf{AE.dec}\big(sk_{ij},\overline{(i,u_{ji})||(i,\pi_{ji})}\big)\\\vspace{-0.1cm} \small\bot\leftarrow \textsf{AE.dec}\big(sk_{ij},\overline{(i,u_{ji})||(i,\pi_{ji})}\big) \implies Cl_i\xrightarrow {\text{Flag } \Cl_j} \Bl \end{gather*}  \item Verifies the shares $u_{ji}(\pi_{ji})$ using checkstrings $\Psi_{{u}_j} (\Psi_{{\pi}_j})$ and flags all clients with invalid shares.\vspace{-0.3cm}
\begin{gather*}\forall \Cl_j \in \Cl_{\setminus i}, 0 \leftarrow \big(\textsf{SS.verify}((i,u_{ji}),\Psi_{{u}_{j}})\wedge \textsf{SS.verify}((i,\pi_{ji}),\Psi_{{\pi}_{j}})\big)\implies  \Cl_i \xrightarrow{\text{Flag }\Cl_j}\Bl  \end{gather*}
\end{itemize}
\vspace{-0.8cm}
\textit{Server}:
\begin{itemize}\item If client $\Cl_i$ flags client $\Cl_j$, the server updates $\textsf{Flag}[j]=\textsf{Flag}[j]\cup \Cl_i$ \item  Updates the list of malicious client  $\CA$ as follows: \begin{itemize}[label=\scalebox{0.7}{$\blacktriangleright$}]\item Adds all clients who have flagged $\geq m+1$ other clients. $\forall \Cl_i \mbox{ s. t. } Z=\{j|\Cl_i \in \textsf{Flag}[j]\}, |Z|\geq m+1\implies \CA\leftarrow\CA\cup \Cl_i$\item Adds all clients with more than $m+1$ flag reports. $|\textsf{Flag}[i]|\geq m+1 \implies \CA\leftarrow\CA\cup \Cl_i$ \item For clients with less flag reports, the server obtains the corresponding shares in the clear, verifies them and updates $\CA$ accordingly. $\forall \Cl_j \text{ s.t } 1 \leq |\textsf{Flag}[j]|\leq m, \forall \Cl_i \text{ s.t. } \Cl_i \text{ has flagged } \Cl_j $ \begin{itemize}[label=\scalebox{1}{$-$}] \item $\Cl_j \xrightarrow{(i,u_{ji}),(i,\pi_{ji})} \Bl$ \item $
       \text{if } \big(\textsf{SS.verify}((i,u_{ji}),\Psi_{{u}_{j}}) \wedge \textsf{SS.verify}((i,\pi_{ji}),\Psi_{\pi_{j}})\big)=0 \implies \CA\leftarrow\CA\cup \Cl_j, \text{ otherwise, } \Cl_i$  uses the verified shares to compute its proof summary  $\sigma_{ji}\vspace{-0.1cm} $\end{itemize}
       \end{itemize}
       \item Publishes $\CA$ on the bulletin. $\Ser\xrightarrow{\CA}\Bl$
       \end{itemize}

     (ii)  \textit{ Generation of proof summaries by the clients.}
   \\\textit{Server}: \begin{itemize}\item Server announces a random number $r\in \Fl$. $\Ser \xrightarrow{r}\Bl$\end{itemize}
   \textit{Client}: Each client $\Cl_i \in \Cl\setminus\CA$
     \begin{itemize}
\item Generates a summary $\sigma_{ji}$ of the proof string $\pi_{ji}$ based on $r$, $\forall \Cl_j \in \Cl\setminus (\CA\cup\Cl_i), \Cl_i\xleftarrow{r}\Bl,\sigma_{ji}=\big( (i,w^{out}_{ji}),(i,\lambda_{ji})\big),\Cl_i\xrightarrow{\sigma_{ji}}\Bl$ \end{itemize}

(iii)\textit{ Verification of proof summaries by the server.}\\
\textit{Server}:
\begin{itemize}\item  
 Downloads and verifies the proof for all clients not on $\CA$ via robust reconstruction of the digests and updates $\CA$ accordingly. \scalebox{0.87}{$\forall \Cl_i \in \Cl\setminus\CA, \Ser\xleftarrow{\sigma_{ij}}\Bl, \big( \textsf{SS.robustRecon}(\{(j,w^{out}_{ij})\}_{\Cl_j \in \Cl\setminus(\CA\cup\Cl_i)}) \neq 1 \vee \textsf{SS.robustRecon}(\{(j,\lambda_{ij})\}_{\Cl_j \in \Cl\setminus(\CA\cup\Cl_i)})\neq 0 \big)\implies \CA\leftarrow\CA\cup \Cl_i$} \item Publishes the updated list $\CA$ on the  bulletin. $\Ser \xrightarrow{\CA} \Bl$\end{itemize}

 \item \textbf{Round 4 (Compute Aggregate).}\\\textit{Client}: Each client $\Cl_i$  \begin{itemize}\item If $\Cl_i$ is on $\CA$, $\Cl_i$ raises a dispute by sending the transcript of the reconstruction of $\sigma_i$ that shows \scalebox{0.9}{$\lambda_i=0\wedge w^{out}_j=1$} and aborts, OR\\ $\forall \Cl_j \in \Cl_{\setminus i}, \Cl_i \xleftarrow{\sigma_{ij}}\Bl,$ $ \Cl_i\xrightarrow{\mbox{Transcript of } \textsf{SS.robustRecon}(\{(j,\sigma_{ij})\}_{\Cl_j \in \Cl\setminus(\CA\cup\Cl_i)})}\Bl$ 
 \item Aborts protocol if it sees any other client on $\CA$ successfully raise a dispute,  OR
 \item 
If no client has raised a dispute and $\Cl_i$ is not on $\CA$, sends the aggregate of the shares of clients in \scalebox{0.9}{$\Cl\setminus\CA$} to the server. \scalebox{0.9}{$\Ul_{i}=\underset{\Cl_j\in \Cl\setminus \CA}{\sum u_{ji}}, \Cl_i\xrightarrow{\Ul_i}\Ser$}
 \end{itemize}\vspace{-0.2cm} 
\textit{Server}:\begin{itemize} \item Reconstructs the final aggregate. $\Ul\leftarrow \textsf{SS.robustRecon}(\{(i,\Ul_i)\}_{\Cl_i \in \Cl\setminus\CA })$
\end{itemize}

\end{itemize}
\end{varwidth}}
\caption{\name: Description of the secure aggregation with verified inputs protocol.}\label{fig:name:protocol}\end{figure*}

Table \ref{tab:complexity} analyses the complexity of  \name~in terms of the number of clients $n$, number of malicious clients $m$ and data dimension $d$. We assume that  \scalebox{0.9}{$|\Vl|$} is of the order of \scalebox{0.9}{$O(d)$}. The total number of one-way communication is $12$ and $9$ for the clients and the server, respectively. A per-round analysis is presented in App. \ref{sec:complexity}. 

\begin{table}
\centering
\resizebox{\columnwidth}{!}{\begin{tabular}{lccccl}\toprule
 &  \bf Computation &  \bf Communication\\ 
\\\midrule
\bf Client   & $O(mnd)$ & $O(mnd)$ \\
\bf Server & $O\big((n+d)n\log^2n\log \log n+md\min(n,m^2)\big)$ & $O\big(n^2+md\min(n,m^2)\big)$ \\\bottomrule
\end{tabular}}
\caption{Computational and communication complexity of \name~for the server and an individual client. 
}\label{tab:complexity}\vspace{-0.8cm}
\end{table}


\section{Security Analysis}\label{sec:security_analysis}
In this section, we formally analyze the security of \name. 
\vspace{-0.2cm}
\begin{theorem}For any public validation predicate \scalebox{0.9}{$\Vl(\cdot)$} that can be expressed by an arithmetic circuit, \name~is a \SAIV~protocol (Def. \ref{def:SAVI}) for \scalebox{0.9}{$|\Cl_M|<\lfloor\frac{n-1}{3}\rfloor$} and $\Cl_\Vl=\Cl\setminus\CA$.\vspace{-0.2cm} \label{thm:main}\end{theorem}
We present a proof sketch of the above theorem here; the formal proof is in App. \ref{app:security}.

\textit{Proof Sketch.} The proof  relies on the following two facts. 
\\\textbf{Fact 1.} \textit{Any set of $m$ or less shares in \name~reveals nothing about the secret.}\\
\textbf{Fact 2.} \textit{A $(n,m\!+\!1,n\!-\!m)$ Reed-Solomon error correcting code can correctly construct the message with up to  $q$ errors and $e$ erasures (message dropout), where $2q+e<n-m
+1$. In \name, we have $q+e\!=\!m$ where $q$ is the number of malicious clients that provide erroneous shares and $e$ is the number of clients that withhold a message or are barred from participation (\emph{i.e.}, are in $\Cl^*$).} 

\textit{Integrity.} We prove that \name~satisfies the integrity constraint of the \SAIV~protocol using the following three lemmas. 
\begin{lemma}\vspace{-0.3cm} 
\name~accepts the update of every honest client.\label{lemma:1}
\begin{equation}\vspace{-0.3cm}
\forall \Cl_i \in \Cl_{H}: \Pr_{\name}[\textsf{Accept } u_i]=1. \label{eq:complete}
\end{equation}
\end{lemma}\vspace{-0.2cm}
\begin{proof} 
By definition, client $\Cl_i\in \Cl_H$ has well-formed inputs, that is, $\Vl(u_i)\!=\!1$. Additionally, $\Cl_i$, by virtue of being honest, submits valid shares. Hence, at least $n-m-1$ other honest clients $\Cl_H\setminus \Cl_i$ will produce correct shares of the proof summary $\sigma_i=(w^{out}_i,\lambda_i)$. Using Fact 2, the server $\Ser$ is able to correctly reconstruct the value of $\sigma_i$. Eq. \ref{eq:complete} is now implied by the completeness property of SNIP.  \vspace{-0.2cm}
\end{proof}\vspace{-0.4cm}
\begin{lemma} 
All updates accepted by \name~are well-formed with probability $1-\mathrm{negl}(\kappa)$. \label{lemma:2}
\begin{gather*}\vspace{-0.3cm}\small\forall \Cl_i \in \Cl, \Pr_\name\big[\Vl(u_i)=1\big| \textsf{Accept } u_i\big]=1-\mathrm{negl}(\kappa).\numberthis\label{eq:soundness}
\end{gather*}
\vspace{-0.8cm}\label{lemma:4}\end{lemma}
The proof relies on the fact that a client will be verified only if it has submitted $\geq n-m-1$ valid shares (see App. \ref{app:proof1}).
\vspace{-0.4cm}
\begin{corollary}\name~rejects all malformed updates with probability $1-\mathrm{negl}(\kappa)$.\vspace{-0.3cm}\end{corollary}
Based on the above lemmas, at the end of Round 3, $\Cl\setminus \CA$ (set of clients whose updates have been accepted) must contain \textit{all} honest clients $\Cl_H$. Additionally, it may contain some clients $\Cl_i$ who have submitted well-formed updates with at least $n-m-1$ valid shares for $\Cl_H$, but may act maliciously for other steps of the protocol (for instance, give incorrect shares of proof summary for other clients or give incorrect shares of the final aggregate). This is acceptable provided that \name~is able to reconstruct the final aggregate containing \textit{only} well-formed updates which is guaranteed by the following lemma. 
\begin{lemma}\vspace{-0.4cm}The aggregate $\Ul$ must contain the updates of \textit{all} honest clients or the protocol is aborted.\label{lemma:3} \begin{gather*}\vspace{-0.5cm} \small\Ul=\Ul_H +\sum_{\Cl_i\in \bar{\Cl}} u_i \mbox{ where }\Ul_H=\sum_{\Cl_i\in \Cl_H}u_i\\\small\bar{\Cl}\subseteq \Cl\setminus\{\CA\cup \Cl_H\}\numberthis \label{eq:aggregate} \vspace{-0.5cm}\end{gather*}\end{lemma}\vspace{-0.3cm}
\begin{proof} \vspace{-0.1cm}If the server \scalebox{0.9}{$\Ser$} acts maliciously and publishes a list \scalebox{0.9}{$\CA$} such that  \scalebox{0.9}{$\CA\cap \Cl_H\neq \varnothing$}, an honest client \scalebox{0.9}{$\Cl_i \in \CA\cap \Cl_H$} publicly raises a dispute. This is possible since all the shares of $\sigma_i$ are publicly logged on \scalebox{0.9}{$\Bl$}. If the dispute is successful, all honest clients will abort the protocol. Note that a malicious client with malformed updates cannot force the protocol to abort in this way since it will not be able to produce a successful transcript with high probability (Lemma \ref{lemma:2}). If no clients raise a successful dispute, Eq. \ref{eq:aggregate} follows directly from Fact 2.  \scalebox{0.9}{$\bar{\Cl}$} represents a set of  malicious clients with well-formed updates which corresponds to \scalebox{0.9}{$\Cl_\Vl\setminus\Cl_H$} in Eq. \ref{eq:SAVI:U}. \end{proof} \vspace{-0.3cm}
\textit{Privacy.} The privacy constraint of \SAIV~states that nothing should be revealed about a private update $u_i$ for an honest client $\Cl_i$, except: 
\squishlist\vspace{-0.2cm}
\item $u_i$ passes the integrity check, \emph{i.e.}, $\Vl(u_i)=1$ 
\item anything that can be learned from the aggregate of honest clients, $\Ul_H$.  
\squishend\vspace{-0.2cm}
We prove that \name~satisfies this privacy constraint with the help of the following two helper lemmas.
\begin{lemma}\vspace{-0.2cm}In Rounds 1-3, for an honest client $\Cl_i\in \Cl_H$, \name~reveals nothing about $u_i$ except $\Vl(u_i)=1$.\vspace{-0.4cm}\label{lemma:5}\end{lemma} The proof uses the fact that only $m$ shares of $\Cl_i$, which correspond to the $m$ malicious clients, can be revealed (see App. \ref{app:proof2}).
\begin{lemma}\vspace{-0.2cm}
In Round $4$, for an honest client $\Cl_i\in \Cl_H$, \name~reveals nothing about $u_i$ except whatever can be learned from the aggregate.\vspace{-0.2cm}
\end{lemma}
\begin{proof}\vspace{-0.2cm}
In Round $4$, from Lemma \ref{lemma:3} and Fact 2,  the information revealed is either the aggregate or $\bot$. \vspace{-0.2cm}
\end{proof}

\vspace{-0.2cm}\section{\textsf{EIFF\texorpdfstring{\MakeLowercase{e}}{e}L}~Optimizations}\label{sec:opt}

\subsection{Probabilistic Reconstruction}\label{sec:opt:new}
The Gao's decoding algorithm alongside the use of verifiable secret sharing guarantees that the correct secret will be recovered (with probability one). However,  we can improve performance at the cost of a small probability of failure. \vspace{-0.2cm}
\\\\\textbf{Verifying Secret Shares.} As discussed in Sec. \ref{sec:complexity}, verifying the validity of the secret shares is the dominating cost for client-side computation. To reduce this cost, we propose an optimization where the validation of the shares corresponding to the proof \scalebox{0.9}{$\pi_i=\big( h_i, (a_i,b_i,c_i)\big)$} can be eliminated. Specifically, we propose the following changes to  Round 3:\vspace{-0.2cm}
\squishlist\item Each client $\Cl_i$ skips verifying the validity of the shares $(i,\pi_{ji})$ for  $\Cl_j \in \Cl_{\setminus i}$. \item Let \scalebox{0.9}{$e=|\CA|$}. The server \scalebox{0.9}{$\Ser$} samples two sets of clients \scalebox{0.9}{$P_1, P_2$} from \scalebox{0.9}{$\Cl \setminus \{\Cl_i\cup\CA\}$} of size at least \scalebox{0.9}{$3m-2e+1$} (\scalebox{0.9}{$P_1, P_2$} can be overlapping) and performs Gao's decoding on both the sets to obtain polynomials $p_1$ and $p_2$. The server accepts the  \scalebox{0.9}{$w^{out}_{i}$} (\scalebox{0.9}{$\lambda_i$}) only iff \scalebox{0.9}{$p_1\!=\!p_2$ }and  \scalebox{0.9}{$p_1(0)=p_1(0)=1 (p_1(0)=p_1(0)=0)$}. The cost of this step is \scalebox{0.9}{$O(n^2\log^2 n\log\log n)$} which is less than verifying the shares of \scalebox{0.9}{$\pi_i$} when \scalebox{0.9}{$m<n\ll d$} (improves runtime by \scalebox{0.9}{$2.3\times$}, see Table \ref{tab:E2E}). \vspace{-0.2cm}\squishend
 Note that a \scalebox{0.9}{$[n,k, n\!-\!k+1]$} Reed-Solomon error correcting code can correct up to \scalebox{0.9}{$\lfloor \frac{n-k -l}{2}\rfloor$} errors with $l$ erasures. Thus, with \scalebox{0.9}{$m\!-\!e$} malicious clients, only \scalebox{0.9}{$3m \!-\! 2e\!+\!1$} shares are sufficient to correctly reconstruct the secret for honest clients.  Since, the random sets $P_1$ and $P_2$ are not known, a malicious client with more than \scalebox{0.9}{$m \!-\! e$} invalid shares can cheat only with probability at most $\nicefrac{1}{\binom{3m-2e+2}{n-e}}$. We cannot extend this technique for the secret shares of the update $u$, because, unlike the value of the digests \scalebox{0.9}{($w^{out}\!=\!1,\lambda\!=\!0$)}, the final aggregate is unknown and needs to be reconstructed from the shares.
 
 \textit{Improvement.} Eliminates verification of check strings for the proof $\pi_i$ which reduces time by $2.3\times$ (Table \ref{tab:E2E}).
  \\\textit{Cost.} Additional \scalebox{0.9}{$\nicefrac{1}{\binom{3m-2e+2}{n-e}}$}
  probability of failure where \scalebox{0.9}{$e=|C^*|$ }.

\textbf{Robust Reconstruction.} In case \scalebox{0.9}{$m\leq \sqrt{n}-2$}, the robust reconstruction mechanism can be optimized as follows. 
Let \scalebox{0.9}{$q=m-|\CA|$} be the number of malicious clients that remain undetected.
The server $\Ser$ partitions the set of clients in $\Cl\setminus \CA$ into at least $q+2$ disjoint partitions, $P=\{P_1,\cdots, P_{q+2}\}$ each of size $m+1$. Let $p_j(x)=c_{j,0} + c_{j,1}x+c_{j,2}x^2+ \cdots +c_{j,m}x^{m}$ represent the polynomial corresponding to the $m+1$ shares of partition $P_j$. Recall that recovering just $p_j(0)=c_{j,0}$ suffices for a typical Shamir secret share reconstruction. However, now, the server $\Ser$ recovers the entire polynomial $p_j$, \emph{i.e.}, all of its coefficients  $\{c_{j,0}, c_{j,1},
\cdots, c_{j,q}\}$ for all $q+2$ partitions. Based on the pigeon hole principle, it can be argued that at least two of the partitions $(P_l,P_k \in P)$ will consist of \textit{honest} clients only. Hence, we must have at least two polynomials $p_l$ and $p_k$ that match and the value of the secret is their constant coefficient $p_l(0)$.
Note that the above mentioned optimization of skipping verifying the shares of the proof can be applied here as well.
A malicious client can cheat (\emph{i.e.}, make the server $\Ser$ accept even when \scalebox{0.9}{$w^{out}_i\neq 1 \vee \lambda_i\neq 0$} or reject the proof for an honest client) only if they can  manipulate the shares of at least two partitions which must contain at least \scalebox{0.9}{$2(m+1)-q$} honest clients.  Since the random partition $P$ is not known to the clients, this can happen only with probability $\nicefrac{1}{\binom{2(m+1)-q}{n-m-1}}$. 

\textit{Improvement.} Reduces the number of polynomial interpolations. \\
\textit{Cost.} Additional \scalebox{0.9}{$\nicefrac{1}{\binom{2(m+1)-q}{n-m-1}}$} probability of failure where \scalebox{0.9}{$q=m-|\CA|$ }.

\vspace{-0.2cm}
\subsection{Crypto-Engineering Optimizations}\label{sec:opt:general}

\textbf{Equality Checks.} The equality operator $=$ is relatively complicated to implement in an arithmetic circuit. To circumvent this issue, we replace any validation check of the form \scalebox{0.9}{$\Phi(u)=c_1\vee \Phi(u)=c_2\vee \cdots $} \scalebox{0.9}{$\vee\Phi(u)=c_k$} in the output nodes of \scalebox{0.9}{$\Vl(\cdot)$}, where \scalebox{0.9}{$\Phi(\cdot)$} is some arithmetic function, by an output of the form \scalebox{0.9}{$(\Phi(u)-c_1)\times\cdots\times(\Phi(u)-c_k)$}. Recall that in \name, the honest clients have well-formed inputs that satisfy \scalebox{0.9}{$\Vl(\cdot)$} by definition. Hence, this optimization does not violate the privacy of honest, which is our security goal.

\textit{Improvement.} Reduces the circuit size $|\Vl|$.
\\\textit{Cost.} No cost.

\textbf{Proof Summary Computation.} In addition to being a linear secret sharing scheme, Shamir's scheme is also multiplicative: given the shares of two secrets $(i,z_i)$ and  $(i,v_i)$, a party can locally compute $(i,s_{i})$ with $s\!=\!z\cdot v$. However, if the original shares correspond to a polynomial of degree $t$, the new shares represent a polynomial of degree $2t$. Hence, we do not rely on this property for the multiplication gates of \scalebox{0.9}{$\Vl(\cdot)$} as it would support only limited number of  multiplications. However, if \scalebox{0.9}{$m\!<\!\frac{n-1}{4}$}, we can still leverage the multiplicative property to generate shares of the random digest \scalebox{0.9}{$\lambda_i=f_i(r)\cdot g_i(r)=h_i(r)$} locally (instead of using Beaver's triples). 

\textit{Improvement.}  Saves a round of communication and reduces the number of robust reconstructions for $\lambda_i$ from three to just one (details in App. \ref{app:background}).\\\textit{Cost.} No cost. 


\textbf{Random Projection}.
As shown in Table \ref{tab:complexity}, both communication and computation grows linearly with the data dimension $d$.  Hence, we rely on the random projection \cite{WH} technique for reducing the dimension of the updates. Specifically, we use the fast random projection using Walsh-Hadamard transforms \cite{Ailon06}.

\textit{Improvement.} Reduces the data dimension which helps both computation and communication cost. \\
\textit{Cost.} Empirical evaluation (Sec. \ref{sec:eval:models}) shows that the efficacy of $\Vl(\cdot)$ is still preserved.
\begin{figure*}[h]
    \begin{subfigure}{0.24\linewidth}
    \centering \includegraphics[width=0.9\linewidth]{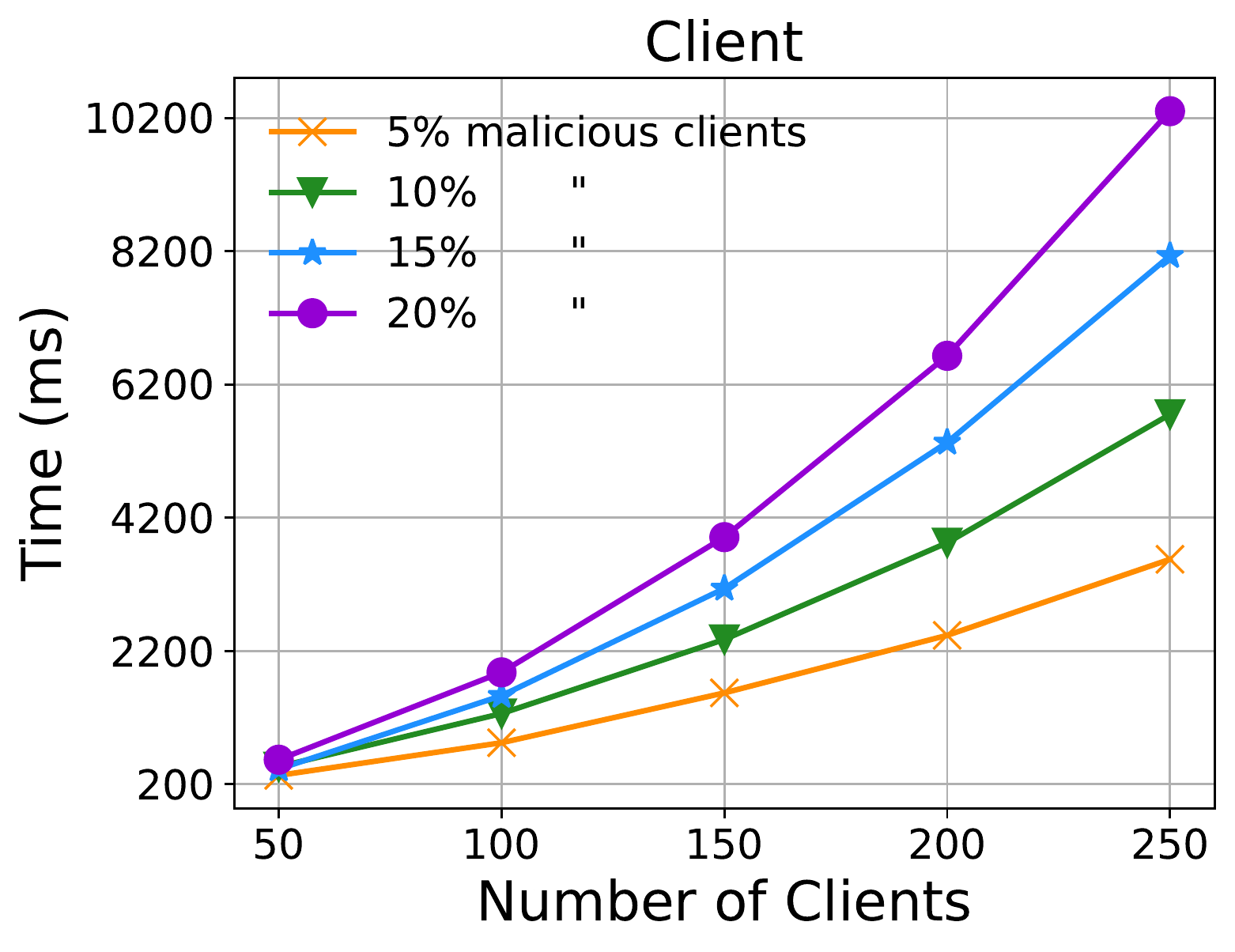}  \vspace{-0.2cm} 
        \label{fig:comp:client:n}\end{subfigure}
         \begin{subfigure}{0.24\linewidth}
        \centering
         \includegraphics[width=0.9\linewidth]{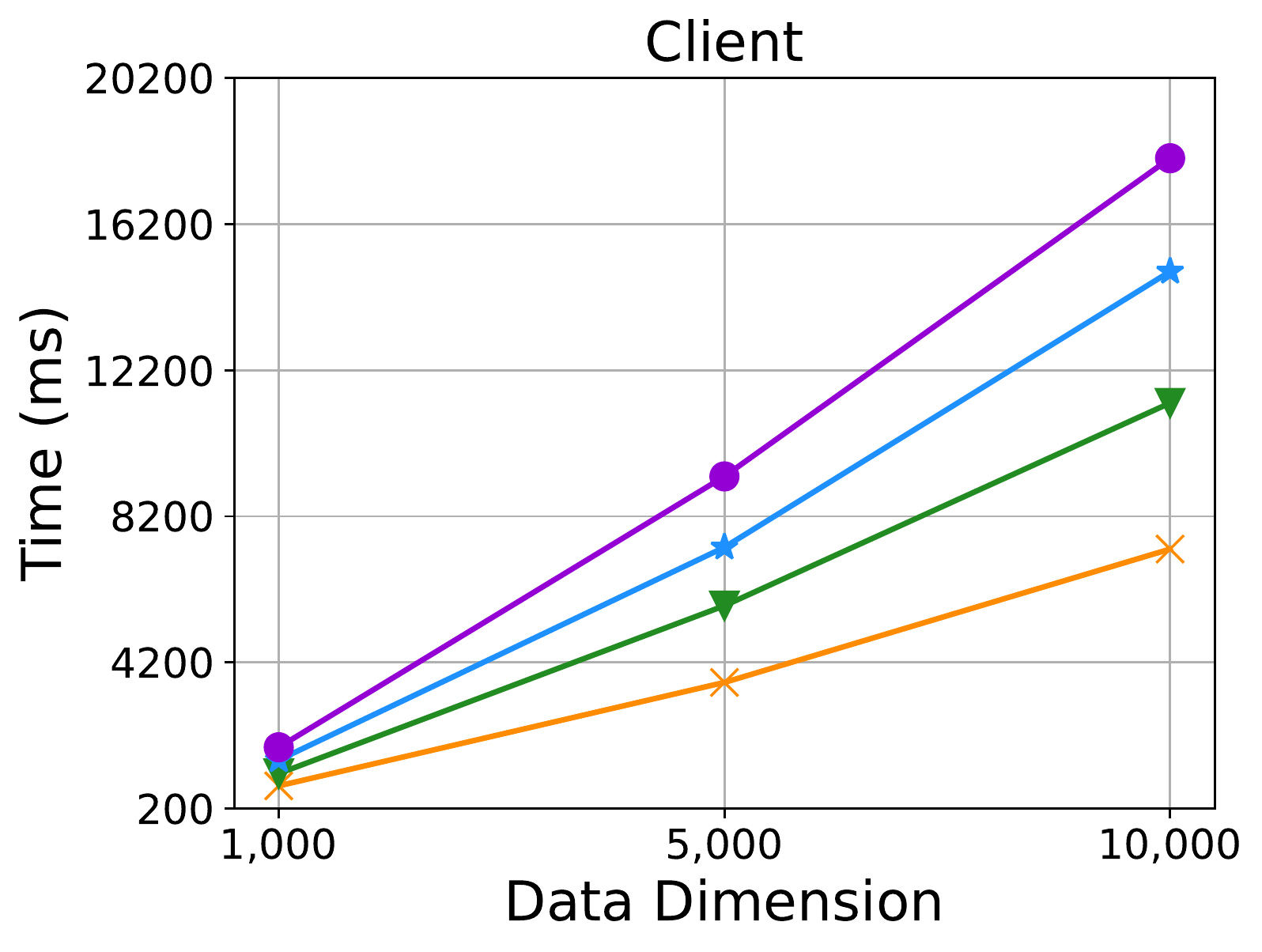}
   \vspace{-0.2cm}    
        \label{fig:comp:client:d}
    \end{subfigure}
                 \begin{subfigure}{0.24\linewidth}
    \centering \includegraphics[width=0.9\linewidth]{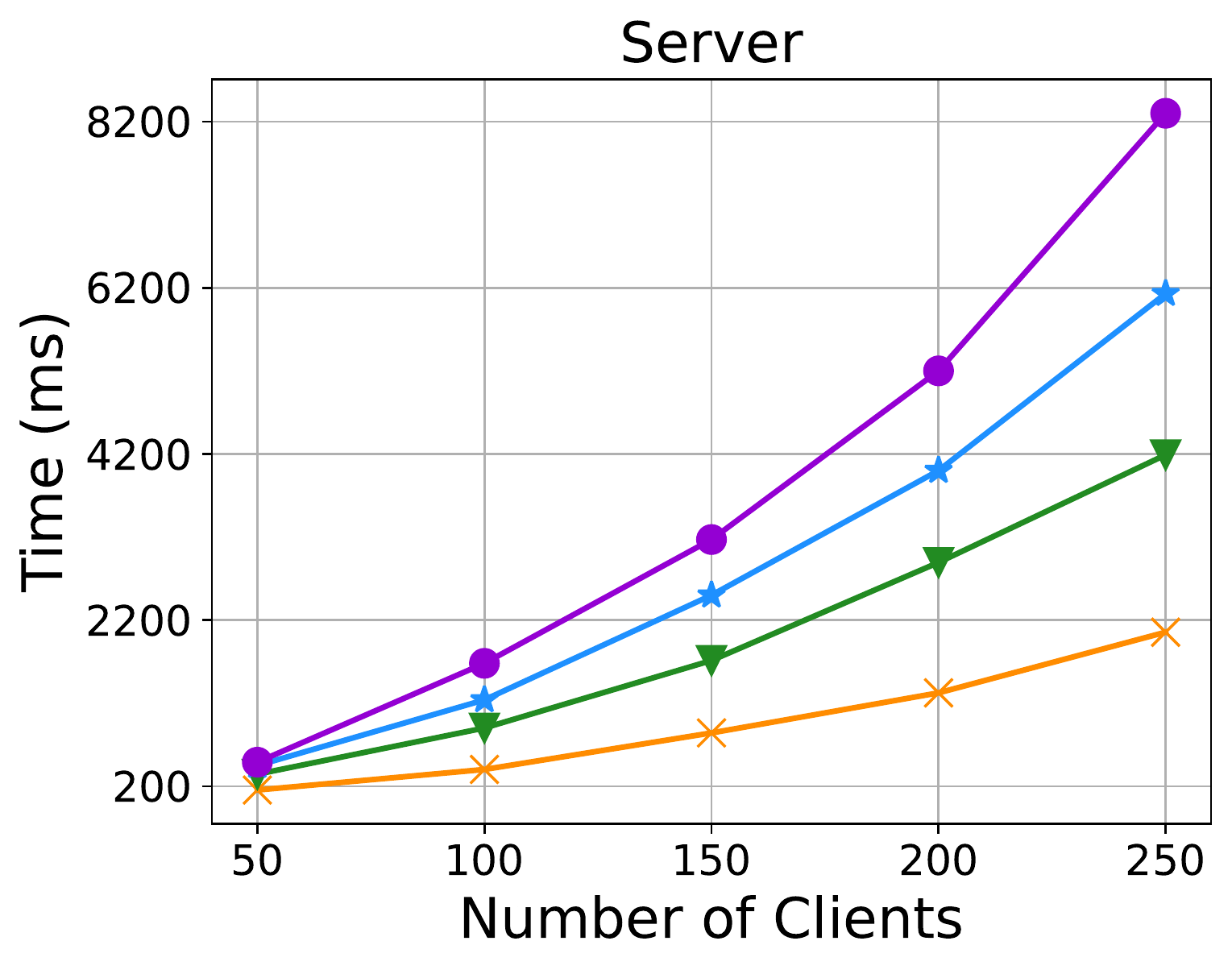}   \vspace{-0.2cm}
        \label{fig:comp:server:n}\end{subfigure}
         \begin{subfigure}{0.24\linewidth}
    \centering \includegraphics[width=0.9\linewidth]{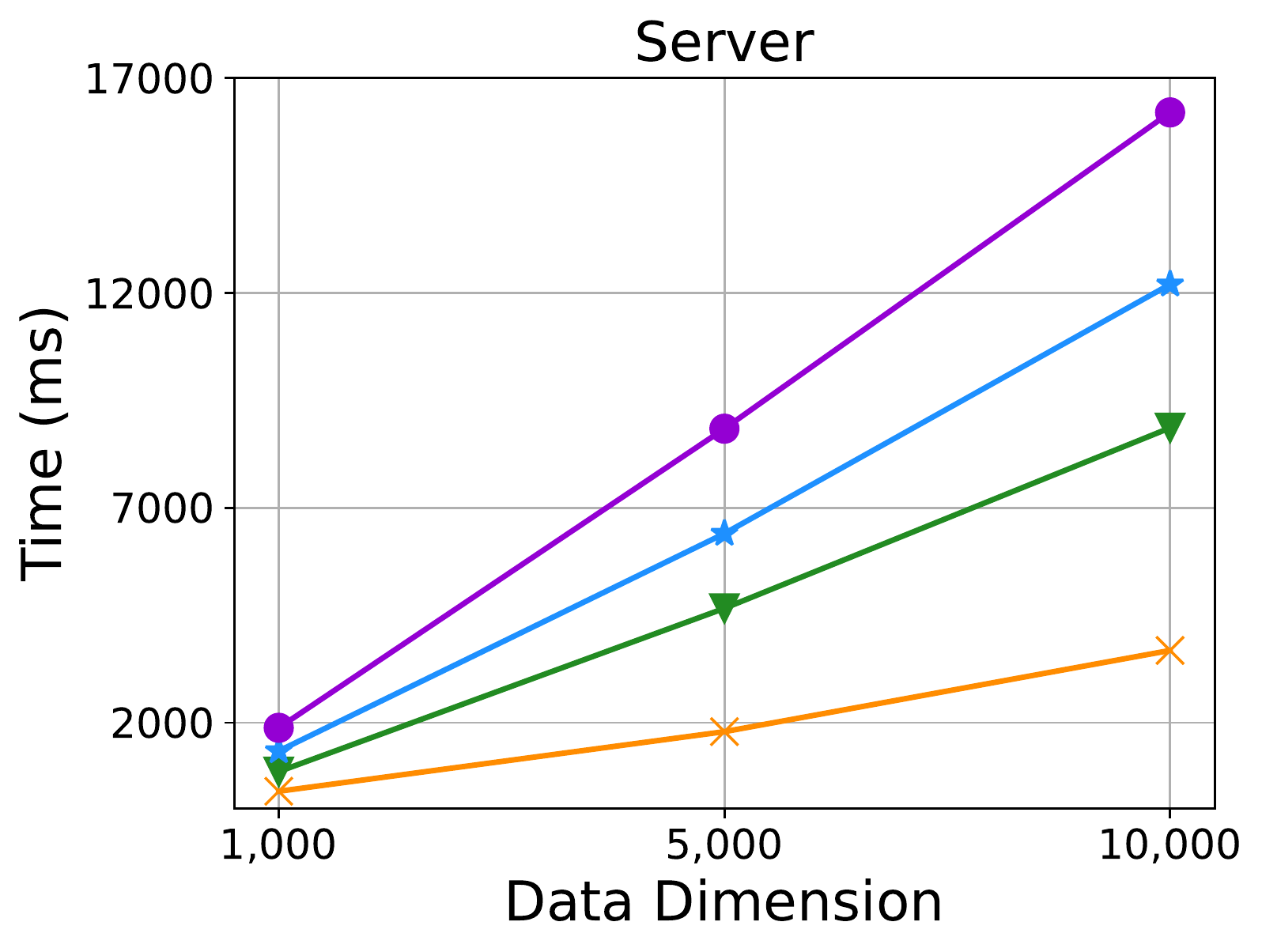}   \vspace{-0.2cm}
        \label{fig:comp:server:d}\end{subfigure}
\vspace{-0.2cm}
   \caption{Computation cost analysis of \name. The \textbf{left} two plots show the runtime of a single client client in milliseconds as a function of: (left)
 the number of clients $n$ and  (right) dimensionality of the updates $d$. The \textbf{right} two plots show the runtime of the server as a function of the same variables. The results demonstrate that performance decays quadratically in $n$, and linearly in $d$. }
   \label{fig:comp}\vspace{-0.3cm}
\end{figure*}
\begin{figure*}
  \begin{subfigure}{0.24\linewidth}
    \centering \includegraphics[width=0.9\linewidth]{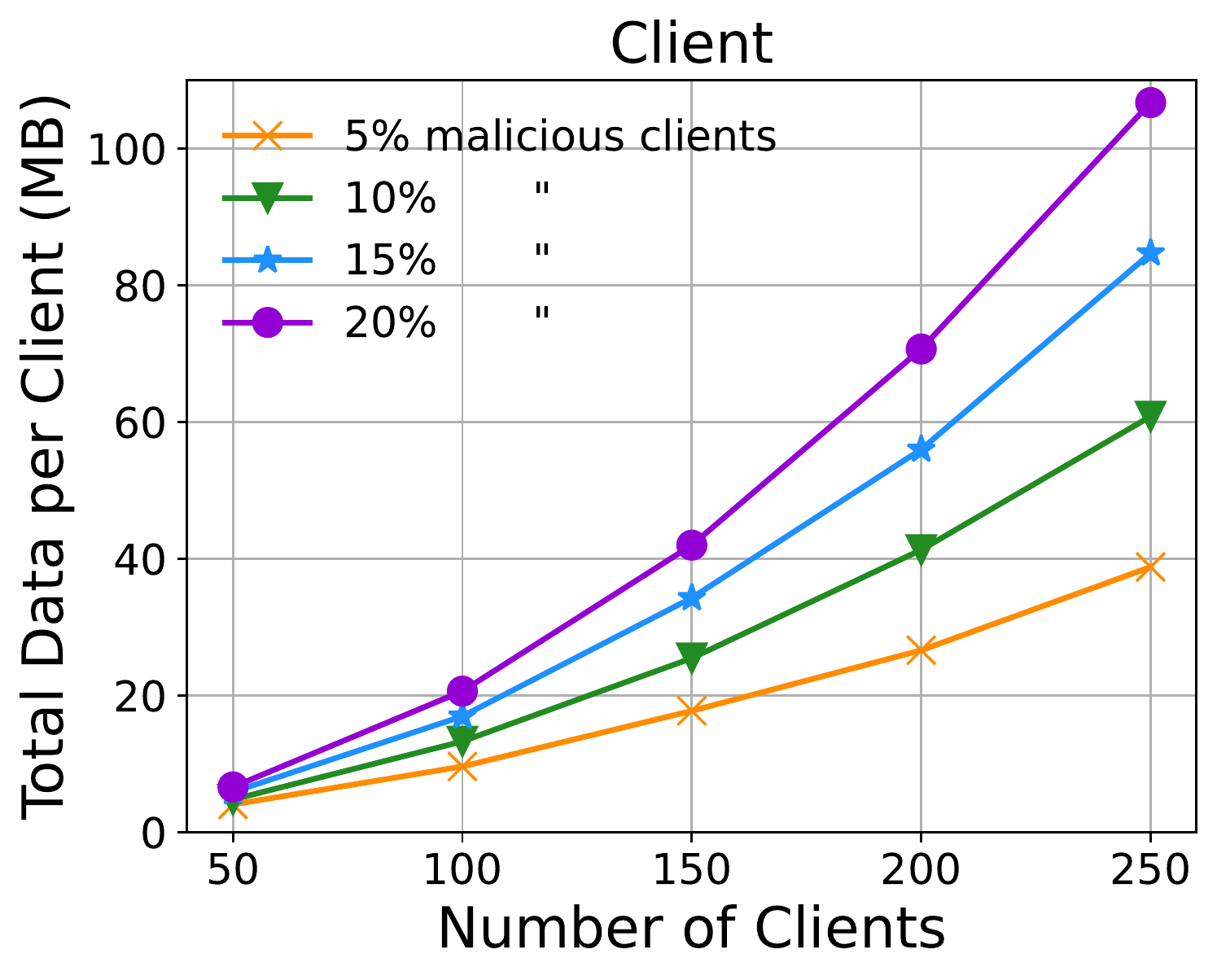}   
        \label{fig:BW:client:d}\end{subfigure}
    \begin{subfigure}{0.24\linewidth}
        \centering
         \includegraphics[width=0.9\linewidth]{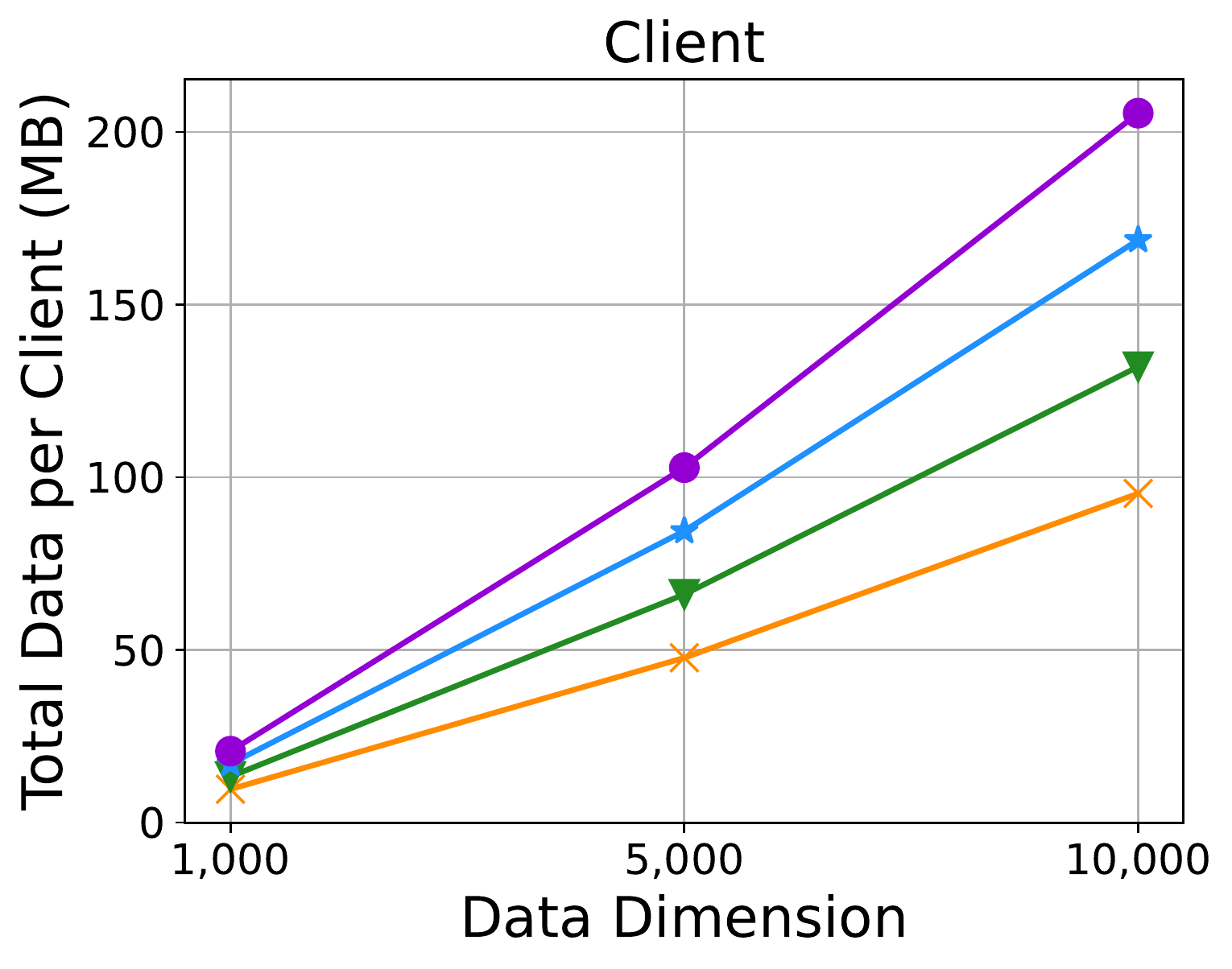}
        \label{fig:BW:client:n}
    \end{subfigure}
         \begin{subfigure}{0.24\linewidth}
    \centering \includegraphics[width=0.9\linewidth]{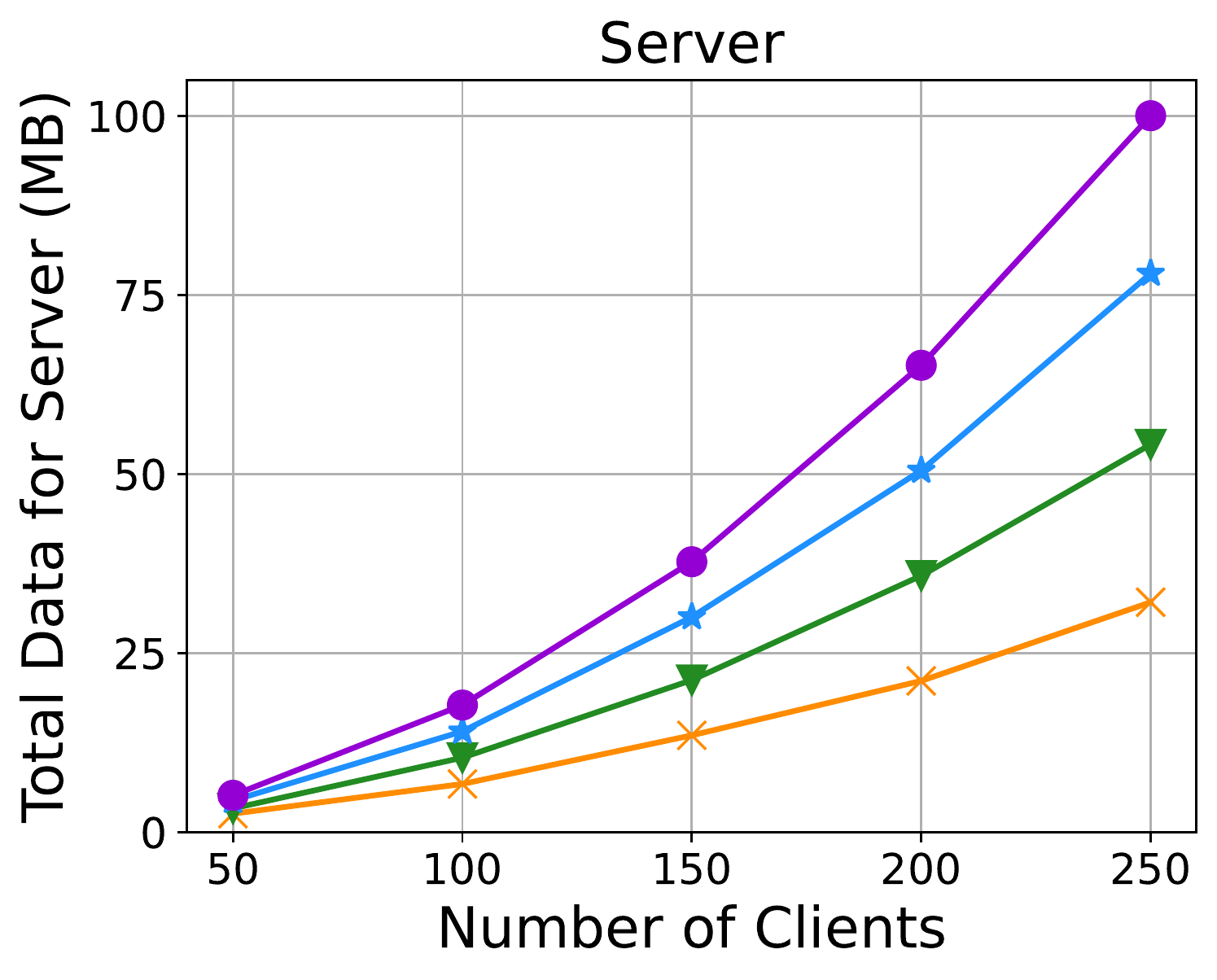}   
        \label{fig:BW:server:d}\end{subfigure}
         \begin{subfigure}{0.24\linewidth}
    \centering \includegraphics[width=0.9\linewidth]{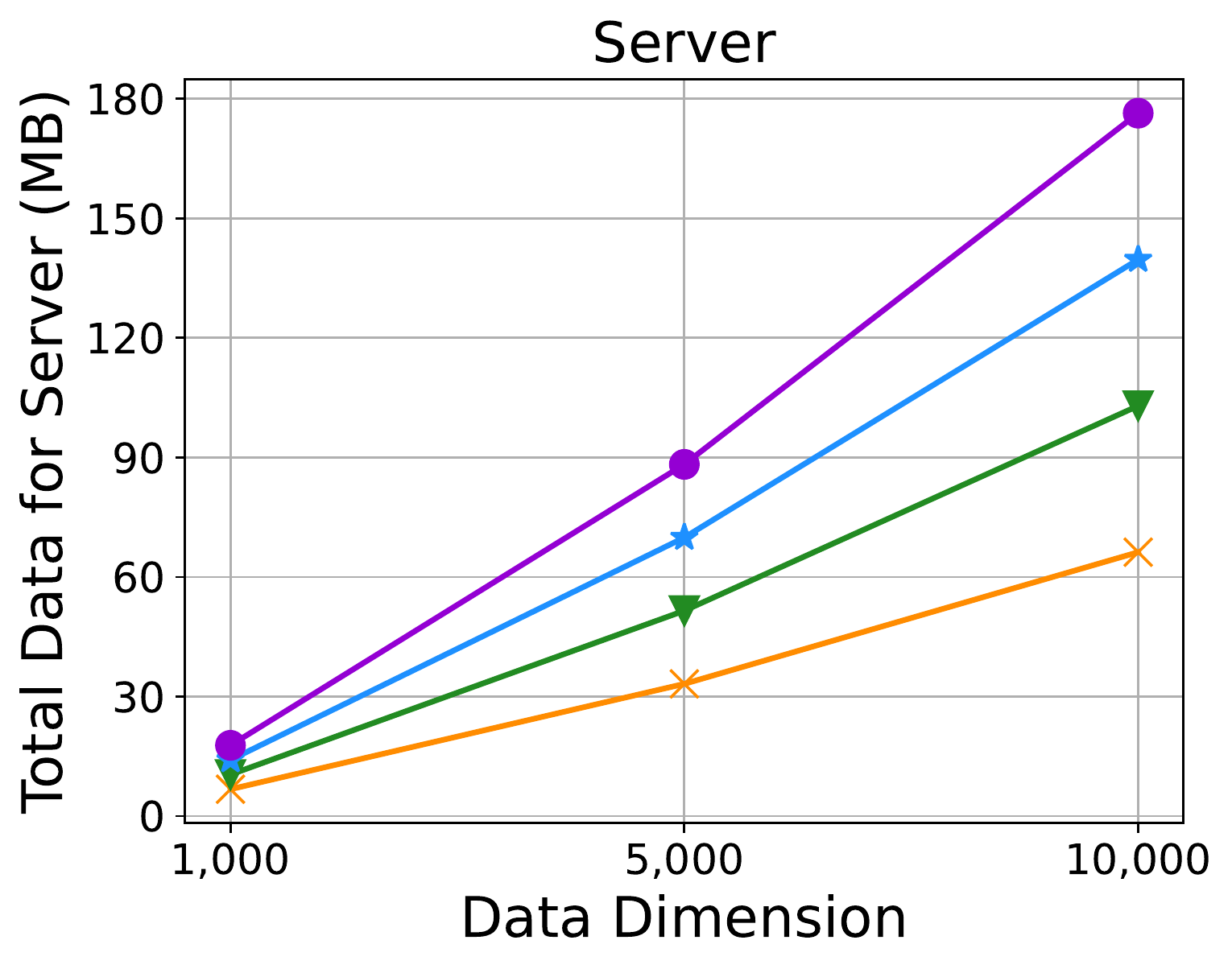}   
        \label{fig:BW:server:n}\end{subfigure}
   \vspace{-0.5cm}\caption{Communication cost analysis of \name. The \textbf{left} two plots show the amount of communication (in MB) for each client as a function of: (left)
 the number of clients $n$ and  (right) dimensionality of the updates $d$. The \textbf{right} two plots show the the amount of communication (in MB) for the server as a function of the same variables. The results show communication increases quadratically in $n$, and linearly in $d$. }
   \label{fig:BW}\vspace{-0.3cm}
\end{figure*}
\section{Experimental Evaluation}\label{sec:evaluation}


\subsection{Performance Evaluation.}\label{sec:evaluation:perf}
In this section, we analyze the performance of \name.\vspace{-0.2cm}\\\\
\textbf{Configuration.} 
We run experiments on two Amazon EC2 \textsf{c5.9large} instances with Intel Xeon Platinum 8000 processors. 
To emulate server-client communication, we use two instances in the US East (Ohio) and US West (Oregon) regions, with a round trip time of 21 ms. 
We implemented \name~in Python and  C++ using NTL library~\cite{NTL}. 
We use AES-GCM for encryption, a $56$-bit prime field $\Fl$ and probabilistic quantization~\cite{quant}. For key agreement, we use elliptic curve Diffie-Hellman~\cite{DH76} over the NIST P-256 curve. Unless otherwise specified, the default settings are $d\!=\!1K$, $n\!=\!100$, $m\!=\!10\%$ and $|\Vl(\cdot)|\approx4d$. We report the mean of 10 runs for each experiment. The rejection probability $(negl(\kappa))$ is dominated by  $\nicefrac{2\Ml-2}{|\Fl|}$
 (soundness error of SNIP, Sec. \ref{sec:system:block}).  $\Ml < 4d < 40K$ in our evaluation so the failure probability is of the order of $O(10^{-12})$.
\\\\\textbf{Computation Costs.}
Fig. \ref{fig:comp} presents \name's runtime.  We vary the number of malicious clients between $5\%$-$20\%$ of the number of clients. 
We observe that per-client runtime of \name~is low: it is $1.3s$ if $m=10\%$, $d=1K$, and $n=100$. 
The runtime scales quadratically in $n$ because a client has $O(mnd)$ computation complexity (see Table \ref{tab:complexity}) and $m$ is a linear function of $n$. As expected, the runtime increases linearly with $d$. A client takes around $11s$ when $d=10K$, $n=100$, and $m=10\%$. 
The runtime for the server is also low: the server completes its computation in about $1s$ for $n=100$, $d=1K$, and $m=10\%$. 
The server's runtime also scales quadratically in $n$  due to the $O(mnd)$ computation complexity (Table \ref{tab:complexity}). 
The runtime increases linearly with $d$.

\noindent In Fig. \ref{fig:breakdown}, we break down the runtime per round. 
We observe that: Round 1 (announcing public information) incurs negligible cost for both clients and the server; and Round 3 (verify proof) is the costliest round for both clients and the server where the dominating cost is verifying the validity of the shares (Sec. \ref{sec:complexity}). Note that the server has no runtime cost for Round 2 since the proof generation only involves clients.

Table \ref{tab:E2E} presents our end-to-end performance which contains the runtimes of a client, the server and the communication latencies.  For instance, the end-to-end runtime for \scalebox{0.9}{$n=100$, $d=1K$} and \scalebox{0.9}{$m=10\%$} is $\sim2.4s$. We also present the impact of one of our key optimizations -- eliminating the verification of the secrets shares of the proof -- which cuts down the costliest step in \name~and improves the performance by $2.3\times$. Additionally, we compare \name's performance with 
BREA~\cite{so2020byzantine}, which is a Byzantine-robust secure aggregator. \name~differs from BREA in two key ways: $(1)$ \name~is a general framework for per-client update integrity checks whereas BREA implements the multi-Krum aggregation algorithm~\cite{blanchard2017machine} that considers the entire dataset to determine the malicious updates (computes all the pairwise distances between the clients and then, detects the outliers), and $(2)$ BREA has an additional privacy leakage as it reveals the values of all the pairwise distances between clients. Nevertheless, we choose BREA as our baseline because, to the best of our knowledge, this is the only prior work that: $(1)$ detects and removes malformed updates, and $(2)$ works in the malicious threat model with $(3)$ a single server (see Table \ref{tab:compare_work}, Sec. \ref{sec:related_work}). We observe that \name~outperforms BREA and that the improvement increases with $n$. For instance, for $n\!=\!250$, \name~is $18.5\times$ more performant than BREA. This is due to BREA's complexity of \scalebox{0.9}{$O(n^3\log^2n\log\log n +mnd)$}, where the \scalebox{0.9}{$O(n^3)$} factor is due to each client partaking in the computation of the \scalebox{0.9}{$O(n^2)$} pairwise distances. 
\begin{table}\small
\centering
\resizebox{0.78\columnwidth}{!}{\begin{tabular}{lccc}
\toprule &  & \multicolumn{2}{c}{\bf Improvement over}
\\\cmidrule(lr){3-4} \textbf{\# Clients} ($n$) & \textbf{Time} (ms)&  Unoptimized \name & BREA~\cite{so2020byzantine}\\
 \midrule
50 & 1,072& 2.3$\times$& 2.5$\times$\\
100 &  2,367 & 2.3$\times$& 5.2$\times$  \\
150 & 4,326 & 2.3$\times$& 7.8$\times$ \\
200 &  6,996 & 2.3$\times$& 12.8$\times$ \\
250  & 10,389 & 2.3$\times$& 18.5$\times$ 
\\
   \bottomrule
\end{tabular}}
\caption{End-to-end time for a single iteration of \name~with $d\!=\!1000$ and $m\!=\!10\%$ malicious clients, as a function of the number of clients, $n$. We also compare it with a variant of \name~without optimizations, and with BREA~\cite{so2020byzantine}. }\label{tab:E2E}
\vspace{-0.9cm}\end{table}
\begin{figure}\centering \includegraphics[width=0.5\linewidth]{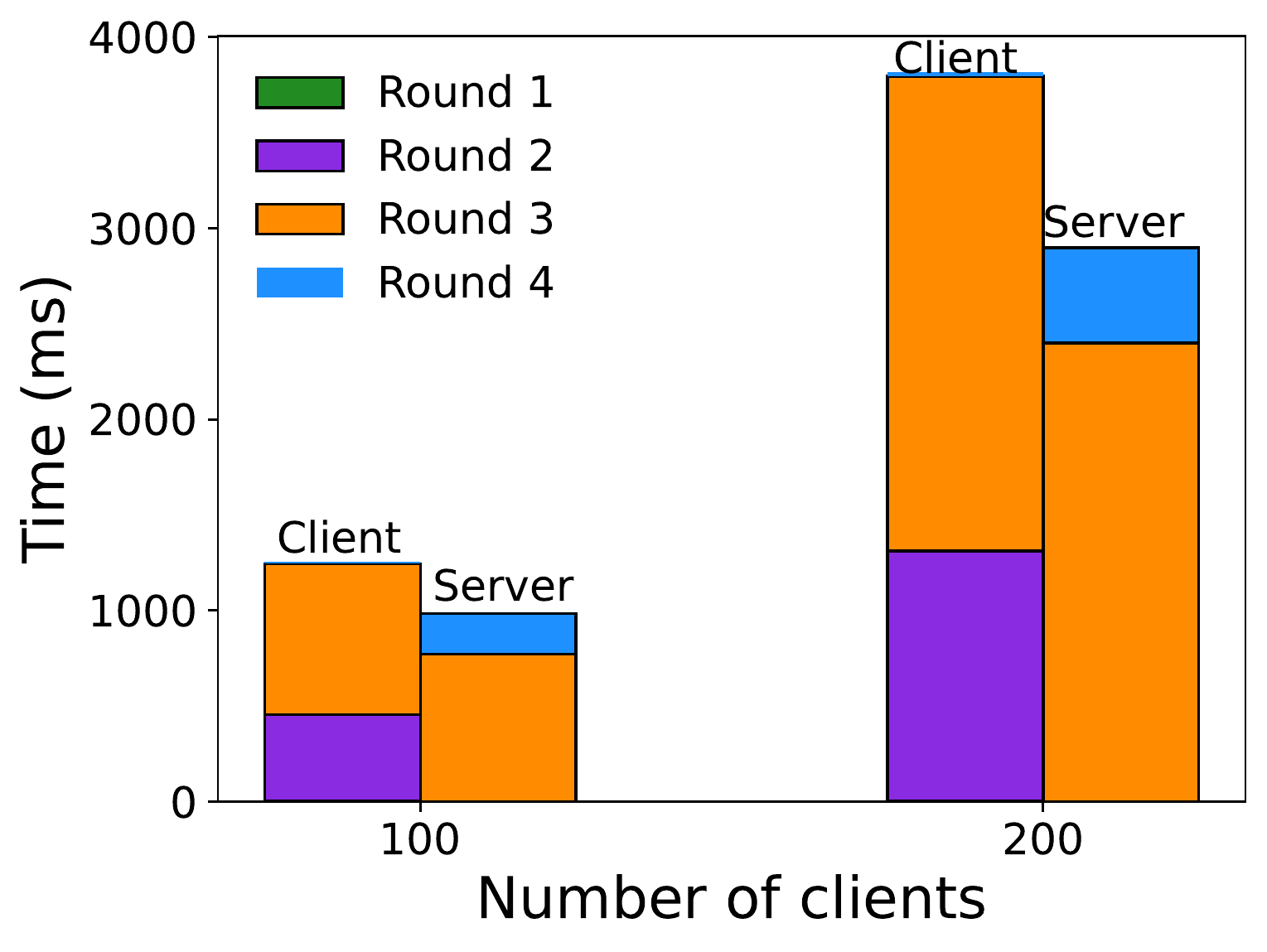}
\vspace{-0.4cm}\caption{Computation cost per round in \name.}\label{fig:breakdown} \vspace{-0.6cm}\end{figure}
\begin{figure*}
    \begin{subfigure}{0.24\linewidth}
        \centering
         \includegraphics[width=0.9\linewidth]{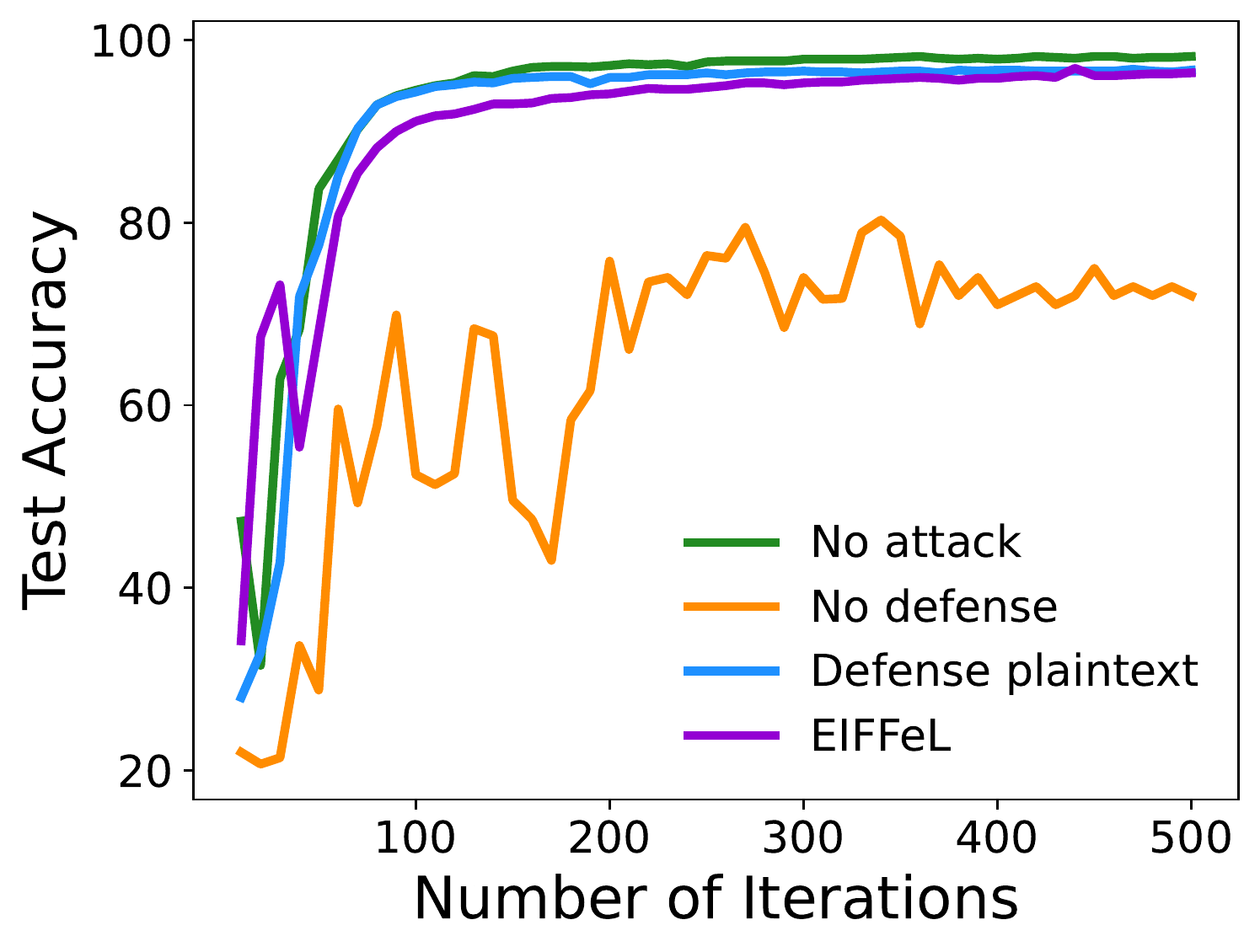}\vspace{-0.1cm}
       \caption{MNIST: Sign flip attack with norm ball validation predicate (defense).}
        \label{fig:MNIST:Ball}
    \end{subfigure}
    \begin{subfigure}{0.24\linewidth}
    \centering \includegraphics[width=0.9\linewidth]{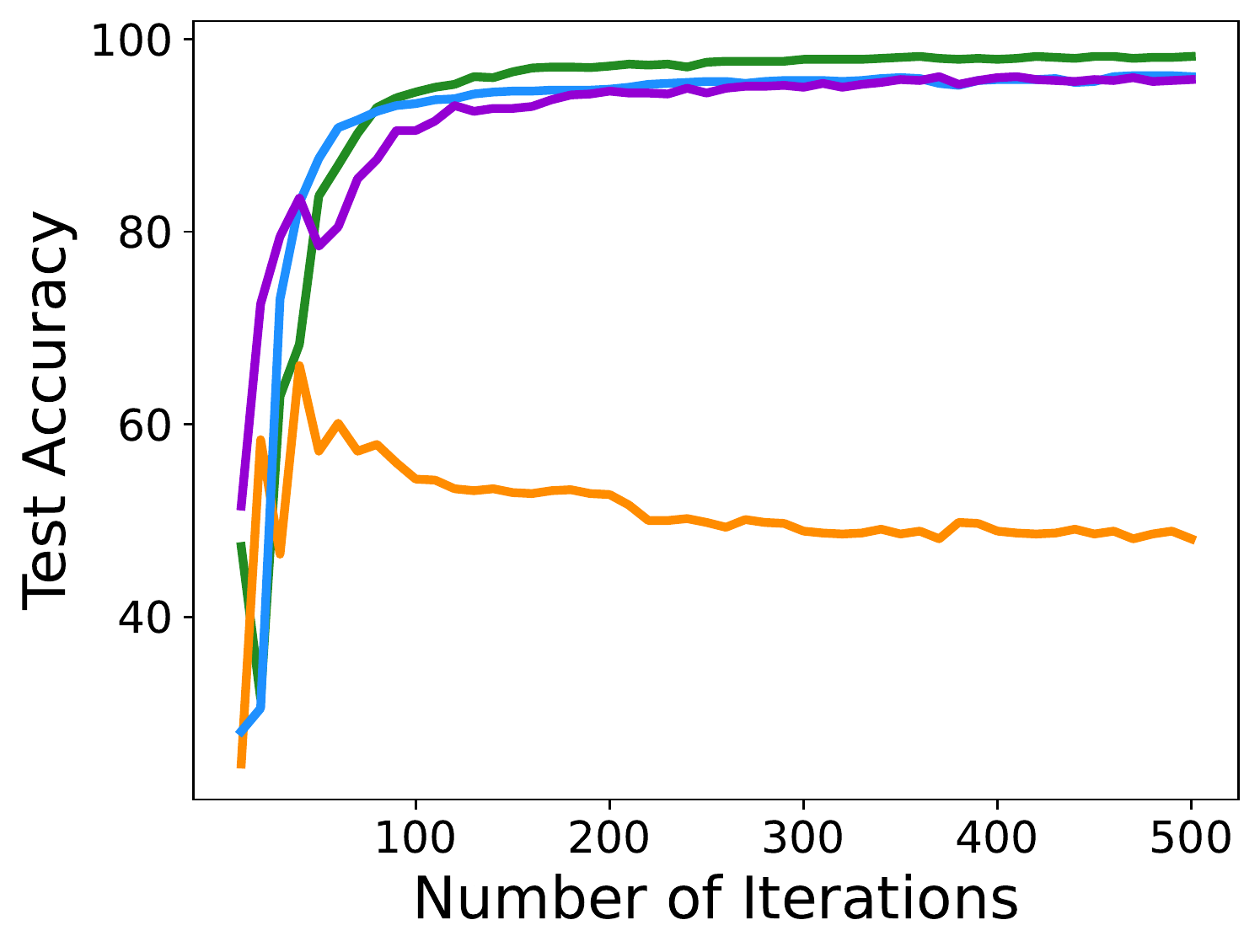}   \vspace{-0.1cm}
 \caption{MNIST: Scaling attack and cosine similarity validation predicate.}
        \label{fig:MNIST:Cosine}\end{subfigure}
         \begin{subfigure}{0.24\linewidth}
    \centering \includegraphics[width=0.9\linewidth]{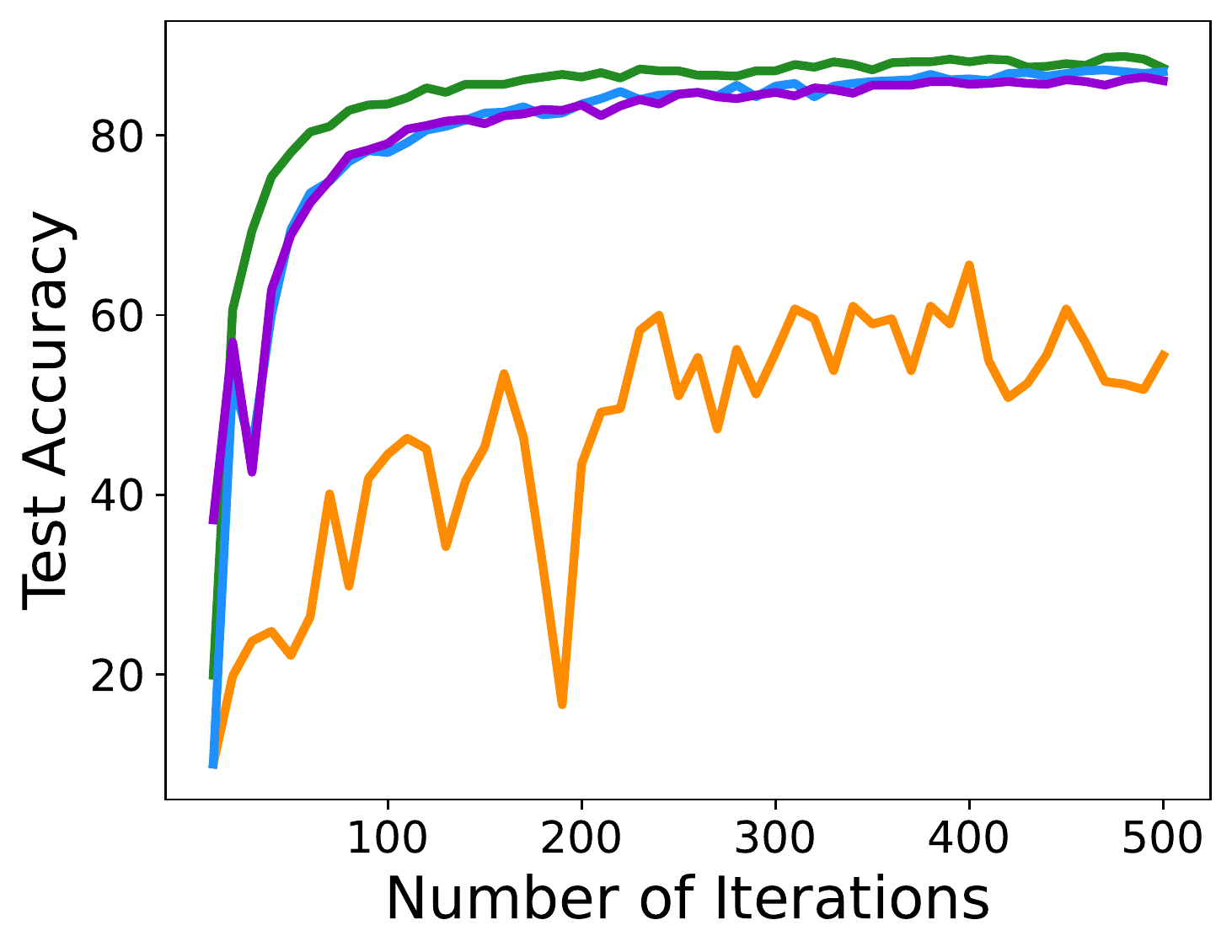}   \vspace{-0.1cm}
 \caption{FMNIST: Additive noise attack with Zeno++ validation predicate.}
        \label{fig:FMNIST:Zeno}\end{subfigure}
         \begin{subfigure}{0.24\linewidth}
    \centering \includegraphics[width=0.9\linewidth]{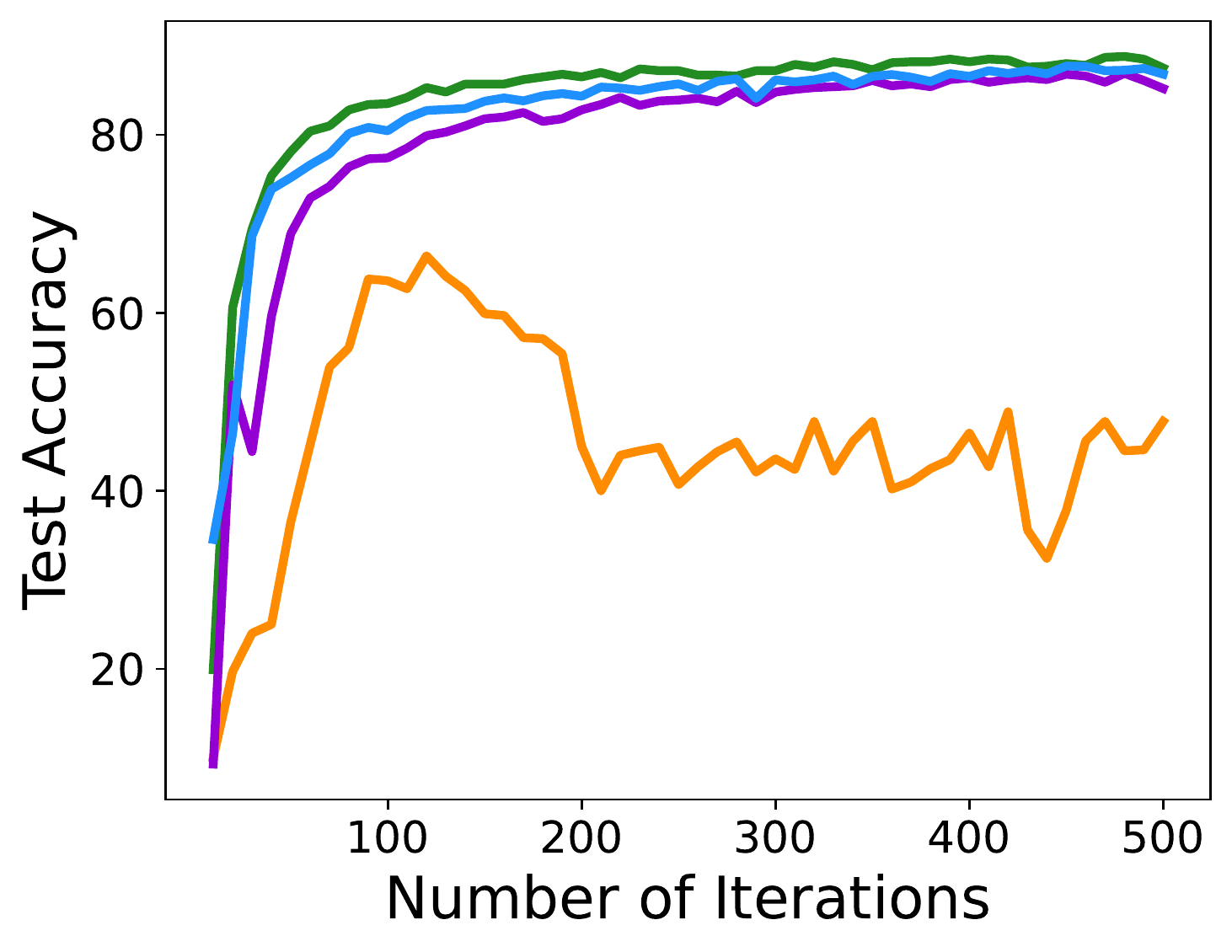}  \vspace{-0.1cm}
 \caption{FMNIST: Sign flip attack with norm ball validation predicate.}
        \label{fig:FMNIST:Ball}\end{subfigure}\\    \begin{subfigure}{0.24\linewidth}
        \centering
         \includegraphics[width=0.9\linewidth]{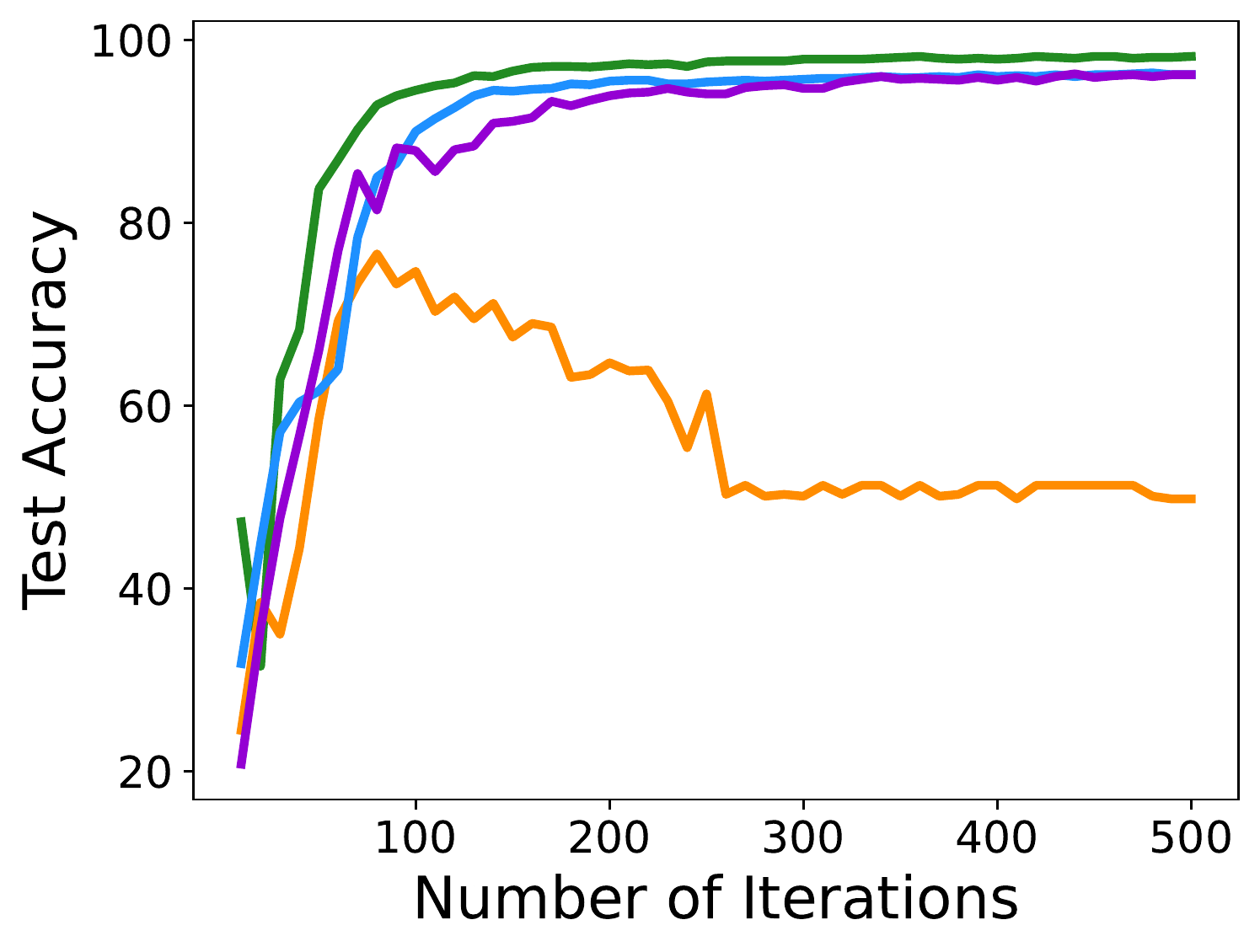}\vspace{-0.1cm}
       \caption{MNIST: Min-Max attack with Zeno++ validation predicate.}
        \label{fig:MNIST:Zeno}
    \end{subfigure}
         \begin{subfigure}{0.24\linewidth}
    \centering \includegraphics[width=0.9\linewidth]{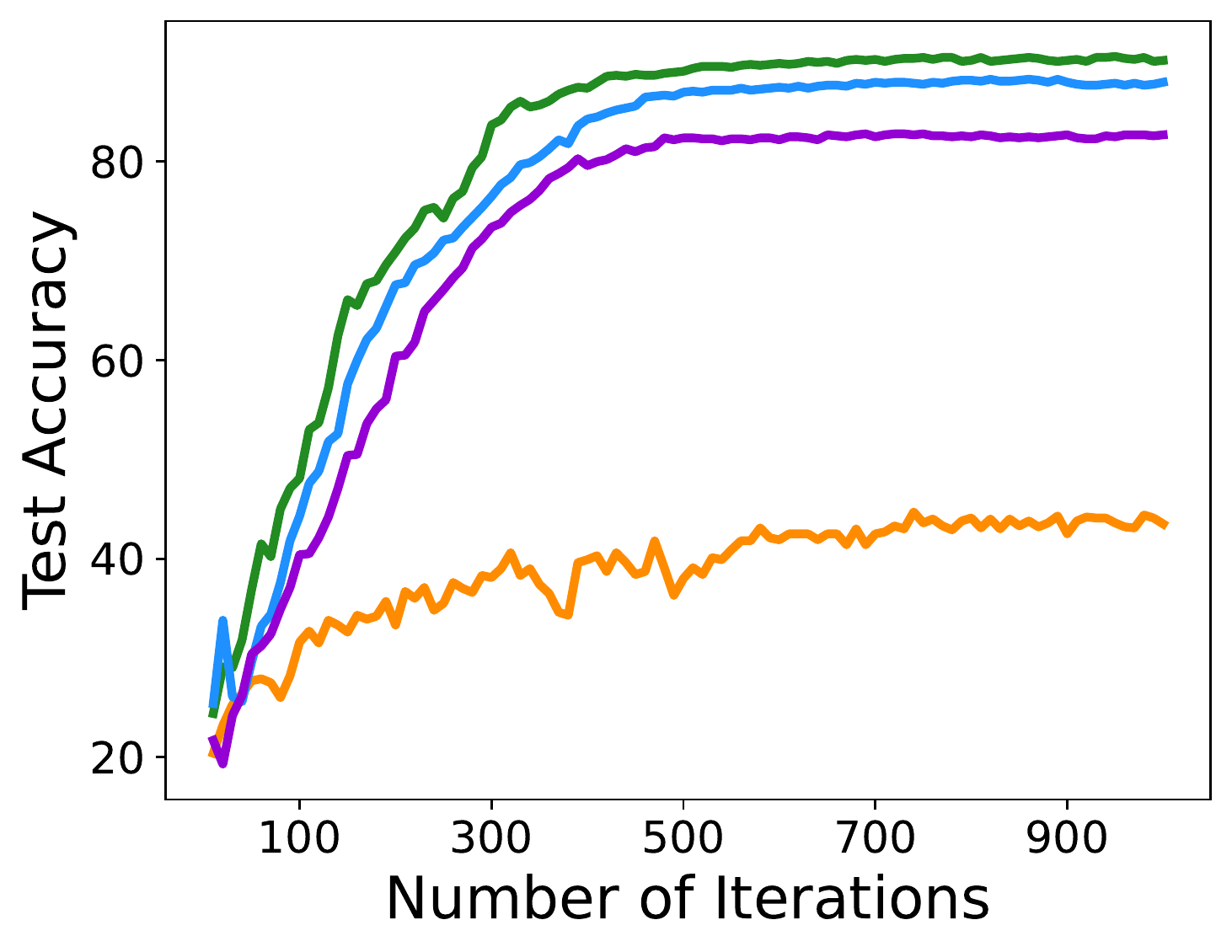}   \vspace{-0.1cm}
 \caption{CIFAR-10: Min-Sum attack with cosine similarity validation predicate.}
        \label{fig:CIFAR:Cosine}\end{subfigure}
          \begin{subfigure}{0.24\linewidth}
    \centering \includegraphics[width=0.9\linewidth]{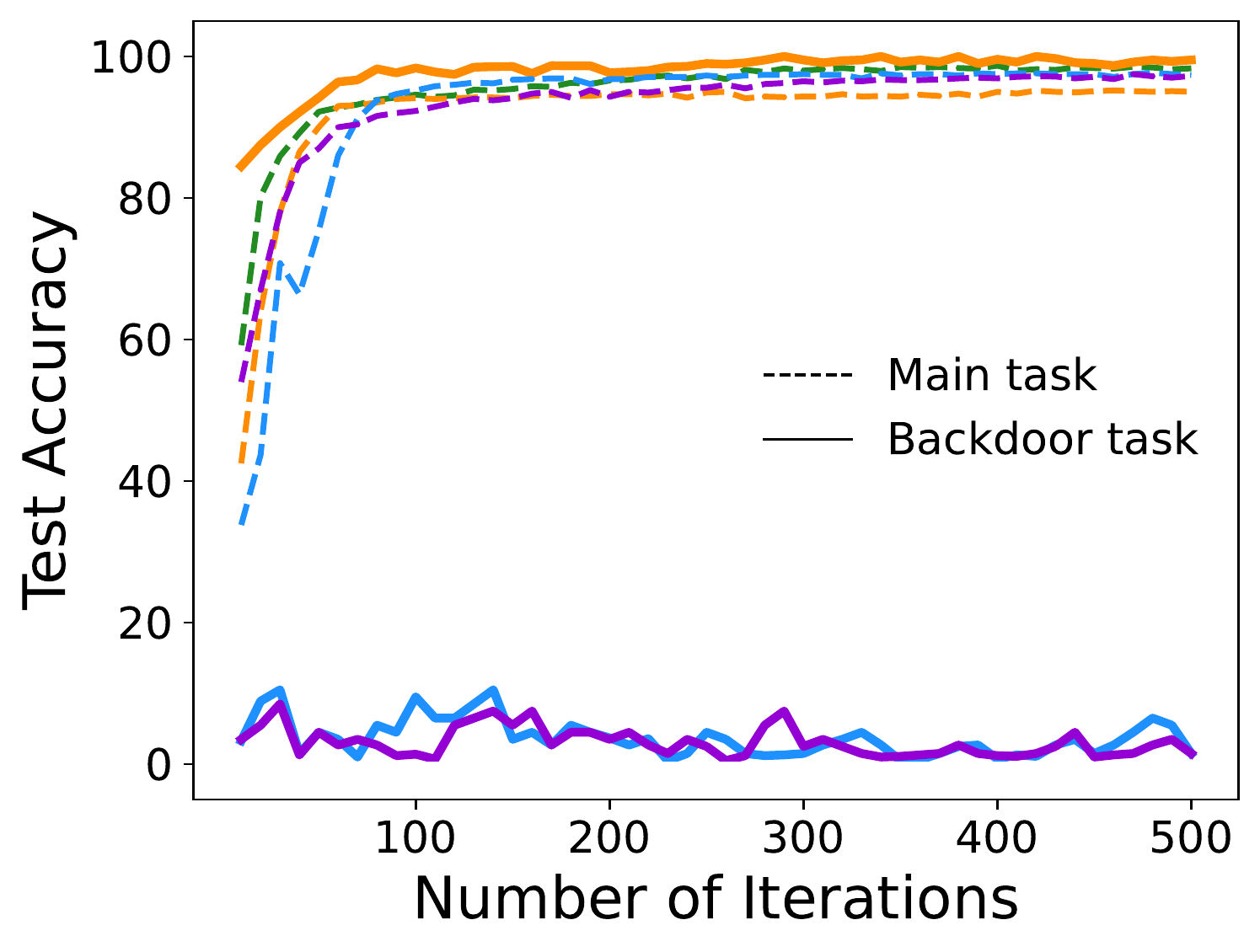}   \vspace{-0.1cm}
 \caption{EMNIST: Backdoor Attack-1 with norm bound validation predicate.}
        \label{fig:EMNIST:Norm}\end{subfigure}
         \begin{subfigure}{0.24\linewidth}
    \centering \includegraphics[width=0.9\linewidth]{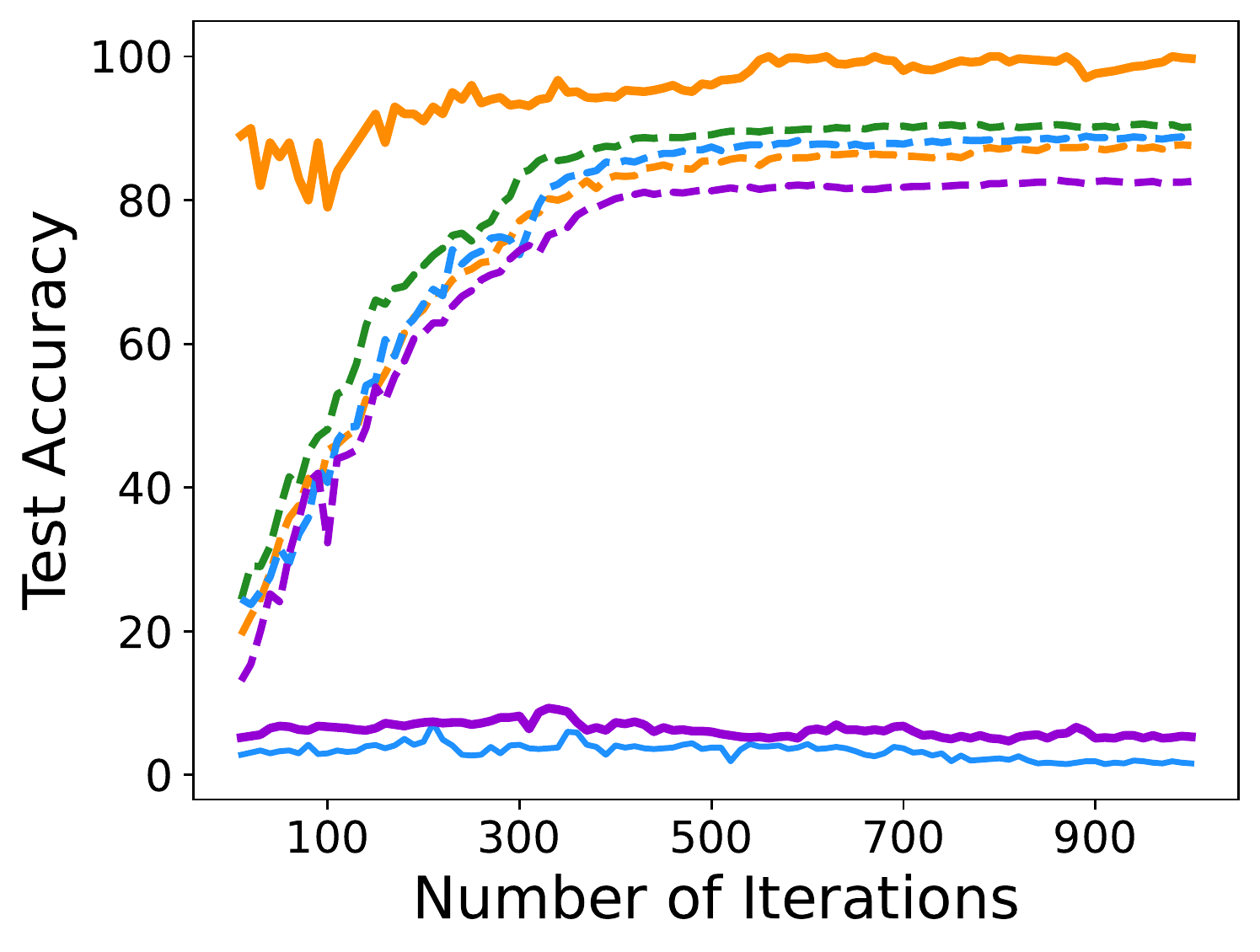}  \vspace{-0.1cm}
 \caption{CIFAR-10: Backdoor Attack-2 with norm bound validation predicate.}
        \label{fig:CIFAR:Norm2}\end{subfigure}
   \caption{Accuracy analysis of \name. Test accuracy is shown as a function of the FL iteration for different datasets and attacks.}
   \label{fig:model}\vspace{-0.3cm}
\end{figure*}
\vspace{-0.2cm}\\\\
\textbf{Communication Cost.} 
Fig. \ref{fig:BW} depicts the total data transferred by a client and the server. The communication complexity is \scalebox{0.9}{$O(mnd)$} for a single client and for the server. Hence, the total communication increases quadratically with $n$ and linearly with $d$, respectively. We observe that \name~has acceptable communication cost. For instance, the total data consumed by a client is $132$MB for the configuration \scalebox{0.9}{$n=100, d=10K, m=10\%$}. 
This is equivalent to streaming a full-HD video for $26s$~\cite{data}. Since most clients partake in FL training iterations infrequently, this communication is acceptable. \vspace{-0.2cm}\\\\\textbf{Note.} Recall, we assume the size of the validation predicate to be \scalebox{0.9}{$|\Vl| = O(d)$} since \scalebox{0.9}{$\Vl(\cdot)$} defines a function on the input which is \scalebox{0.9}{$d$}-dimensional. This assumption is validated by the state-of-the-art predicates tested in Sec. \ref{sec:eval:models}. The above experiments use \scalebox{0.9}{$|\Vl|\approx 4d$}. Hence, the overall complexity (App. \ref{sec:complexity}) is dominated by the \scalebox{0.9}{$O(mnd)$} term and does not depend on the validation predicate. 
\vspace{-0.1cm}\subsection{Integrity Guarantee Evaluation}\label{sec:eval:models}
In this section, we evaluate \name's efficacy in ensuring update integrity on real-world datasets.\\
\textbf{Datasets.} We evaluate \name~on three image datasets:
\squishlist\vspace{-0.2cm}
\item \textit{MNIST}~\cite{MNIST} is a digit classification dataset of $60K$ training images and $10K$ test images with ten classes. 
\item \textit{EMNIST}~\cite{emnist} is a writer-annotated handwritten digit classification dataset with $\sim 340K$ training and $\sim 40K$ testing images.
\item \textit{FMNIST}~\cite{FMNIST} is 
identical to MNIST in terms number of classes, and number of training and test images. \item \textit{CIFAR-10}~\cite{CIFAR10} contains RGB images with ten object classes. It has $50K$ training and $10K$ test images. 

\squishend\vspace{-0.2cm}
 \textbf{Models.} We test \name~with three classification models:
\squishlist \vspace{-0.2cm}
\item \textit{LeNet-5} \cite{Lecun1998} has five layers and $60K$ parameters, and  is used to experiment on MNIST and EMNIST.
\item For FMNIST, we use a five-layer convolutional network with $70K$ parameters and a similar architecture as LeNet-5.
\item We use \textit{ResNet-20} \cite{He2016DeepRL} with $20$ layers and $273K$ parameters for CIFAR-10.  
\squishend \vspace{-0.2cm}
\textbf{Validation Predicates.} 
To demonstrate the flexibility of \name, we evaluate four validations predicates, which represent the current \textit{state-of-the-art} defenses against data poisoning, as follows:
\squishlist\vspace{-0.3cm}
\item \textit{Norm Bound}~\cite{Sun2019CanYR}. This method checks whether the \scalebox{0.9}{$\ell_2$}-norm of a client update is bounded: \scalebox{0.9}{$
\Vl(u) = \mathbb{I}[ \lVert u\rVert_2 < \rho ]
$}
where \scalebox{0.9}{$\mathbb{I}[\cdot]$} is the indicator function and the threshold \scalebox{0.9}{$\rho$} is computed from the public dataset $\DP$.
\item \textit{Norm Ball}~\cite{steinhardt2017certified}. This method checks whether a client update is within a spherical radius from $v$ which is the gradient update computed from the clean public dataset $\DP$: \scalebox{0.9}{$\Vl(u)=\mathbb{I}\big[\lVert u-v\Vert_2\leq \rho \big]$}
where radius $\rho$ is also computed from $\DP$.
\item \textit{Zeno++}~\cite{Xie20} compares the client update with a loss gradient $v$ that is computed on the public dataset $\DP$: \scalebox{0.9}{$
\Vl(u) = 
\mathbb{I}[\gamma\langle v,u\rangle - \rho||u||_2$} \scalebox{0.9}{$\geq -\gamma\epsilon]
$}
where $\gamma$, $\rho$ and $\epsilon$ are threshold parameters also computed from  $\DP$ and $u$ is $\ell_2$-normalized to have the same norm as $v$.
\item \textit{Cosine Similarity}~\cite{Cao2021FLTrustBF,bagdasaryan2018backdoor}. This method compares the cosine similarity between the client update $u$ and the global model update of the last iteration $u'$: \scalebox{0.9}{$\Vl(u)=\mathbb{I}\big[\frac{\langle u,u'\rangle}{\lVert u\rVert_2 \lVert u'\rVert_2}< \rho \big]$} where $\rho$ is computed from $\DP$ and $u$ is $\ell_2$-normalized to match norm of $u'$.
\vspace{-0.2cm}\squishend

\textbf{Poisoning Attacks.}
To test the efficacy of \name's implementations of the four validation predicates introduced above, we test it against seven poisoning attacks:
\vspace{-0.2cm} 
\squishlist
\item \textit{Sign Flip Attack}~\cite{Damaskinos2018AsynchronousBM}. In this attack, the malicious clients flip the sign of their local update: \scalebox{0.9}{$\hat{u}=-c\cdot u, c\in \mathbb{R}_+$}.
\item \textit{Scaling Attack}~\cite{bhagoji2019analyzing} scales a local update to increase its influence on the global update: \scalebox{0.9}{$\hat{u}=c\cdot u, c \in \mathbb{R}_+$}.
\item \textit{Additive Noise Attack}~\cite{Li2019RSA} adds Gaussian noise to the local update: \scalebox{0.9}{$\hat{u}=u+\eta, \eta \sim \mathcal{N}(\sigma,\mu)$}.
\item \textit{Min-Max Attack}~\cite{Shejwalkar21} sets all the malicious updates to be: \scalebox{0.9}{$\mathrm{argmax}_{\gamma}$} \scalebox{0.9}{$\max_{i \in [n]}||\hat{u}-u_i||_2\leq \max_{i,j\in [n]}||u_i-u_j||_2; \hat{u}=\frac{1}{n}\sum_{i=1}^n u_i + \gamma\cdot u^p$}, \\where $u^p$ is a dataset optimized perturbation vector. Here, the adversary is assumed to have access to the benign (well-formed) updates of \textit{all} clients.  This attack finds the malicious gradient
whose maximum distance from a benign gradient is less than
the maximum distance between any two benign gradient. 
\item \textit{Min-Sum Attack}~\cite{Shejwalkar21} sets all the malicious updates to be: \scalebox{0.9}{$\mathrm{argmax}_{\gamma}$} \scalebox{0.9}{$\sum_{i \in [n]}||\hat{u}-u_i||_2\leq \max_{i\in [n]}\sum_{j\in[n]}||u_i-u_j||_2; \hat{u}=\frac{1}{n}\sum_{i=1}^n u_i + \gamma\cdot u^p$}, where $u^p$ is a dataset optimized perturbation vector. Here, the adversary is assumed to have access to the benign updates of \textit{all} clients.  
This attack finds the malicious gradient such that the
sum of its distances from all the other gradients is less than the sum of distances of any benign gradient from other benign
gradients. 
\item \textit{Backdoor Attack-1}~\cite{Sun2019CanYR} classifies the digit seven as the digit one for EMNIST.  
\item \textit{Backdoor Attack-2}~\cite{bagdasaryan2018backdoor} classifies images of green cars as birds for CIFAR-10.
\squishend
\textbf{Configuration.} We use the same  configuration as before. We implement the image-classification models in PyTorch. We randomly select 10K samples from each training set as the public dataset $\DP$ and train on the remaining samples. EMNIST is collected from 3383 clients with $\sim100$ images per client.
For all other datasets, the training set is divided into 5K subsets to create the local dataset for each client. For each training iteration, we sample the required number of data subsets out of these 5K subsets. 

\noindent\textbf{Results.}
Fig. \ref{fig:model} shows the accuracy of different image-classification models in \name. We set $n=100$ and $m=10\%$, and use random projection to project the updates to a dimension $d$ of 1K (MNIST, EMNIST), 5K (FMNIST), or 10K (CIFAR-10). For the two backdoor attacks, we consider $m=5\%$. Our experiment assesses how the random projection affects the efficacy of the integrity checks.  We observe that for MNIST (Figs. \ref{fig:MNIST:Ball}, \ref{fig:MNIST:Cosine} and \ref{fig:MNIST:Zeno}), EMNIST (Fig. \ref{fig:EMNIST:Norm}) and FMNIST (Fig. \ref{fig:FMNIST:Zeno} and \ref{fig:FMNIST:Ball}), \name~achieves performance comparable to a baseline that applies the defense (validation predicate) on the plaintext. 
In most cases, the defenses retain their efficacy even after random projection. This is because they rely on computing inner products and norms of the update; these operations preserve their relative values after the projection with high probability~\cite{WH}.
We observe a drop in accuracy ($\sim7\%$) on CIFAR-10 (Figs. \ref{fig:CIFAR:Cosine} and \ref{fig:CIFAR:Norm2}) as updates for ResNet-20 with 273K parameters are projected to 10K.
The end-to-end per-iteration time $(m=10\%)$ for MNIST, EMNIST, FMNIST, and CIFAR-10 is $2.4s$ (Table \ref{tab:E2E}), $2.4s$, $10.7s$, and $20.5s$, respectively. 
The associated communication costs for the client are $13.3$MB,  $13.3$MB, $65.8$MB, and $132$MB (Fig. \ref{fig:BW}). Additional evaluation results are presented in Fig. \ref{fig:app:eval} (App. \ref{app:eval}).

\section{Discussion} \label{sec:discussion}
In this section, we discuss possible avenues for future research (additional discussion in App. \ref{app:discussion}).

\textbf{Handling Higher Fraction of Malicious Clients.} For \scalebox{0.9}{$\lfloor \frac{n-1}{3}\rfloor < m$} \scalebox{0.9}{$< \lfloor \frac{n-1}{2} \rfloor$} (honest majority), the current implementation of \name~can detect but not remove malformed inputs (Gao's decoding algorithm returns \scalebox{0.9}{$\bot$} if  \scalebox{0.9}{$m>\lfloor \frac{n-1}{3}\rfloor$}). Robust reconstruction in this case could be done via Guruswami-Sudan list decoder \cite{GS}. We do not do so in \name~because the reconstruction might fail  sometimes.

\noindent \textbf{Handling Client Dropouts.}  In practice, clients might have only sporadic access to connectivity and so, the protocol must be robust to clients dropping out. \name~can already accommodate malicious client dropping out -- it is straightforward to extend this for the case of honest clients as well. 

\noindent\textbf{Identifying All Malicious Clients.} Currently, \name~identifies a partial list of malicious clients. To detect all malicious clients, one can use: $(1)$ \textsf{PVSS} to identify all clients who have submitted at least one invalid share, and $(2)$ decoding algorithms such as Berlekamp-Welch~\cite{Blahut1983} that can detect the location of the errors from the reconstruction.  We do not use them in \name~as they have higher computation cost. 

\noindent\textbf{Reducing Client's Computation.} Currently, verifying the validity of the secret shares is the dominant cost for clients. This task can be offloaded to the server $\Ser$ by using a publicly verifiable secret sharing scheme (\textsf{PVSS}) \cite{Schoenmakers99,Stadler96,Tang} where the validity of a secret share can be verified by any party. 
 However, typically \textsf{PVSS} employs public key cryptography (which is costlier than symmetric cryptography) which might increase the end-to-end running time.  
 
\noindent \textbf{Additional Defense Strategies.} \name~supports any defense strategy that can be expressed as a per-update anomaly detection mechanism (captured via the public validation predicate $\Vl(\cdot)$). A recent line of work~\cite{Jia2021,Levine22,panda22,Rosenfeld20,Wang20,Xie21} proposes a complimentary style of defense which involves inspecting the final aggregate and the resulting model. For instance, Sparsefed~\cite{panda22} is a state-of-the-art backdoor defense where the server selects the top $k$ dimensions of the final aggregate and updates the model only along those dimensions (others are set to zero). CRFL~\cite{Xie21} provides certified robustness against backdoor attacks by clipping and perturbing the final aggregate and performing parameter smoothing on the global model during testing.  Jia et al.~\cite{Jia2021} propose an ensemble learning mechanism where the server learns multiple global models on randomly selected subset of clients and takes majority vote among the global models for test-time prediction. Such defenses can be immediately integrated with \name~since the server has access to the final aggregate and the updated global model in the clear.

\noindent \textbf{Towards poly-logarithmic complexity.}  Currently,  dominant term in the complexity is $O(mnd)$ which results in a $O(n^2)$ dependence on $n$ (since we consider $m$ is a fraction of $n$). This can be reduced to $O(n\log^2 nd)$ by using the techniques from \cite{Bell2020}. A detailed discussion is presented in App. \ref{app:discussion}. 


\section{Related Work}
\label{sec:related_work}

\setlength{\textfloatsep}{2pt}
\begin{table}[ht]
\centering
\caption{Comparison of \name~with Related Work}\vspace{-0.3cm}
\resizebox{\columnwidth}{!}{
\begin{tabular}{lccccc}
\toprule
    \multirow{2}{*}{Work} & Malicious & Single & Removes  & Arbritrary \\
   & Threat Model & Server & Malformed Inputs & Integrity Checks\\
\midrule
He et.al~\cite{he2020secure} & \textcolor{red}{$\times$} &  \textcolor{red}{$\times$} & \textcolor{red}{$\times$} &  \textcolor{red}{$\times$} \\
FLGuard~\cite{nguyen2021flguard} & \textcolor{red}{$\times$} & \textcolor{red}{$\times$} & \textcolor{red}{$\times$} & \textcolor{red}{$\times$}  \\
RoFL~\cite{burkhalter2021rofl}  & \textcolor{red}{$\times$}  & \textcolor{green}{\checkmark} & \textcolor{red}{$\times$} & \textcolor{red}{$\times$} \\
BREA*~\cite{so2020byzantine} & \textcolor{green}{\checkmark} & \textcolor{green}{\checkmark} & \textcolor{green}{\checkmark} & \textcolor{red}{$\times$} \\
\rowcolor{aliceblue} \name (Our) & \textcolor{green}{\checkmark} & \textcolor{green}{\checkmark} & \textcolor{green}{\checkmark} & \textcolor{green}{\checkmark} \\
\bottomrule
\end{tabular}
}
\footnotesize
    {*Has additional privacy leakage}
\label{tab:compare_work}
\end{table}

\textbf{Secure Aggregation.} Prior work has addressed the problem of (non-Byzantine) secure aggregation in FL~\cite{bonawitz2017practical,Bell2020,Dream,so2021turboaggregate}. A popular approach is to use pairwise random masking to protect the local updates~\cite{bonawitz2017practical,Dream}. Advancements have been made in the communication overhead \cite{Konecn2016FederatedLS, Bonawitz2019FederatedLW, so2021turboaggregate}.

\noindent\textbf{Robust Machine Learning.}
A large number of studies have explored methods to make machine learners robust to Byzantine failures \cite{bagdasaryan2018backdoor,bhagoji2019analyzing,kairouz2019advances}.
Many of these robust machine-learning methods require the learned to have full access to the training data or to fully control the training process~\cite{cretu2008casting,gu2017badnets,liu2018finepruning,shen2019learning,steinhardt2017certified,wang2019neural} which is infeasible in FL.
Another line of work has focused on the development of estimators that are inherently robust to Byzantine errors~\cite{blanchard2017machine,chen2018draco,pan2020justin,rajput2019detox,yin2019byzantine}.
In our work, we target a set of methods that provides robustness by checking per-client updates ~\cite{blanchard2017machine,fung2018mitigating,shen2016auror}.

\noindent \textbf{Verifying Data Integrity in Secure Aggregation.}
Table \ref{tab:compare_work} compares \name~with prior work. There are three key differences between RoFL~\cite{burkhalter2021rofl} and \name{}: $(1)$ RoFL is designed only for range checks with $\ell_2$ or $\ell_{\infty}$ norms. Specifically, RoFL uses Bulletproofs which is especially performant for range proofs (range proofs can be aggregated where one can prove that $n$ commitments lie within a given range by providing only an additive $O(log(n))$ group elements over the length of a single proof). RoFL's performance is primarily based on this aspect of Bulletproof and all of its optimizations work only for range proofs. As such RoFL cannot support any other checks with the same performance as currently reported in the paper. By contrast, \name{} is a general framework that supports arbitrary validation predicates with good performance. 
$(2)$ RoFL is susceptible to DoS attacks because it \textit{only} detects malformed updates and aborts if it finds one. Specifically, the recovery of the final aggregate in RoFL requires a step of nonce cancellation that involves all the inputs by design. Hence, even if one of the input is invalid, the final aggregate will be ill-formed. By contrast, \name~is a \SAIV~protocol that  detects and removes malformed updates in every round.
$(3)$ RoFL assumes an honest-but-curious server, whereas \name~considers a  malicious threat model.  
BREA~\cite{so2020byzantine} also removes  outlying updates but, unlike \name, it leaks pairwise distances between inputs. 
Alternative solutions~\cite{nguyen2021flguard,he2020secure} for distance-based Byzantine-robust
aggregation uses two non-colluding servers in the semi-honest threat model, which is incompatible with FL. 

\section{Conclusion}\label{sec:conclusion}
Practical FL settings need to ensure both the privacy and integrity of model updates provided by the clients.
In this paper, we have formalized these goals in a new protocol, \SAIV, that securely aggregates \textit{only} well-formed inputs (\emph{i.e.}, updates). 
To demonstrate the feasibility of \SAIV, we have proposed \name: a system that efficiently instantiates a \SAIV~protocol. 



\bibliographystyle{plain}
\bibliography{references}
\clearpage

\section{Appendix}
\begin{table}\small
\caption{Notations}
\centering
\resizebox{0.95\columnwidth}{!}{\begin{tabular}{cl}
\toprule Symbol & Description\\
 \midrule
$n$ & Total number of clients\\
$m$ & Number of malicious clients\\
$\mathcal{S}$ & Server\\
 $\Cl_i$ & $i$-th client \\
$D_i$ & Private dataset of $\Cl_i$\\
$\Cl$ & Set of all $n$ clients\\
$\Cl_H$ & Set of $n-m$ honest clients\\
$\Cl_M$ & Set of $m$ malicious clients\\
$\Vl(\cdot)$ & Validation predicate\\
$\mathcal{M}$ & Global model to be trained\\
$u_i$ & Local update (gradient) of client $\Cl_i$\\
$\mathcal{U}$ & Aggregate update\\
$\Cl_{\Vl}$ & Set of clients  such that for all $\Cl_i \in \Cl_{\Vl}$, $\Vl(u_i)=1$
 \\
 $\mathcal{U}_\Vl$ & Aggregate of valid updates only $\mathcal{U}_\Vl = \sum_{\Cl_i\in \Cl_\Vl}u_i  $\\
$\kappa$ & Security parameter\\
$(i,s_i)$ & $i$-th Shamir's secret share for a secret $s\in \Fl$ \\
$\Psi$ & Check string for the verifiable secret sharing\\
$pp$ & Public parameters of the cryptographic protocols\\
$pk$ & Public key\\
$sk$ & Secret key\\
$sk_{ij}$ & Shared secret key between clients $\Cl_i$ and $\Cl_j$\\
$\mathcal{P}$ & Prover in the SNIP protocol\\
$\mathcal{V}_i$ & $i$-th verifier in the SNIP protocol\\
$\pi$ & SNIP proof \\
$h/f/g$ & Polynomials generated by $\Pro$ for the construction of $\pi$\\
 $(a,b,c)$ & Beaver's triplet generated by $\Pro$ for the construction of $\pi$\\
$[s]_i$ & $i$-th additive secret share for a secret $s \in \Fl$\\
 $w^{out}$ &  Value of the output wire of the circuit $\Vl(\cdot)$\\
 $\sigma$ & Proof summary broadcasted by the verifiers\\
$\Fl$ & A prime field \\
$\Ml$  & Number of multiplication gates in $\Vl(\cdot)$\\
$\Bl$ & Public bulletin board\\
$\Cl_{\setminus i}$ &Set of all clients except $\Cl_i$, $\Cl_{\setminus i} = \Cl\setminus \Cl_i$\\
$\CA$ & List of malicious clients maintained by $\Ser$ in \name\\
$(j,s_{ij})$ & Client $\Cl_j$'s (Shamir secret) share of client $\Cl_i$'s secret $s_i\in \Fl$ in \name\\
$\pi_i$ & Client $\Cl_i$'s proof in \name\\
$h_i/f_i/g_i$ & Polynomials generated by client $\Cl_i$ for the construction of $\pi_i$ in \name\\
$(a_i,b_i,c_i)$ & Beaver's triplet generated by client $\Cl_i$ for the construction of $\pi_i$ in \name\\
$w^{out}_i$ &  Value of the output wire of the circuit $\Vl(u_i)$ for client $\Cl_i$\\
$\Psi_{\pi_i}$ & Check string generated by client $\Cl_i$ for the shares of their proof $\pi_i$ \\

$\Psi_{u_i}$ & Check string generated by client $\Cl_i$ for the shares of their update $u_i$ \\
 $\sigma_{ji}$ & Client $\Cl_j$'s shar of the summary for client $\Cl_i$'s proof  in \name \\
   $\lambda_{ji}$ & Client $\Cl_j$'s share of the random digest for client $\Cl_i$'s proof  in \name \\
     $\lambda_{i}$ & Client $\Cl_i$'s random digest reconstructed from the shares $\{\lambda_{ji}\}, j \in \Cl_{\setminus i}$ \\   $\sigma_{i}$ & Client $\Cl_i$'s proof summary reconstructed from the shares $\{\sigma_{ji}\}, j \in \Cl_{\setminus i}$ \\\bottomrule
\end{tabular}}\label{tab:notations}
\end{table}
\subsection{Building Blocks Cntd.}\label{app:background}

\textbf{Arithmetic Circuit.} An arithmetic circuit, \scalebox{0.9}{$\mathcal{C}: \Fl^k\mapsto \Fl$}, represents a computation over a finite field $\Fl$. It can be represented by a directed acyclic graph (DAG) consisting of three types of nodes: $(1)$ inputs, $(2)$ gates and $(3)$ outputs.  
Input nodes have in-degree zero and out-degree one: the $k$ input nodes return input variables $\{x_1,\cdots,x_k\}$ with $x_i \in \Fl$.
Gate nodes have in-degree two and out-degree one; they perform either the $+$ operation (addition gate) or the $\times$ operation (multiplication gate).
Every circuit has a single output node with out-degree zero. 
A circuit is evaluated by traversing the DAG, starting from the inputs, and assigning a value in $\Fl$ to every wire until the output node is evaluated.

\textbf{Shamir's Secret Sharing Scheme.} The scheme is \emph{additive}, \emph{i.e.}, it allows addition of two secret shared values locally. Formally, for all $s,w \in \Fl$ and $Q\subseteq P, |Q|\leq t$: \begin{gather}s+w\leftarrow \textsf{SS.recon}(\{(i,s_i+w_i)_{i\in Q}\})\end{gather} Additionally, the scheme is a linear secret sharing scheme which means that any linear operations performed on the individual shares translates to operations performed on the secret, upon reconstruction. Specifically, for $Q \subseteq P, |Q|\geq t$ and $\alpha,\beta \in \Fl$: 
\begin{gather}\alpha s+\beta\larrow\textsf{SS.recon}(\{(i,\alpha s_i+\beta)_{i\in Q}\})\end{gather} This means that a party can perform linear operations on the secret \textit{locally}.

\textit{Verifiable Secret Shares.} To make the Shamir's Secret shares verifiable, we use Feldman's \cite{Feldman87} VSS technique.
Let $c(x)=c_0+c_1x+\cdots c_{t-1}x^{t-1}$ denote the  polynomial used in generating the shares where $c_0=s$ is the secret. 
The check string are the commitments to the coeffcients given by
\begin{gather}
\psi_i=g^{c_i} , i \in \{0,\cdots,t-1\}
\end{gather}
where $g$ denotes a generator of $\Fl$.  All arithmetic is taken
modulo $q$ such that $(p|q-1)$ where $p$ is the prime of $\Fl$.

For verifiying a share $(j,s_j)$, a party needs to check  whether $g^{s_j}=\prod_{i=0}^{t-1}\psi_i^{j^i}$.
The privacy of the secret $s=c_0$ is implied by the 
the intractability of computing discrete
logarithms~\cite{Feldman87}.

\textbf{Short Non-Interactive Proofs (SNIP).} Here, we detail the SNIP protocols. SNIP works in two stages as follows:

\textit{(1) Generation of Proof.} The prover $\mathcal{P}$ evaluates the circuit $\Vl(\cdot)$ on its input $x$ to obtain the value that
every wire in the circuit takes on during the computation
of $\Vl(x)$. Using these wire values,  $\mathcal{P}$ constructs 
three randomized polynomials $f$ , $g$, and $h$, which encode
the values of the input and output wires of each of the $\Ml$
multiplication gates in the computation of $\Vl(x)$.

Let us label the $M$  multiplication gates in the $\Vl(\cdot)$ circuit in the  topological order from inputs to outputs as $\{
1, \cdots , \Ml\}$. Let $u_t$ and $v_t$ denote the values on  the left and right input wires
of the $t$-th multiplication gate for $t \in [M]$. The prover $\mathcal{P}$ samples two values $u_0$ and $v_0$ uniformly at random from $\Fl$. $f$ and $g$ are defined to be the lowest degree polynomials such that $f(t) = u_t$ and $g(t) = v_t, \forall t \in [\Ml]$. Next, $h$ is defined as the polynomial
$h = f \cdot g$. The polynomials $f$ and $g$ has degree at most $M$
and the polynomial $h$ has degree at most $2\Ml$. It is easy to see that  $h(t)$ is
 the value of the output wire $(u_t\cdot v_t)$ of the $t$-th
multiplication gate in the $ \Vl(x)$ circuit since $h(t) = f(t) \cdot g(t) = u_t\cdot v_t, \forall t \in [\Ml]$. The prover $\mathcal{P}$ can construct the polynomials $f$ and $g$ using polynomial interpolation and can multiply them
produce $h = f \cdot g$. Additionally, $\mathcal{P}$ samples a single set of Beaver's multiplication triples \cite{Beaver91}: $(a,b,c)\in \Fl^3$ such that $a\cdot b = c \in \Fl$. Prover $\mathcal{P}$ constructs the proof share $[\pi]_i = \langle [f(0)]_i,[g(0)]_i,[h]_i,([a]_i,[b]_i,[c]_i)\rangle$ \footnote{Note that we omitted the terms $[f(0)]_i$ and  $[g(0)]_i$ from $\pi_i$ in Sec. \ref{sec:system:block}  for the ease of exposition.} for verifier $\mathcal{V}_i$ by splitting:
\begin{itemize} \item the random values $f(0) = u_0$ and $g(0) = v_0$, using
additive secret sharing, \item  the coefficients of $h$ (denoted by $[h]_i$,  and
\item the sampled Beaver's triplets $(a,b,c)$. \end{itemize}

The prover then sends the respective shares of the input and the proof $( [x]_i, [\pi]_{i})$ to each of the verifiers $\Ver_i$.

\textit{ (2) Verification of Proof.} Using  $[x]_i$, the share of the provers's  private value $x$, and $[f(0)]_i$, $[g(0)]_i$, and $[h]_i$, each verifier $\Ver_i$ can \textit{locally} (i.e., without communicating with the other verifiers/prover) produce shares $[f]_i$ and
$[g]_i$ of the polynomials $f$ and $g$ as follows:
\begin{itemize}
    \item $\Ver_i$ reconstructs a share of every wire for the $\Vl(x)$ circuit. This is possible since $\Ver_i$ has access to $(1)$ a share of each of the input wire values $([x]_i)$ and $(2)$ a  share of each wire value coming out of a multiplication gate $([h]_i(t), t \in [\Ml]$ is a share of the $t$-th such wire). Hence, $\Ver_i$  can derive all other wire value shares via affine operations on the wire value shares it already has.
    \item Using these wire value shares and shares of $f(0)$
and $g(0)$,  $\Ver_i$  uses polynomial
interpolation to construct $[f]_i$ and $[g]_i$
\end{itemize}

To verify that $\Vl(x)=1$ and hence, accept the input $x$, the verifiers need to check two things: 
\begin{itemize} \item check the consistency of $\Pro$'s computation of $\Vl(x)$, and \item check that the value of final output wire of the computation, $\Vl(x)$, denoted by $w^{out}$ is indeed $1$. \end{itemize}
For carrying out the above mentioned checks, the verifier $\Ver_i$ broadcasts a summary $\sigma_i=([w^{out}]_i,[\lambda]_i)$, where $[w^{out}]_i$ is $\Ver_i$'s share of the output wire and $[\lambda]_i$ is a share of a random digest that the verifier computes from the shares of the other wire values and the proof  share $\pi_i$. The details are discussed  as follows:

\textit{(2a) Checking the Consistency of $\Pro$'s Computation of $\Vl(x)$.} 
For honest provers and verifiers, the verifiers will now hold shares of polynomials $f$ , $g$, and $h$ such that $f \cdot g = h$.
In contrast, a malicious prover could have sent the
verifiers shares of a different polynomial $\hat{h}$ such that, for some
$t \in [M], hˆ(t)$ is not the value on the output wire in
the $t$-th multiplication gate of the $\Ver(x)$ circuit.
In this case, the verifiers end up reconstructing shares of polynomials $\hat{f}$ and $\hat{g}$ that might not be equal to $f$ and $g$. Then, we have $\hat{h}\neq  \hat{f} \cdot \hat{g}$ as explained below. Consider the least $t'$ for which $\hat{h}(t') \neq h(t')$. For all $t \leq t'$,
$\hat{f}(t) = f(t)$ and $g(t) = g(t)$, by construction. Since,

\begin{gather}
    \hat{h}(t') \neq h(t') = f(t') \cdot g(t') = \hat{f}(t)' \cdot \hat{g}(t'),
\end{gather}
it must be that $\hat{h}(t') \neq \hat{f}(t') \cdot \hat{g}(t')$, so $\hat{h} \neq \hat{f} \cdot \hat{g}.$ The verifiers can employ the above check using the Schwartz-Zippel randomized polynomial identity test~\cite{Schwartz80,Zippel79} as explained later in this section.

\textit{(2b) Output Verification.} In case all the verifiers are honest,  each $\Ver_i$
now holds a set of shares of the values of all the wires of the $\Vl(x)$ circuit. So to confirm that
$\Vl(x) = 1$, the verifiers need only broadcast their shares of
the output wire $w^{out}_i$. The verifiers can thus reconstruct its exact value from all the broadcasted shares $w^{out}=\sum_{i=1}^k [w^{out}]_i$ and check whether $w^{out}=1$, in which case  it must be that $\Vl(x) = 1$ (except with some small failure probability due to the polynomial identity test).

\textit{Polynomial Identity Test}. Recall that each verifier $\Ver_i$
holds shares $[\hat{f}]_i$, $[\hat{g}]_i$ and $[\hat{h}]_i$
of the  polynomials $\hat{f}$ , $\hat{g}$ and $\hat{h}$. Furthermore, it holds that $\hat{f} \cdot \hat{g} = \hat{h}$
if and only the set of the wire value shares, held by the verifiers, sum up to the internal
wire values of the $\Vl(x)$ circuit computation.
The verifiers now execute a variant of the Schwartz-Zippel randomized polynomial identity test 
to check whether this relation holds. The main idea of the test is that  if $\hat{f}(t) \cdot \hat{g}(t) \neq \hat{h}(t)$, then the polynomial
$t \cdot (\hat{f}(t) \cdot \hat{g}(t) - \hat{h}(t))$ is a non-zero polynomial of degree at
most $2\Ml + 1$. (The utility of multiplying the polynomial $\hat{f} \cdot \hat{g} - \hat{h}$ by $t$ is explained in the next paragraph)
 Such a polynomial can have at
most $2\Ml + 1$ zeros in $\Fl$, so for a $r \in \Fl$ chosen at random
and after evaluating $r \cdot (
\hat{f}(r) \cdot \hat{g}(r) - \hat{h}(r))$, the verifiers will detect
that $\hat{f } \cdot \hat{g} \neq \hat{h}$ with probability at least $1 \frac{2\Ml+1}{|\Fl|}$.

For the polynomial identity test, one of the
verifiers samples a random value $r \in \Fl$ and broadcasts it. Each verifier $\Ver_i$ can locally compute the shares $[\hat{f}(r)]_i$, $[\hat{g}(r)]_i$, and $[\hat{h}(r)]_i$ since polynomial evaluation
requires only affine operations. $\Ver_i$ then
applies a local linear operation to these last two shares to
obtain the shares $[r \cdot \hat{g}(r)]_i$ and $[r \cdot \hat{h}(r)]_i$.

\textit{Multiplication of Shares.} Note that  the verifiers need to securely multiply their shares  $[\hat{f}(r)]_i$ and $[r \cdot \hat{g}(r)]_i$
to get a share $[r \cdot
\hat{f}(r) \cdot \hat{g}(r)]_i$ without leaking anything
to each other about the values $\hat{f}(r)$ and $\hat{g}(r)$. This can be performed via the Beaver's MPC multiplication protocol (described later). Using this protocol, verifiers with access to one-time-use shares $([a]_i, [b]_i, [c]_i) \in \Fl^3$ of
random values such that $a \cdot b = c \in \Fl$ (“multiplication
triples”), can execute a
multi-party multiplication of a pair of secret-shared values.
For SNIPs, the prover $\mathcal{P}$ generates the multiplication triple on behalf of the verifiers and sends shares of
these values to each verifier. If $\mathcal{P}$ produces the shares
of these values correctly, then the verifiers can perform
a multi-party multiplication of shares to complete the
correctness check as discussed above. More importantly, 
we can ensure that even if $\mathcal{P}$ sends 
shares of an invalid multiplication triple,
the verifiers will still catch the cheating prover with high
probability. Let's assume that the cheating prover sends the
shares $([a]_i, [b], [c]i) \in \Fl^3$
such that $a \cdot b \neq c \in \Fl$.
Let $a \cdot b = (c + \alpha) \in \Fl$, for some constant
$\alpha > 0$. Executing the polynomial identity
test using the above triples will shift the result of the test by $\alpha$. So the verifiers will be effectively testing
whether the polynomial
\begin{gather}
\hat{Q}(t) = t \cdot (\hat{f}(t) \cdot \hat{g}(t) - \hat{h}(t)) + \alpha \end{gather}
is identically zero. Whenever $\hat{f} \cdot \hat{g} \neq \hat{h}$, it holds that
$t \cdot (\hat{f}(t) \cdot \hat{g}(t) - \hat{h}(t))$ is a non-zero polynomial. So, if
$\hat{f} \cdot \hat{g}\neq \hat{h}$, then $\hat{Q}(t)$ must also be a non-zero polynomial.
Note that the multiplying the term $"\hat{f}\cdot\hat{g}-\hat{h}"$ by $t$ ensures that whenever this expression is non-zero, the
resulting polynomial $\hat{Q}$ is guaranteed to be non-zero, even if $\hat{f}$ , $\hat{g}$, and $\hat{h}$ are
constants, and the prover chooses $\alpha$ adversarially. Since SNIP assumes honest verifiers, we may assume that the prover did not know
the random value $r$ while generating
its multiplication triple. This implies that $r$ is distributed
independently of $\alpha$ which means that we will catch a cheating prover with probability $1-\frac{2M+1}{|\Fl|}$.

\textbf{Beaver's MPC Multiplication Protocol.} SNIP uses  Beaver's 
multiplication triples as follows. 
A multiplication triple is a one-time-use triple of values
$(a, b, c) \in \Fl^3$, chosen at random subject to the constraint
that $a \cdot b = c \in \Fl$.  In SNIP,
computation, each verifier $\Ver_i$ holds a share $([a]_i, [b]_i, [c]_i) \in \Fl^3$ of the triple.
Using their shares of one such triple $(a, b, c)$, the verifiers
can jointly evaluate shares of the output of a multiplication
gate $yz$. To do so, each verifier  uses its shares $[y]_i$ and $[z]_i$
of the input wires, along with the first two components of
its multiplication triple to compute the following values:
\begin{gather}[d]_i = [y]_i - [a]_i
\\ [e]_i = [z]_i - [b]_i
\end{gather}
Each verifier $\Ver_i$  then broadcasts $[d]_i$ and $[e]_i$. Using the
broadcasted shares, every verifier can reconstruct $d$ and $e$
and can compute:
\begin{gather}[\lambda]_i = de/k + d[b]_i + e[a]_i + [c]_i
\end{gather}
Clearly, $\sum_{i=1}^k[\lambda]_i=yz$. Thus, this step requires a round of communication for the broadcast and three reconstructions for $d$, $e$ and $\lambda$.

For SNIPs on Shamir's secret shares, the verifier $\Ver_i$ compute the shares  $(i,\lambda_i)$ where $\lambda_i=de+db_i+ea_i+c_i$ which gives $yz\leftarrow \textsf{SS.recon}((i,\lambda_i))$.

As mentioned in Sec. \ref{sec:opt}, we can leverage the multiplicativity of Shamir's secret shares to generate $\lambda_i$ for client $\Cl_i$ locally. Specifically, each client can locally multiply the shares $(j,f_{ij})$ and $(j,g_{ij})$ to generate $(i,(f_j\cdot g_j)_i)$. In order to make the shares consistent, $\Cl_i$  multiplies the share of $(i,h_{ji})$ with $(i,z_i)$ where $z=1$ (these can be generated and shared by the server $\Ser$ in the clear). In this way, $\Cl_j$ can locally generate a share of the digest $(j,d_{ij})$ that correspond to a polynomial of degree $2m$. Since $m<\frac{n-1}{4}$, this optimization is still compatible with robust reconstruction. In this way, we save one round of communication and require only one reconstruction for $\lambda_i$ instead of three.

\subsection{Complexity Analysis}\label{sec:complexity}
We present the complexity analysis of \name~in terms of the number of clients $n$, number of malicious clients $m$ and data dimension $d$ (Table \ref{tab:complexity}).

\textbf{Computation Cost.}
Each client $\Cl_i$'s computation cost can be broken into six components: $(1)$ performing \scalebox{0.9}{$n\!-\!1$} key agreements -- \scalebox{0.9}{$O(n)$}; $(2)$ generating proof $\pi_i$ for \scalebox{0.9}{$\Vl(u_i)=1$} -- \scalebox{0.9}{$O(|\Vl|+\Ml\log\Ml)$}\footnote{We use standard discrete FFT for all polynomial operations \cite{FFT}.}; $(3)$ creating secret shares of the update $u_i$ and the proof $\pi_i$ -- \scalebox{0.9}{$O(mn(d+\Ml))$}\footnote{This uses the fact that the Lagrange coefficients can be pre-computed~\cite{Lagrange}.}; $(4)$ verifying the validity of the received shares -- \scalebox{0.9}{$O(mn(d+\Ml)$};  $(5)$ generating proof digest for all other clients -- \scalebox{0.9}{$O(n|\Vl|)$}; and $(6)$ generating shares of the final aggregate -- \scalebox{0.9}{$O(nd)$}. Assuming \scalebox{0.9}{$|\Vl|$} is of the order of \scalebox{0.9}{$O(d)$}, the overall computation complexity of each client $\Cl_i$ is \scalebox{0.9}{$O(mnd)$}. \\The server $\Ser$'s computation costs can be divided into three parts: $(1)$ verifying the validity of the flagged shares -- \scalebox{0.9}{$O(md\min(n,m^2))$}; $(2)$ verifying the proof digest for all clients -- \scalebox{0.9}{$O(n^2\log^2n\log\log n)$}; and $(3)$ computing the final aggregate -- \scalebox{0.9}{$O(dn\log^2n\log\log n)$}. Hence, the total computation complexity of the server is \scalebox{0.9}{$O\big((n+d)n\log^2n\log \log n$} \scalebox{0.9}{$+md\min(n,m^2)\big)$}.

\textbf{Communication Cost.}  The communication
cost of each client $\Cl_i$ has seven components: $(1)$ exchanging keys with all other clients -- \scalebox{0.9}{$O(n)$}; $(2)$ receiving \scalebox{0.9}{$\Vl(\cdot)$} -- \scalebox{0.9}{$O(|\Vl|)$}; $(3)$ sending encrypted secret shares and check strings for all other clients -- \scalebox{0.9}{$O(n(d+\Ml)+md)$}; $(4)$ receiving encrypted secret shares and check strings from all other clients -- \scalebox{0.9}{$O(n(d+\Ml)+mnd)$}; $(5)$ sending proof digests for every other client -- \scalebox{0.9}{$O(n)$}; $(6)$ receiving the list of corrupt clients $\Cl$ -- \scalebox{0.9}{$O(m)$}; and $(7)$ sending the final aggregate -- \scalebox{0.9}{$O(d)$}. Thus, the communication complexity for every client is \scalebox{0.9}{$O(mnd)$}. \\
The servers communication costs include: $(1)$ sending the validation predicate -- \scalebox{0.9}{$O(|\Vl|)$}; $(2)$
receiving check strings and secret shares from flagged clients -- \scalebox{0.9}{$O(md\min(n,m^2))$}; $(3)$ receiving proof digests -- \scalebox{0.9}{$O(n^2)$}; $(4)$ sending the list of malicious clients -- \scalebox{0.9}{$O(m)$}; and $(5)$ receiving the shares of the final aggregate -- \scalebox{0.9}{$O(nd)$}. Hence, the overall communication complexity of the server is \scalebox{0.9}{$O(n^2+md\min(n,m^2))$}. 
 The total number of one-way communication is $12$ and $9$ for the clients and server, respectively, \textit{independent} of the complexity of the validation predicate.

 \subsection{Proof for Lemma \ref{lemma:4}}\label{app:proof1}
 
 \begin{proof}In Round 3, the proof corresponding to a  client $\Cl_i$ is verified iff it has submitted valid shares for the $n\!-\!m\!-\!1$ honest clients $\Cl_H\setminus \Cl_i$. This is clearly true if $\Cl_i$ is honest. If $\Cl_i$ is malicious, \emph{i.e.}, it submitted at least one invalid share: 
\squishlist
\item \textit{Case 1:} \scalebox{0.9}{$|\textsf{Flag}[i]|\geq m+1$}. It is clear that $\Cl_i$ has submitted an invalid share to at least one honest client and, hence, is removed from the rest of the protocol. 
\item \textit{Case 2:} \scalebox{0.9}{$|\textsf{Flag}[i]|\leq m$}. All honest clients in $\Cl_H$ will be flagging $\Cl_i$. Hence, $\Cl_i$ either has to submit the corresponding valid shares or be removed from the protocol. 
\squishend
Given $n\!-\!m\!-\!1$ valid shares, using Fact 2, we know that \name~reconstructs the proof summary for $\Cl_i$ correctly. Eq. \ref{eq:soundness} then follows from the soundness property of SNIP. 
\end{proof}
\subsection{Proof for Lemma \ref{lemma:5}}\label{app:proof2}
\begin{proof}
In Round 2, observe that the shares $(j,u_{ij}), (j,\pi_{ij})$ for each client $\Cl_j \in \Cl_{\setminus i}$ are encrypted with the pairwise secret key and distributed. Hence, a collusion of $m$ malicious clients (and the server $\Ser$)\footnote{The server does not have access to any share of its own in \name.} can access \textit{at most} $m$ shares of any honest client $\Cl_i\in \Cl_H$. This is true even in Round 3 where:
\squishlist\item A malicious client might falsely flag $\Cl_i$.
\item No honest client in $\Cl_H\setminus \Cl_i$ will flag $\Cl_i$ since they would be receiving valid shares (and their encryptions) from $\Cl_i$.
\item $\Ser$ cannot lie about who flagged who, since everything is logged publicly on the bulletin $\Bl$.
\squishend
Thus, only $m$ shares of $\Cl_i$ can be revealed which correspond to the $m$ malicious clients.
\\Since at least $m+1$ shares are required to recover the secret, any instantiation of the SNIP verification protocol (\emph{i.e.}, reconstruction of the values of $\sigma_i=(w^{out}_i,\lambda_i)$) requires at least one \textit{honest} client to act as the verifier. Hence, at the end of Round 3, from Fact 1 and the zero-knowledge property of SNIP, the only information revealed is that $\Vl(u_i)=1$.
\end{proof}
\subsection{Security Proof}\label{app:security}

\begin{theorem} Given a public validation predicate \scalebox{0.9}{$\Vl(\cdot)$}, security parameter \scalebox{0.9}{$\kappa$}, set of $m$ malicious clients $\Cl_M, \lfloor m<\frac{n-1}{3}\rfloor$ and a malicious server $\Ser$, there exists a probabilistic polynomial-time (P.P.T.) simulator $\textsf{Sim}(\cdot)$ such that: \begin{gather*}\textsf{Real}_{\name}\big(\{u_{\Cl_H}\}, \Omega_{\Cl_M\cup\Ser}\big)\equiv_C\textsf{Sim}\big(\Omega_{\Cl_M\cup \Ser},\Ul_H,\Cl_H\big)\\\mbox{where }\Ul_H=\sum_{\Cl_i\in \Cl_H}u_i.\end{gather*}  $\{u_{\Cl_H}\}$ denotes the input of all the honest clients, $\textsf{Real}_\name$ denotes a random variable representing the joint view of all the parties in \name's execution,  $\Omega_{\Cl_M\cup\Ser}$ indicates a polynomial-time algorithm implementing
the “next-message” function of the parties in $\Cl_M\cup\Ser$, 
and $\equiv_C$ denotes computational indistinguishability.\end{theorem}
\begin{proof}We prove the theorem by a standard hybrid argument. Let $\Omega_{\Cl_M\cup \Ser}$ indicate the polynomial-time algorithm that denotes
the “next-message” function of parties in $\Cl_M\cup \Ser$. That is, given a
party identifier $c \in \Cl_M\cup \Ser$, a round index $i$, a transcript $T$ of all
messages sent and received so far by all parties in 
$\Cl_M\cup \Ser$, joint
randomness $r_{\Cl_M\cup \Ser}$ for the corrupt parties’ execution, and access
to random oracle $O$, $\Omega_{\Cl_M\cup \Ser}(c, i, T, r_{\Cl_M\cup \Ser}$) outputs the message for
party $c$ in round $i$ (possibly making several queries to $O$ along
the way).  We note that $\Omega_{\Cl_M\cup \Ser}$ is thus effectively
choosing the inputs for all corrupt users.

We will define a simulator \textsf{Sim} through a series of
(polynomially many) subsequent modifications to the real
execution $\textsf{Real}_\name$, so that the views of $\Omega_{\Cl_M\cup\Ser}$ in any two subsequent executions are computationally indistinguishable. 
\begin{enumerate}\item$\textsf{Hyb}_0$: This random variable is distributed exactly as the
view of $\Omega_{\Cl_M\cup\Ser}$ in $\textsf{Real}_{\name}$, the joint view of the parties $\Cl_M\cup\Ser$
in a real execution of the protocol.
\item $\textsf{Hyb}_1$:
In this hybrid, for 
any pair of honest clients $\Cl_i,\Cl_j \in \Cl_H$, the simulator changes the key from $\textsf{KA.agree}(pk_j,sk_i)$ to  a uniformly random key. We use Diffie-Hellman key exchange protocol in \name.
The DDH assumption~\cite{DH76}  guarantees that this hybrid is indistinguishable from the previous one. 
also be able to break the DDH.
\item $\textsf{Hyb}_2$: This hybrid is identical to $\textsf{Hyb}_1$
, except additionally, \textsf{Sim} will abort if $\Omega_{\Cl_M\cup \Ser}$ succeeds to deliver, in
round 2, a message to an honest client $\Cl_i$
on behalf of another honest client $\Cl_j$, such that $(1)$ the
message is different from that of  $\textsf{Sim}$, and $(2)$ the message does
not cause the decryption  to fail. Such a message would directly
violate the IND-CCA security
of the encryption scheme.

\item $\textsf{Hyb}_3$: In this round, for every honest party in $\Cl_H$, \textsf{Sim} samples $s_i\in \Fl$ such that $\Vl(s_i)=1$ and replaces all the shares and the check strings accordingly. This allows the server to compute the $\sigma_i=(w^{out}_i,\lambda_i)$ such that $w^{out}_i=1 \wedge \lambda_i=0$  for all honest clients in the same way as in the previous hybrid. An adversary noticing any difference would break $(1)$ the computational discrete logarithm assumption used by the VSS~\cite{Feldman87}, OR  $(2)$ the \textsf{IND-CCA} guarantee of the encryption scheme, OR $(3)$ the information theoretic perfect secrecy of Shamir's secret sharing scheme with threshold $m+1$,  OR $(4)$ zero-knowledge property of SNIP.

\item $\textsf{Hyb}_4$: In this hybrid, \textsf{Sim} uses $\Ul_H$ to compute the following polynomial. Let $(j,S_j)$ represent the share of $\sum_{i\in \Cl_H}{s_i}$ for a malicious client $\Cl_j\in\Cl\setminus\Cl_H$ where $s_i$ denotes the random input \textsf{Sim} had sampled for $\Cl_i\in \Cl_H$ in $\textsf{Hyb}_3$. \textsf{Sim} performs polynomial interpolation to find the $m+1$-degree polynomial $p*$ that satisfies  $p*(0)=\Ul_H$ and $p(j)=S_j$. Next, for all honest client, \textsf{Sim} computes the share for $\Ul=\Ul_H+\sum_{\Cl_j\in\bar{C}}u_j$ (Eq. \ref{eq:aggregate}) by using the polynomial $p*$ and the relevant messages from $\Omega_{\Cl_M\cup\Ser}$. Clearly, this hybrid is indistinguishable from the previous one by the perfect secrecy of Shamir's secret shares. This concludes our proof.
\end{enumerate}
\end{proof}
\subsection{Additional Evaluation Results}\label{app:eval}

In this section, we provide some additional evaluation results on model accuracy in Fig. \ref{fig:app:eval}. We use the same configuration as the one reported in Sec. \ref{sec:evaluation}. Our observations are in line with our discussion in Sec. \ref{sec:eval:models}. 
\subsection{Discussion Cntd.}\label{app:discussion}
Here, we present additional avenues of future work for \name.

\noindent\textbf{Revealing Malicious Clients.} In our current implementation, \name~publishes the (partial) list of malicious clients $\CA$. 
To hide the identity of malicious clients, we could include an equal number of honest clients in the list before publishing it, thereby providing those clients plausible deniability. 
We leave more advanced cryptographic solutions as a future direction.

\begin{figure*}[]
    \begin{subfigure}{0.24\linewidth}
        \centering
         \includegraphics[width=0.9\linewidth]{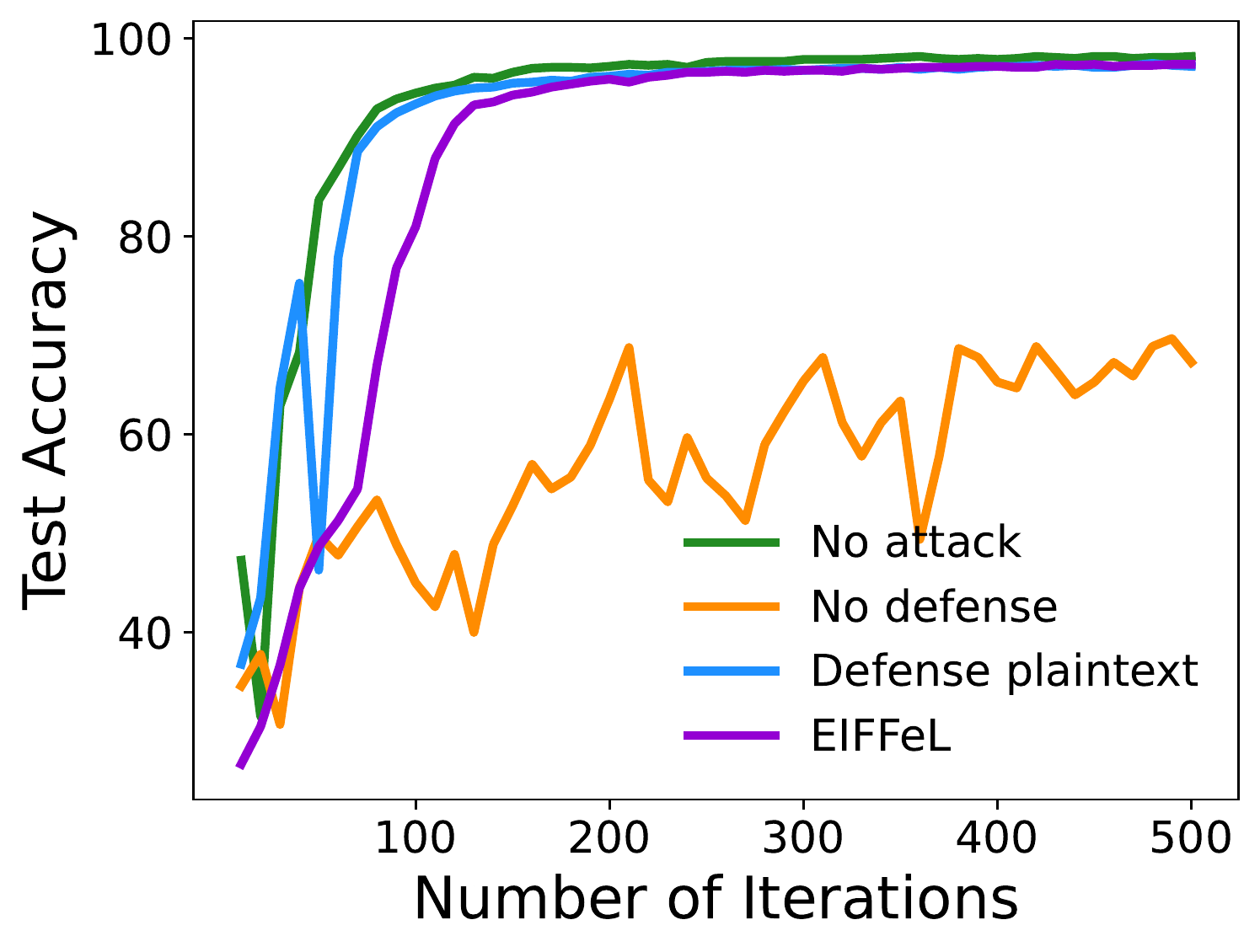}
       \caption{MNIST: Scaling attack with norm bound validation predicate.}
        \label{fig:MNIST:Norm}
    \end{subfigure}
    \begin{subfigure}{0.24\linewidth}
    \centering \includegraphics[width=0.9\linewidth]{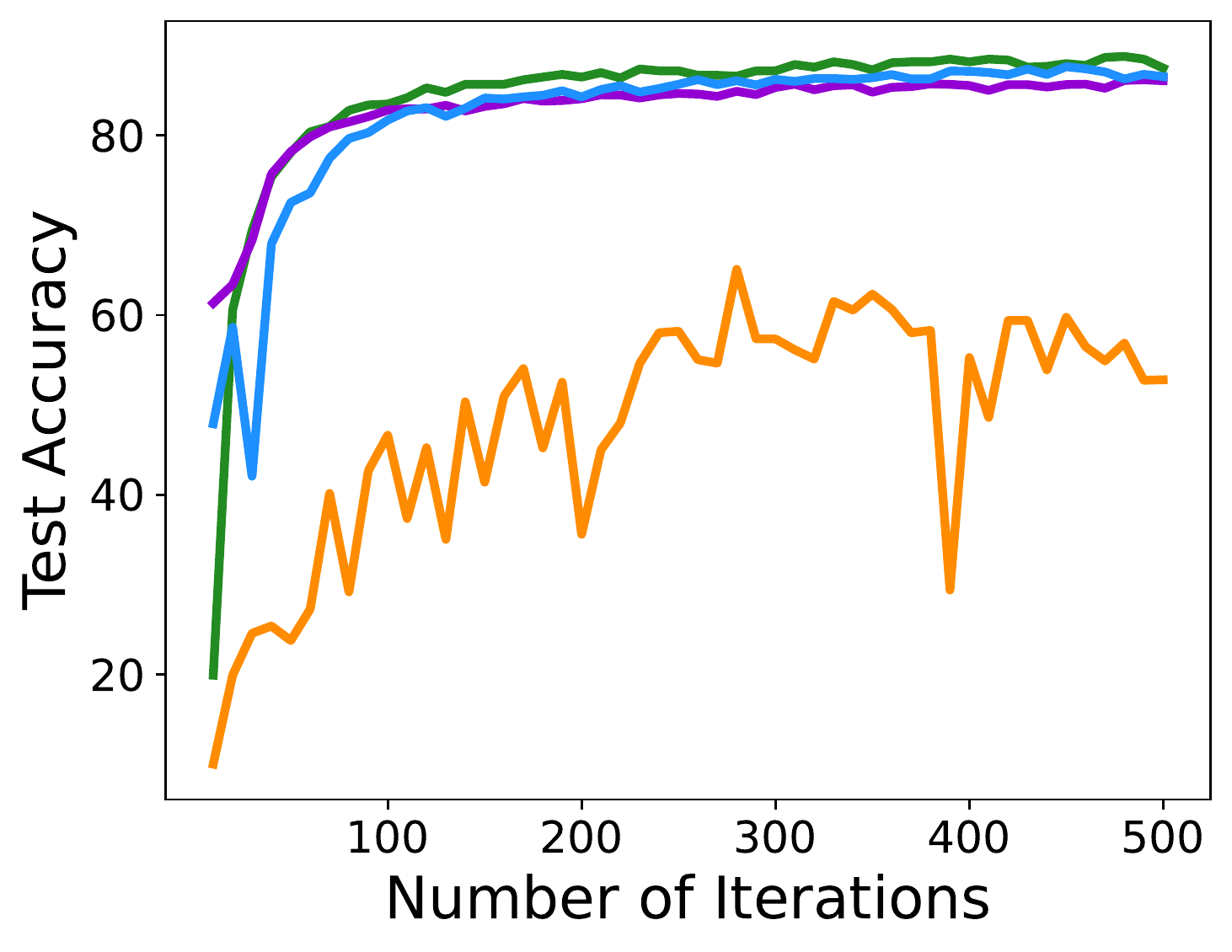}  
 \caption{FMNIST: Scaling attack with norm bound validation predicate.}
        \label{fig:FMNIST:Norm}\end{subfigure}
         \begin{subfigure}{0.24\linewidth}
    \centering \includegraphics[width=0.9\linewidth]{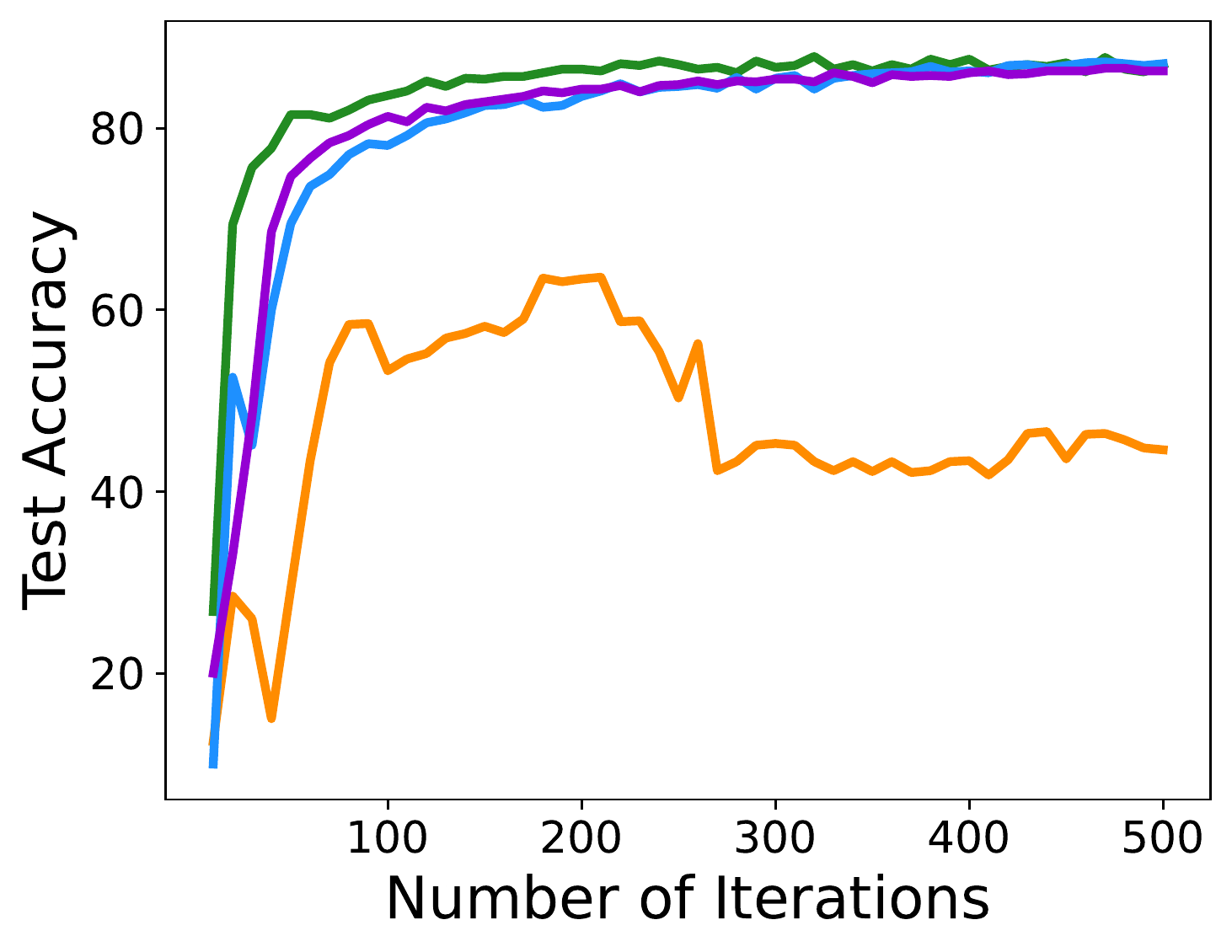}  
 \caption{FMNIST: Min-Sum attack with cosine similarity validation predicate. }
        \label{fig:FMNIST:Cosine}\end{subfigure}
  \begin{subfigure}{0.24\linewidth}
    \centering \includegraphics[width=0.9\linewidth]{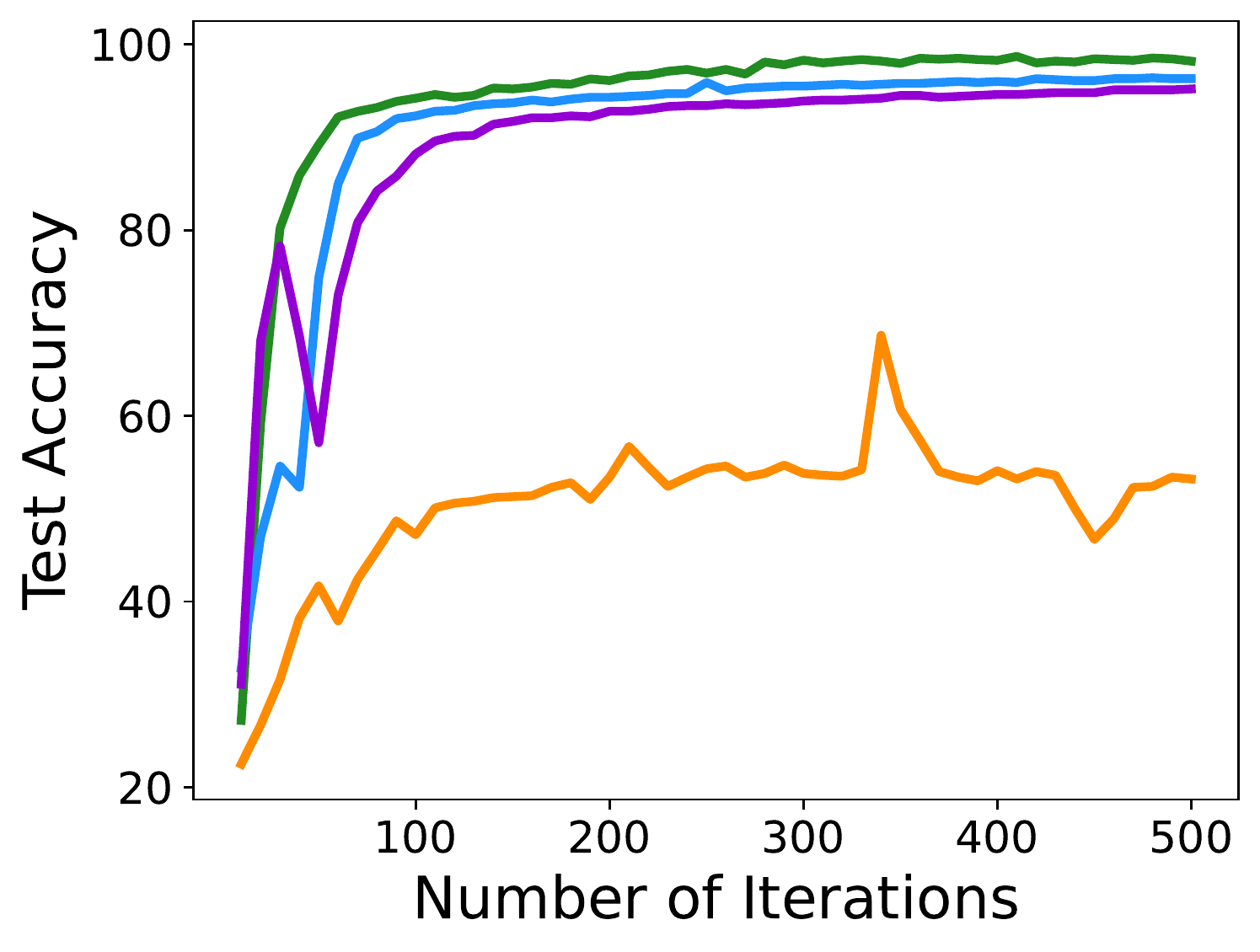}  
 \caption{EMNIST: Additive noise attack with Zeno++ similarity validation predicate.}\label{fig:EMNIST:Zeno}\end{subfigure} 
         \begin{subfigure}{0.24\linewidth}
    \centering \includegraphics[width=0.9\linewidth]{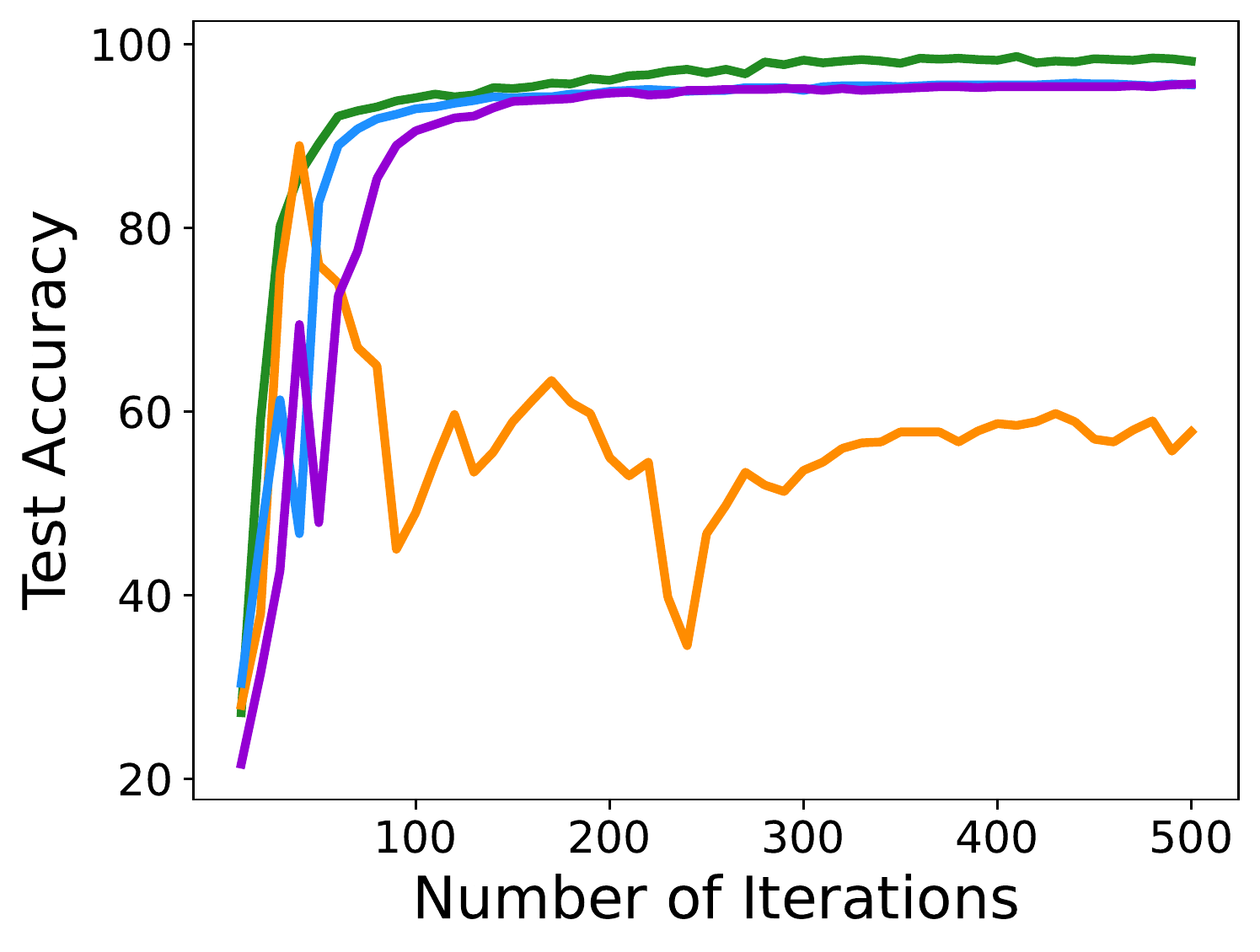}  
 \caption{EMNIST: Scaling attack with cosine similarity validation predicate.}
        \label{fig:EMNIST:Cosine}\end{subfigure}   \begin{subfigure}{0.24\linewidth}
        \centering
         \includegraphics[width=0.9\linewidth]{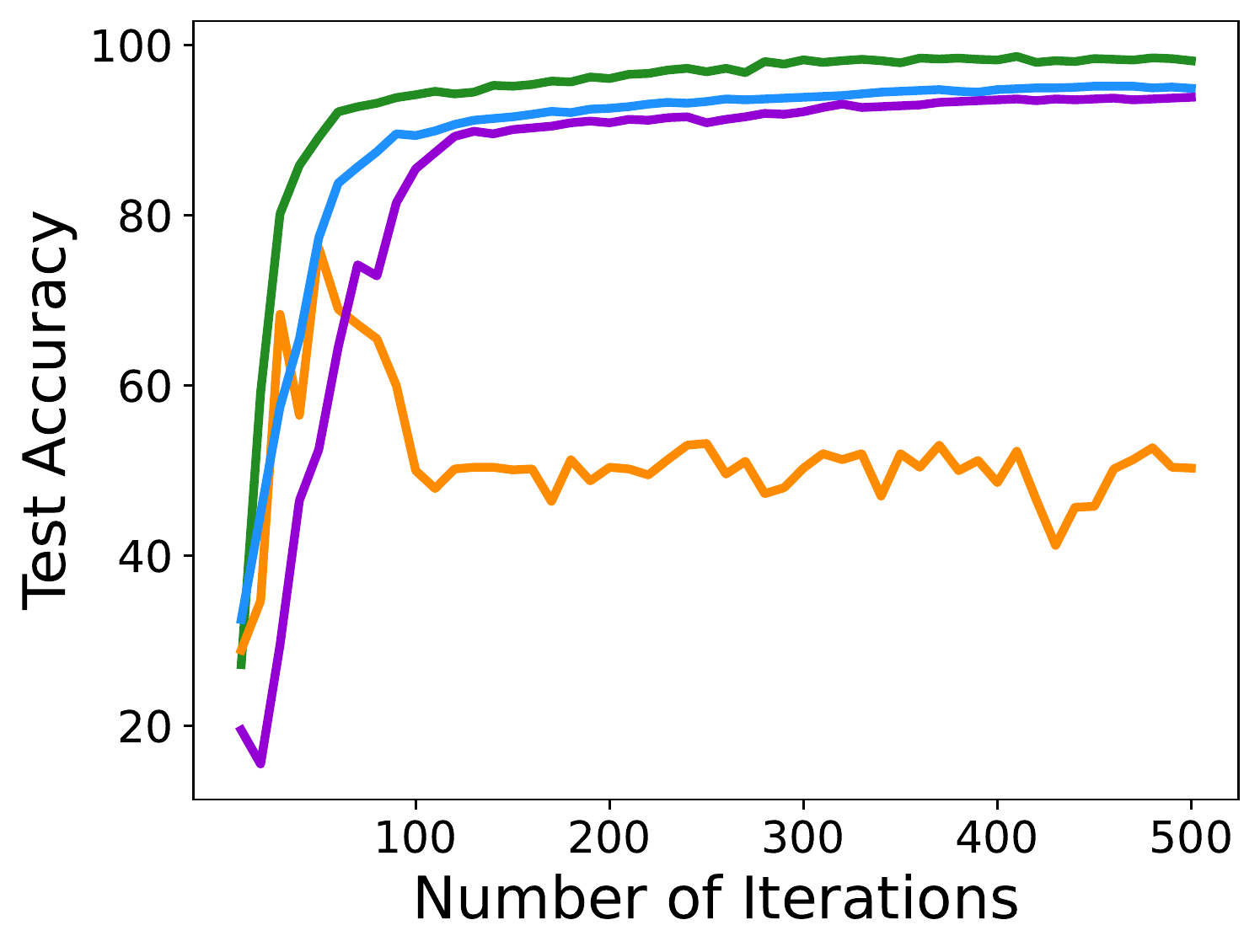}
       \caption{EMNIST: Sign flip attack with norm ball validation predicate.}
        \label{fig:EMNIST:Ball}
    \end{subfigure}
        \begin{subfigure}{0.24\linewidth}
    \centering \includegraphics[width=0.9\linewidth]{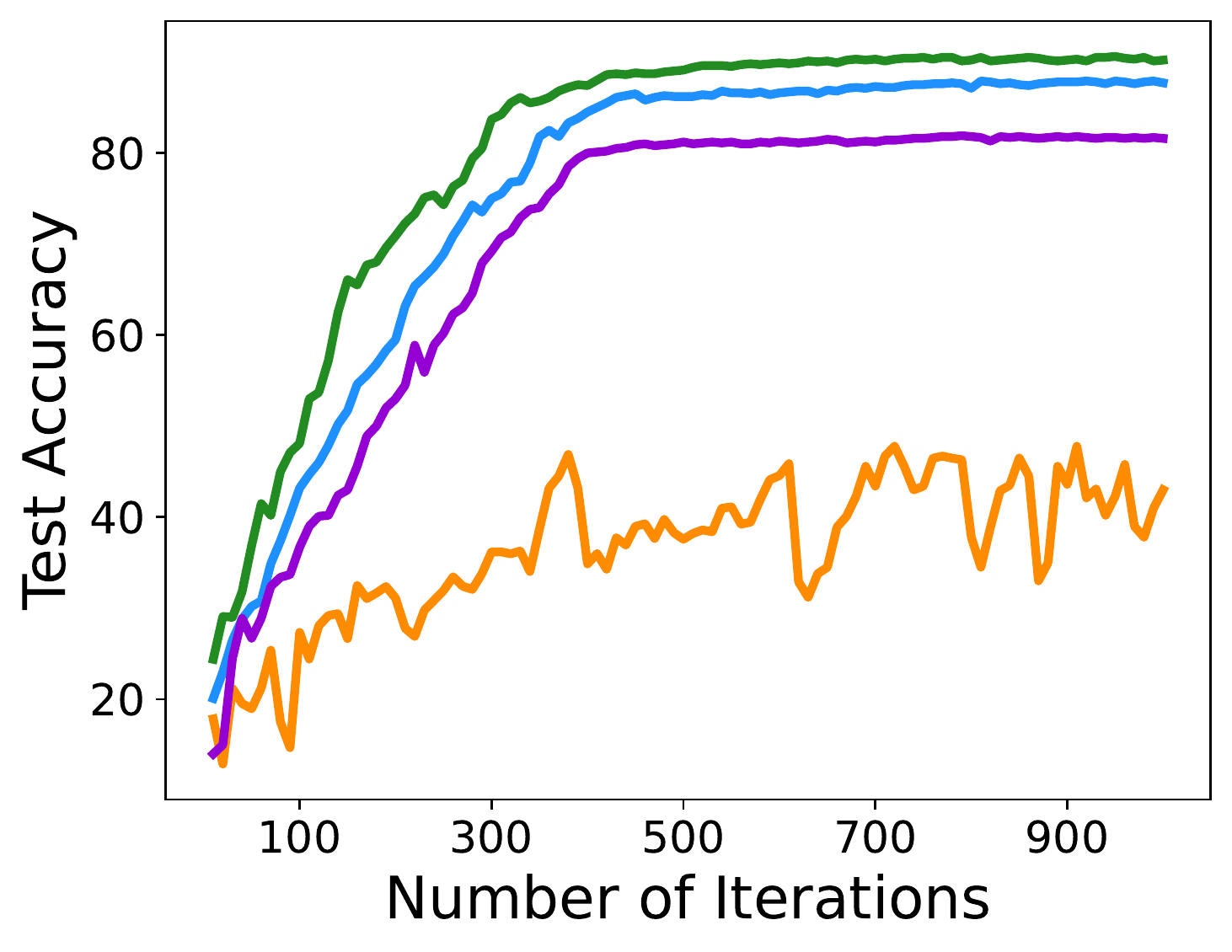}  
 \caption{CIFAR-10: Min-Max attack with Zeno++ validation predicate.}\label{fig:CIFAR:Zeno}
 \end{subfigure}
 \begin{subfigure}{0.24\linewidth}
    \centering \includegraphics[width=0.9\linewidth]{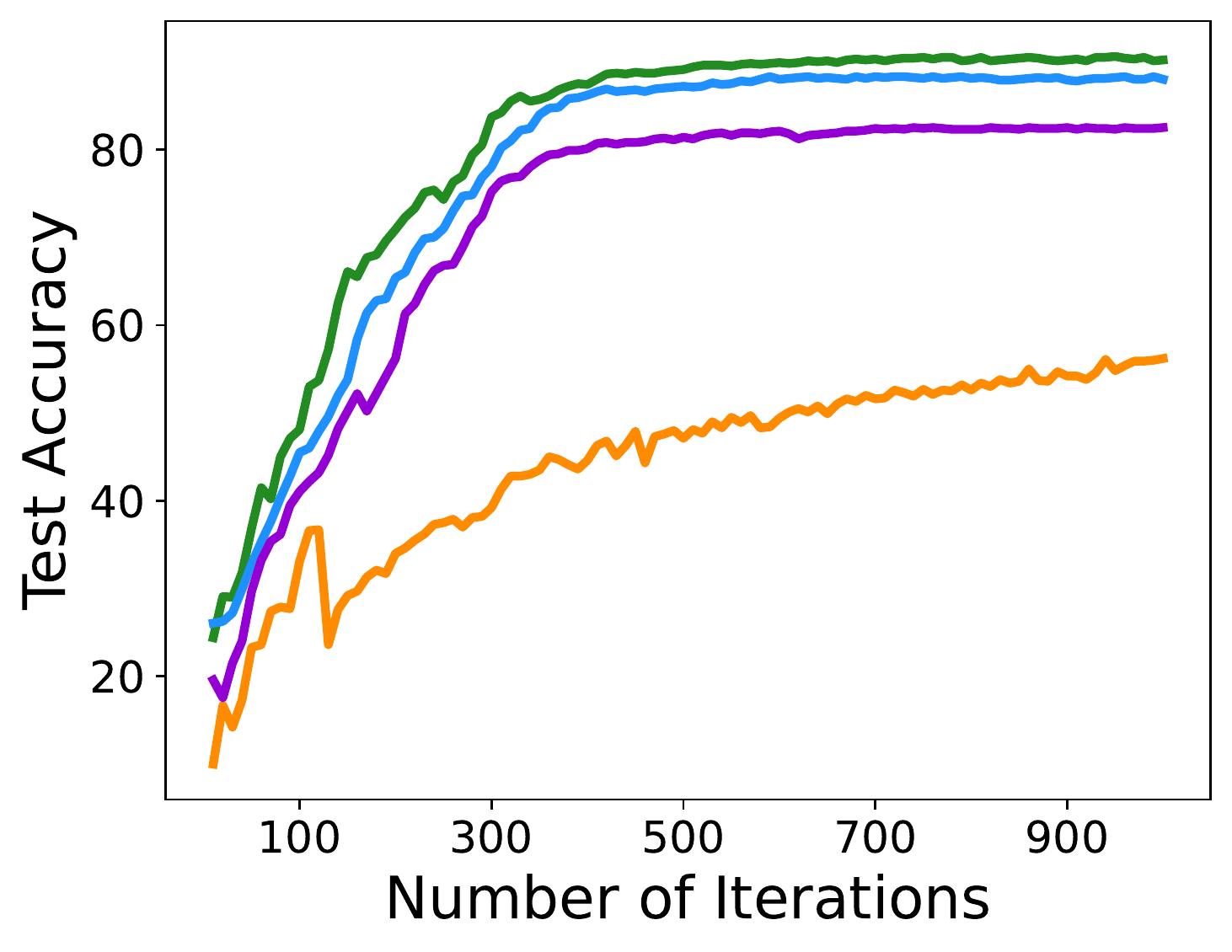}  
 \caption{CIFAR-10: Scaling attack with norm bound validation predicate}\label{fig:CIFAR:Norm}\end{subfigure}
 \\
          \begin{subfigure}{0.24\linewidth}
    \centering \includegraphics[width=0.9\linewidth]{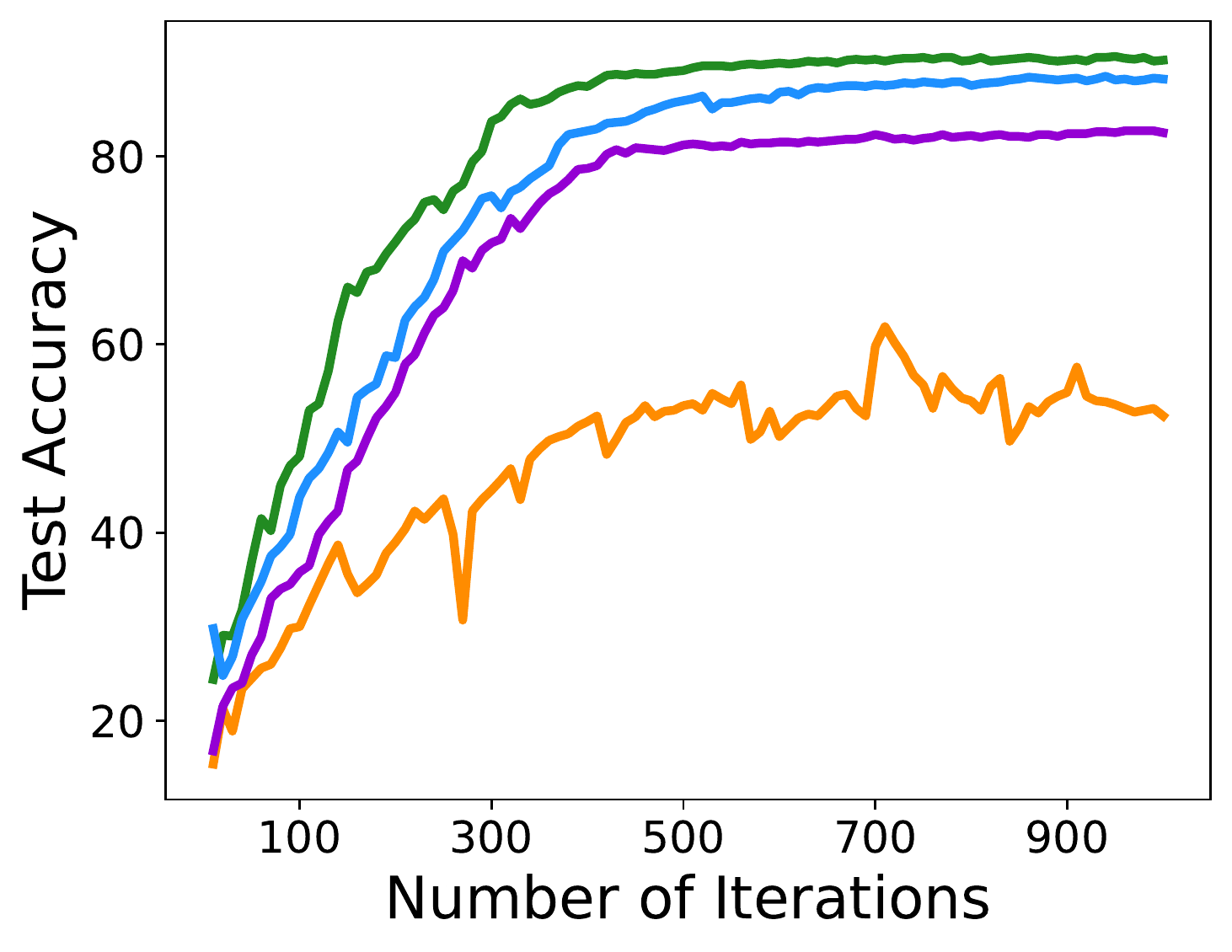}   
 \caption{CIFAR-10: Sign flip attack with norm ball validation predicate.}
        \label{fig:CIFAR:Ball}\end{subfigure}
        
   \caption{Accuracy analysis of \name~continued. Test accuracy is shown as a function of the FL iteration for different datasets and attacks.}
   \label{fig:app:eval}\vspace{-0.3cm}
\end{figure*}

\noindent\textbf{Private Validation Predicate.}  If $\Vl(\cdot)$ contains some secrets of the server $\Ser$, we can employ multiple servers where the computation of $\Vl(u)$ is done at the servers \cite{corrigan2017prio}. We leave a single-server solution of this problem for future work.

\noindent\textbf{Byzantine-Robust Aggregation.} In \name, the integrity check is done individually on each client update, independent of all other clients.
An alternative approach to compare the  local
model updates of \textit{all} the clients (via pairwise distance/ cosine similarity)  \cite{blanchard2017byzantine,blanchard2017machine,Cao2021FLTrustBF, Fang2020LocalMP,Mhamdi} and remove statistical outliers before using
them to update the global model. A general framework to support secure Byzantine-robust aggregations rules, such as above, is an interesting future direction.

\noindent\textbf{\scalebox{0.9}{$\Vl(\cdot)$} Structure.} If \scalebox{0.9}{$\Vl(\cdot)$} contains repeated structures, the $G$-gate technique~\cite{Boneh2019} can improve efficiency. 

\noindent\textbf{Complex Aggregation Rules.} 
\name~can be used for more complex aggregation rules, such as mode,  by extending SNIP with affine-aggregatable encodings (AFE)~\cite{corrigan2017prio}.

\noindent\textbf{Differential Privacy.} The privacy guarantees of \name~can be enhanced by using differential privacy (DP) to reveal a \textit{noisy} aggregate using techniques such as~\cite{Kairouz2021TheDD}. Adding DP would also provide additional robustness guarantees~\cite{sun2019backdoor,naseri2021local}. 

\noindent\textbf{Scaling \name.} Our experimental results in Sec. \ref{sec:evaluation} show that \name~has reasonable performance for clients sizes up to $250$. One way of scaling \name~for larger client sizes can be by dividing the clients into smaller subsets of size $\sim250$ and then running \name~for each of these subsets \cite{scale-FL}.

\noindent\textbf{Towards poly-logarithmic complexity.} Currently,  dominant term in the complexity is $O(mnd)$ which results in a $O(n^2)$ dependence on $n$ (since we consider $m$ is a fraction of $n$). This can be reduced to $O(n\log^2 nd)$ by using the techniques from \cite{Bell2020}.  Specifically, instead of having each client verify the proofs of all others (complete graph for the verification) we can follow the exact construction of the $k$-regular graph $G$ from \cite{Bell2020} such that $k=O(\log n$) and only neighbors in $G$ act as verifiers for each other. The exact steps are as follows:  \begin{enumerate} \item Each client $C_i$ generates $n$ shares. It sends the corresponding shares to its $k$-neighbors (according to graph G from [9])  and the verification can be done on them as described currently in \name. Note that these shares follow a $t$-out-of-$n$ scheme where $t<k$.
\item $C_i$ encrypts the shares for the non-neighbors using a threshold (denoted by $t_{enc}$) fully homomorphic encryption scheme such as BGV~\cite{BGV} (the threshold variant can be obtained using work such as \cite{Thresh1,Thresh2}). Note that this threshold, $t_{enc}>m$ is different from that of the secret shares. This encryption is necessary for ensuring data privacy since for the threshold of the shares we could have $t<m$.

\item For the aggregation step, first the clients check the validity of the shares of its non-neighbors  (this can be done via homomorphic multiplications as shown by Feldman~\cite{Feldman87}). Next, only the shares corresponding to the clients that (i) pass the first step of input verification,  and (ii) have valid shares, are aggregated. Note that the shares corresponding to neighbors (for verification) can be encrypted using the public key of the encryption scheme for this step.

\item Each client now has the ciphertext of its share of the aggregate (corresponding to $(i,U_{i})$ where $U_{i}=\sum_{C_j\in C\setminus C*}u_{ji}$ in the current EIFFeL protocol) which is sent to the server. The server performs the reconstruction directly on these ciphertexts (using their homomorphic property) and obtains the ciphertext of the final aggregate. This can then be decrypted with the help of the clients to obtain the final aggregate in the clear.
\end{enumerate}

\end{document}